\documentclass[final]{jfp-epi}
\usepackage{jfp-cup2epi}

\usepackage{tikz}
\usepackage{stmaryrd}
\usepackage{listings}
\usepackage{tensor}
\usepackage{prooftree}
\usepackage[vskip=3pt]{quoting}

\newcommand{\refToRule}[1]{\textsc{\small (#1)}}
\newcommand{\refItem}[2]{\cref{#1}(\ref{#1:#2})}
\newcommand{\meta}[1]{\colorbox{lightgray}{$#1$}}
\newcommand{\Space}{\hskip 1em}
\newcommand{\BigSpace}{\hskip 2em}
\newcommand{\emath}[1]{\ensuremath{#1}} 

\newcommand{\ple}[1]{\emath{\langle #1 \rangle}}
\newcommand{\Pair}[2]     {\ple{{#1},{#2}}}
\newcommand{\Triple}[3]     {\ple{{#1}{,}{#2}{,}{#3}}}
\newcommand{\fv}[1]{\aux{fv}(#1)}
\newcommand{\dom}[1]{\aux{dom}(#1)}
\newcommand{\NamedRule}[4]{\scriptstyle{\textsc{(#1)}}\
\displaystyle                  
\frac{#2}{#3}         
\begin{array}{l}
#4     
\end{array}
}
\newcommand{\NamedCoRule}[4]{\scriptstyle{\textsc{(#1)}}\
\displaystyle                  
\genfrac{}{}{1.5pt}{0}{#2}{#3}         
\begin{array}{l}
#4     
\end{array}
}
\newcommand{\fun}[3]{\ensuremath{#1 : #2 \rightarrow #3}}

\newcommand{\RR}{\mathit{R}} 
\newcommand{\RSet}{\RC\RR}  
\newcommand{\RC}[1]{\emath{|#1|}} 
\newcommand{\rgr}{r}
\newcommand{\sgr}{s}
\newcommand{\tgr}{t} 
\newcommand{\rsum}{+}
\newcommand{\rmul}{{\cdot}}
\newcommand{\rzero}{\mathbf{0}}
\newcommand{\rone}{\mathbf{1}} 
\newcommand{\rord}{\preceq}

\newcommand{\RS}{\emath{\mathit{S}}}
\newcommand{\SSet}{\RC\RS} 
\newcommand{\N}{\mathbb{N}} 
\newcommand{\NN}{\emath{\mathsf{Nat}}}
\newcommand{\Triv}{\emath{\mathsf{Triv}}} 
\newcommand{\RRPos}{\mathsf{R}_{\ge 0}^\infty}
\newcommand{\RPos}{[0,\infty]}
\newcommand{\LL}{\mathsf{L}} 
\newcommand{\Extend}[1]{#1^{\infty}} 
\newcommand{\Int}[1]{\mathsf{Int}(#1)} 
\newcommand{\smprod}{\owedge} 

\newenvironment{grammatica}{$\begin{array}{lcll}}{\end{array}$}
\newcommand{\produzione}[3]{#1&::=&#2&\mbox{#3}}
\newcommand{\seguitoproduzione}[2]{&\mid&#1&\mbox{#2}}
\newcommand{\terminale}[1]{\texttt{#1}}

\newcommand{\aux}[1]{\mathsf{#1}}
\newcommand{\mvar}[1]{\mathit{#1}}
\newcommand{\produzioneinline}[2]{#1::=#2}

\newcommand{\e}{\mvar{e}}
\newcommand{\x}{\mvar{x}}
\newcommand{\y}{\mvar{y}} 
\newcommand{\z}{\mvar{z}}
\newcommand{\f}{\mvar{f}}
\newcommand{\unit}{\terminale{unit}}
\newcommand{\App}[2]{#1\, {\!#2}}
\newcommand{\Fun}[2]{\lambda #1.#2}
\newcommand{\PairExp}[4]{\langle^{#1}#2,#3\tensor*[^{#4}]{\rangle}{}}
\newcommand{\Inl}[2]{\terminale{inl}^{#1}#2}
\newcommand{\Inr}[2]{\terminale{inr}^{#1}#2}
\newcommand{\Case}[5]{\terminale{match}\ #1\ \terminale{with} \ \terminale{inl}\, #2 \rightarrow #3 \ \terminale{or} \ \terminale{inr}\, #4 \rightarrow #5}
\newcommand{\Match}[4]{\terminale{match}\  #3 \ \terminale{with}\ \langle#1,#2\rangle\rightarrow #4}
\newcommand{\RecFun}[3]{\terminale{rec}\,#1.\Fun{#2}{#3}}
\newcommand{\MatchUnit}[2]{\terminale{match}\ #1\ \terminale{with} \ \unit \ \rightarrow #2}
\newcommand{\Seq}[2]{#1\terminale{;}#2}

\newcommand{\Return}[1]{\terminale{return}\ #1} 
\newcommand{\Let}[3]{\terminale{let}\ #1 = #2\ \terminale{in}\ #3} 

\newcommand{\private}{\aux{private}}
\newcommand{\public}{\aux{public}}
\newcommand{\mutable}{\aux{mutable}}
\newcommand{\readonly}{\aux{readonly}}
\newcommand{\priv}{\aux{priv}}
\newcommand{\pub}{\aux{pub}}
\newcommand{\xtt}{\texttt{x}}
\newcommand{\ytt}{\texttt{y}}
\newcommand{\ptt}{\texttt{p}}
\newcommand{\un}{\terminale{u}}
\newcommand{\diverge}{\aux{div}}

\newcommand{\val}{\mathbf{v}}
\newcommand{\cctx}{\gamma}
\newcommand{\redval}[5]{\reduce{\ConfP{#1}{#2}}{#3}{\ConfP{#4}{#5}}}
\newcommand{\redvalS}[5]{\reduce{\ConfP{#1}{#2}}{#3}{\ConfP{#4}{#5}}}
\newcommand{\ve}{\mvar{v}}
\newcommand{\reduce}[3]{#1\Rightarrow_{#2}#3}
\newcommand{\reduceNarrow}[3]{#1{\Rightarrow_{#2}}#3}
\newcommand{\Conf}[2]{#1\mid#2}
\newcommand{\ConfP}[2]{#1{\mid}#2}
\newcommand{\env}{\rho}
\newcommand{\Subst}[3]   {#1[#2/#3]}
\newcommand{\AddToEnv}[4]{#1,\EnvElem{#2}{#3}{#4}}
\newcommand{\EnvElem}[3]{#1:\Pair{#2}{#3}}

\newcommand{\Unit}{\terminale{Unit}}
\newcommand{\Un}{\terminale{U}}
\newcommand{\funType}[3]{#1 \rightarrow_{#2} #3}
\newcommand{\PairT}[2]{#1\otimes #2}
\newcommand{\SumT}[2]{#1 + #2}

\newcommand{\T}{\mvar{T}} 
\newcommand{\ST}{\mvar{S}}
\newcommand{\Graded}[2]{{#1}^{#2}}
\newcommand{\GradedInd}[3]{#1_#2^{#3_#2}}

\newcommand{\ctxord}{\preceq}
\newcommand{\ctxsum}{+}
\newcommand{\ctxmul}{\rmul}
\newcommand{\VarType}[2]{#1:#2}
\newcommand{\VarGradeType}[3]{#1 :_{#2} #3}
\newcommand{\VarGrade}[2]{#1:#2}
\newcommand{\IsWFExp}[3]{#1\vdash#2:#3}

\newcommand{\IsWFEnv}[3]{#1\vdash#2 \triangleright #3}
\newcommand{\IsWFConf}[4]{#1\vdash\Conf{#2}{#3}:#4}

\newcommand{\VariantT}[2]{#1{:}#2}
\newcommand{\Variant}[3]{#1^{#2} #3}
\newcommand{\lab}{\ell}
\newcommand{\APairExp}[4]{[^{#1}#2,#3\tensor*[^{#4}]{]}{}}
\newcommand{\Proj}[2]{\pi_{#1} #2}
\newcommand{\Fst}[1]{\Proj{1}{#1}} 
\newcommand{\Snd}[1]{\Proj{2}{#1}} 
\newcommand{\APairT}[2]{#1\times #2}
\newcommand{\If}[3]{\terminale{if}\ #1\ \terminale{then}\ #2\ \terminale{else}\ #3}

\newcommand{\True}{\terminale{true}}

\newcommand{\even}{\terminale{even}}
\newcommand{\ev}{\terminale{ev}\,}
\newcommand{\tru}{\terminale{t}}
\newcommand{\bvar}{\mvar{b}}
\newcommand{\n}{\terminale{n}}
\newcommand{\nottt}{\terminale{not}}
\newcommand{\Succ}[2]{\App{\terminale{s}^{#1}}{#2}}
\newcommand{\zero}{\terminale{z}}
\newcommand{\CaseNat}[4]{\terminale{if-z(}#1, #2, \Succ{}{#3} {\rightarrow} #4\terminale{)}}
\newcommand{\comb}[2]{#2{.}#1}

\newcommand{\bv}{\mathcal{V}}
\newcommand{\bvPair}[2]{#2^{#1}}
\newcommand{\NW}{\aux{nw}}
\newcommand{\NoWaste}[2]{{#1}\ctxord^\star_{{\tiny\NW}}{#2}}
\newcommand{\NoNoWaste}[2]{#1\not\ctxord^\star_{{\tiny\NW}}#2}
\newcommand{\NWConf}[2]{\vdash_{{\tiny\NW}}\Conf{#1}{#2}}
\newcommand{\NWreduce}[3]{#1\Rrightarrow_{#2}#3}
\newcommand{\red}[5]{\NWreduce{\Conf{#1}{#2}}{#3}{\Conf{#4}{#5}}}
\newcommand{\envSum}{+}
\newcommand{\OKSum}[2]{#1\|#2}
\newcommand{\ZeroOut}[2]{#1{\subseteq}#2}
\newcommand{\Disjoint}[2]{#1{\|}#2}
\newcommand*{\added}[1][]{\cctx^{a}_{#1}}
\newcommand*{\used}[1][]{\cctx^{c}_{#1}}
\newcommand{\E}{\mvar{E}}
\newcommand{\Erase}[1]{\aux{erase}(#1)}

\newcommand{\Vars}{\aux{V}} 
\newcommand{\CCTX}[1]{\cctx_{#1}}
\newcommand{\graph}{\aux{G}}
\newcommand{\Nodes}[1]{V(#1)} 
\newcommand{\Graph}[1]{\graph_{#1}}
\newcommand{\Id}{\mathbb{I}}
\newcommand{\GraphStar}[1]{\graph^\star_{#1}}
\newcommand{\Start}[1]{\aux{in}(#1)}
\newcommand{\End}[1]{\aux{out}(#1)}

\newcommand{\weight}[2]{#1(#2)}
\newcommand{\p}{\mvar{p}}
\newcommand{\Paths}[3][]{\mathsf{P}\ifblank{#1}{}{_{\!\!#1}}(#2,#3)} 
\newcommand{\Pathsnz}[3]{\mathsf{P}_{\!\!\!{\ne}\rzero}^{#1}(#2,#3)}
\newcommand{\Gstar}[1]{{#1}^{\star}}
\newcommand{\reacharrowOne}{\rightarrow}
\newcommand{\reacharrow}{\xrightarrow\star}
\newcommand{\reachOne}[3]{#1\reacharrowOne_#2#3}
\newcommand{\reach}[3]{#1\reacharrow_#2#3}
\newcommand{\notReach}[3]{#1\not\reacharrow_#2#3}
\newcommand{\funEx}{\terminale{f}}
\newcommand{\Row}[2]{#1_#2}
\newcommand{\supp}[1]{\mathsf{S}_{#1}} 
\newcommand{\mtx}{M} 

\newcommand{\conf}{\mvar{c}}
\newcommand{\res}{\textsc{r}}
\newcommand{\ind}[3]{\vdash\NWreduce{#1}{#2}{#3}}
\newcommand{\gen}[3]{\vdash_\infty\NWreduce{#1}{#2}{#3}}

\newcommand{\WT}{\mathit{WT}}
\newcommand{\WTInfty}{\WT^\infty}
\newcommand{\IsWFconf}[3]{#1\vdash#2:#3}
\newcommand{\CSet}{\mvar{C}}
\newcommand{\erase}[1]{\aux{erase}(#1)}

\newcommand{\envE}{\tilde{\env}}
\newcommand{\NWreduceE}[2]{#1\Rrightarrow#2}

\newcommand{\length}[1]{{\mid}#1{\mid}}

\usepackage{enumitem}

\usepackage{aliascnt}
\usepackage[capitalize]{cleveref}

\theoremstyle{plain}
  \newtheorem{theorem}{Theorem}[section]

  \newaliascnt{lemma}{theorem}
\newtheorem{lemma}[lemma]{Lemma}
\aliascntresetthe{lemma}
\crefname{lemma}{Lemma}{Lemmas}

\newaliascnt{proposition}{theorem}
\newtheorem{proposition}[proposition]{Proposition}
\aliascntresetthe{proposition}
\crefname{proposition}{Proposition}{Propositions}

\newaliascnt{corollary}{theorem}
\newtheorem{corollary}[corollary]{Corollary}
\aliascntresetthe{corollary}
\crefname{corollary}{Corollary}{Corollaries}

\theoremstyle{definition}
\newaliascnt{definition}{theorem}
\newtheorem{definition}[definition]{Definition}
\aliascntresetthe{definition}
\crefname{definition}{Definition}{Definitions}

\newaliascnt{example}{theorem}
\newtheorem{example}[example]{Example}
\aliascntresetthe{example}
\crefname{example}{Example}{Examples}

\newaliascnt{remark}{theorem}
\newtheorem{remark}[remark]{Remark}
\aliascntresetthe{remark}
\crefname{remark}{Remark}{Remarks}
\renewcommand{\cite}{\citet}

\jfpVolume{36}
\jfpArticle{10}
\jfpDOI{10.46298/jfp.17800}
\jfpYear{2026}
\received[Submitted]{June 2025}
\received[accepted]{June 2026}

\begin{document} 
%
%
%

\title{Don't exhaust, don't waste}
\subtitle{Resource-aware soundness for big-step semantics}

\author{Riccardo Bianchini}
\orcid{0000-0003-0491-7652}
\affiliation{%
  \institution{University of Genoa}
  \city{}
  \country{Italy}
  \authoremail{riccardo.bianchini@edu.unige.it}
}

\author{Francesco Dagnino}
\orcid{0000-0003-3599-3535}
\affiliation{%
  \institution{University of Genoa}
  \city{}
  \country{Italy}
  \authoremail{francesco.dagnino@dibris.unige.it}
}

\author{Paola Giannini}
\orcid{0000-0003-2239-9529}
\affiliation{%
  \institution{University of Eastern Piedmont}
  \city{}
  \country{Italy}
  \authoremail{paola.giannini@uniupo.it}
}

\author{Elena Zucca}
\orcid{0000-0002-6833-6470}
\affiliation{%
  \institution{University of Genoa}
  \city{}
  \country{Italy}
  \authoremail{paola.giannini@uniupo.it}
}

\begin{abstract}
We extend the semantics and type system of a lambda calculus equipped with common constructs to be \emph{resource-aware}. That is, the semantics keeps track of the usage of resources, and is stuck, besides in case of type errors, if either a needed resource is exhausted, or a provided resource would be wasted.  In such way, the type system guarantees, besides standard soundness, that for well-typed programs there is a computation where no resource gets either exhausted or wasted. 

The extension is parametric on an arbitrary \emph{grade algebra}, modeling an assortment of possible usages, and does not require ad-hoc changes to the underlying language. To this end, the semantics needs to be formalized in big-step style; as a consequence, expressing and proving (resource-aware) soundness is challenging, and is achieved by applying recent techniques based on coinductive reasoning.  
\end{abstract}

\maketitle


\section{Introduction}\label{sect:introduction}
An increasing interest has been devoted in research to \emph{resource-awareness}, that is, to formal techniques for reasoning about how programs use external resources, such as files, locks, and memory.
The original inspiration for such techniques has been \emph{linear logic}, introduced by \cite{Girard87} in his seminal paper, where logical statements are considered resources which cannot be duplicated or discarded. The idea has been firstly carried to programming languages by \emph{substructural} type systems  \citep{Walker05}, particularly \emph{linear} and \emph{affine} ones. Such type systems enjoy an extended form of soundness which we call \emph{resource-aware soundness}; that is, they statically approximate not only the expected result of a program, but also how each external resource, handled through a name (free variable),  is used: either unused, or used exactly/at most once, in the linear/affine case, respectively, or in an unrestricted way.  Note that, differently from linear ones, affine (and unrestricted) resources are \emph{discardable}, meaning that there is no constraint that they should be actually used by the program.

More recently, the notion has been generalized to an arbitrary assortment of usages, formally modeled by \emph{grades}, elements of an algebraic structure (\emph{grade algebra)} defined in slightly different ways in  the  literature\footnote{Essentially variants of an ordered semiring.}  \citep{BrunelGMZ14,GhicaS14,McBride16,Atkey18,GaboardiKOBU16,AbelB20,OrchardLE19,WoodA22,DalLagoG22}. For instance, grades can be natural numbers counting how many times a resource is used, or even express non-quantitative properties, e.g., that a resource should be used with a given  access  level.  In this way, linear/affine type systems can be generalized to type systems parametric on the grade algebra, and the notion of resource-aware soundness can be generalized as well.

In order to formally express and prove that a type system enjoys resource-aware soundness,  the execution model needs to take into account resource consumption. This has been achieved in some recent proposals \citep{ChoudhuryEEW21,BianchiniDGZ@ECOOP23,TorczonSAAVW23}, including the conference version of this paper \citep{BianchiniDGZ@OOPSLA23}. 
The basic feature of  the  \emph{resource-aware semantics} is that substitution of variables is not performed once and for all, as in the standard $\beta$-rule. Instead, program execution happens in an \emph{environment} of available resources, and each variable occurrence is replaced when needed, by  consuming in some way the associated resource. In this way, reduction is stuck if \emph{resources are exhausted}, that is, the current availability is not enough to cover the usage required by the program, and this situation is expected to be prevented by the type system. 
Formally, ``enough'' means that $\rgr\rsum\sgr'\rord\sgr$, where $\rgr$ and $\sgr$ are the required and available amount, respectively, and $\rsum, \rord$ are the sum and order of the grade algebra. Note that the amount $\sgr'$ is left available. 

In this way, the program cannot consume more resources than those available; however, the converse property does not hold, that is, resources can be \emph{wasted}. 
 For instance, there is no guarantee that a linear resource is actually used. 

The advantages of constraints which impose no-waste, such as linear type systems, have been firstly 
 illustrated by \cite{Wadler90} in his seminal paper. They include avoiding garbage collection or reference counting, ensuring the correct use of files\footnote{See the example in \cref{sect:examples}.  }, and, in session type systems, prevent a receiver  from  indefinitely wait on a channel. However, in the existing proposals mentioned above \citep{ChoudhuryEEW21,BianchiniDGZ@ECOOP23,TorczonSAAVW23,BianchiniDGZ@OOPSLA23},  there is no formal way to express that the type system ensures no-waste as well.
  

In this paper, we provide,  to the best of  our knowledge,  the first statement and proof of resource-aware soundness where well-typedness implies that \emph{all} the following situations (which  are  stuck in the semantics) are prevented:
\begin{enumerate}   \label{intr}
\item   \label{intr1} type errors, as in the standard soundness  
\item \label{intr2}  resource exhaustion, as in the resource-aware soundness in previous work 
\item  \label{intr3}  resource wasting (novelty of this work)  
\end{enumerate}
We share with the conference paper \citep{BianchiniDGZ@OOPSLA23} the aim of  adding  resource-awareness in a light, abstract and general way, without requiring ad-hoc changes to the underlying language.  To achieve this aim, we also share  the following features:
\begin{itemize}
\item The reference language is an extended lambda calculus, intended to be representative of typical language features.  A contribution of our work is that, differently from most graded operational models, we consider a call-by-value semantics. Moreover, we provide an  in-depth  investigation of \emph{resource consumption in recursive functions}, and our graded type system smoothly includes \emph{equi-recursive types}.  
\item We keep the original syntax, with \emph{no boxing/unboxing} constructs; to this end, there is no explicit promotion rule, but different grades can be assigned to an expression, assuming different contexts.
\item The resource-aware semantics is given in \emph{big-step} style, so that no annotations in subterms are needed, again  keeping  the underlying language unaffected. 
 A consequence of this choice is that proving and even expressing (resource-aware) soundness of the type system becomes challenging, since in big-step semantics non-terminating and stuck computations are indistinguishable \citep{CousotC92,LeroyG09}. To solve this problem, the big-step judgment is extended to model divergence 
  explicitly,  and  is  defined by a \emph{generalized inference system}  \citep{AnconaDZ@esop17,Dagnino19}, where rules are interpreted in an \emph{essentially coinductive}, rather than inductive, way, in the sense that infinite proof trees are allowed, in a controlled way.  
Our proof of resource-aware soundness is a significant application of  these  innovative techniques. 
\end{itemize}

On the other hand, as mentioned above, differently from the conference paper \citep{BianchiniDGZ@OOPSLA23}, we state and prove resource-aware soundness \emph{including no-waste}, thanks to the following novel features:
\begin{itemize}
\item Formal definition of no-waste (transitive) usage of an environment of resources, in terms of the associated graph of direct dependencies.
\item No-waste refinement of the resource-aware semantics, meaning that a computation is stuck, in addition to the situations (\ref{intr1}) and (\ref{intr2}), when it would be wasting, i.e.,  (\ref{intr3}).
\item Formal characterization of configurations which are no-waste, in the sense that they do not get stuck due to wasting errors.
\item Proof that well-typed configurations are no-waste.
\end{itemize}

 We provide first, in \cref{sect:informal}, a gentle introduction to resource-aware semantics by examples. Then,  we formally define grade algebras in \cref{sect:algebraic}.
To make the presentation simpler, and modular with respect to the conference version, in \cref{sect:calculus} we present resource-aware reduction as in  the conference paper \citep{BianchiniDGZ@OOPSLA23}, focusing on its general notions. In \cref{sect:typesystem} we provide the type system, and in \cref{sect:examples} more significant examples and discussions.  In \cref{sect:nw-sem} we present the refined resource-aware reduction which ensures no-waste. In \cref{sect:soundness} we prove resource-aware soundness. 
\cref{sect:related} surveys related work, and \cref{sect:conclu} summarizes the contributions and outlines future work. Omitted proofs can be found in the Appendix.


\section{Preliminary examples}\label{sect:informal}
In order to illustrate the key ideas of resource-aware evaluation, we illustrate its expected behaviour on some simple examples.

We  assume a minimal syntax of expressions:
\begin{quoting}
\begin{grammatica}
\produzione{\e}{\x \mid \unit\mid\PairExp{\rgr}{\e_1}{\e_2}{\sgr}\mid\Match{\x}{\y}{\e_1}{\e_2} }{}
\end{grammatica}
\end{quoting}
only comprising variables in an infinite set $\Vars$, a constant, and (graded) pairs, with the corresponding destructor.
The metavariables $\rgr$ and $\sgr$ denote \emph{grades}, which the reader can initially think as natural numbers. Thus, the pair constructor is decorated with a grade for each subterm, intuitively meaning  ``how many copies'' are contained in the compound term. For instance, a pair of shape
$\PairExp{2}{\e_1}{\e_2}{2}$ contains ``two copies'' of each component, hence to evaluate (one copy of) such pair we need to obtain $2$ copies of the results of $\e_1$ and $\e_2$.

execution model which takes into account resource consumption. To achieve this, the basic feature \citep{ChoudhuryEEW21,BianchiniDGZ@ECOOP23,TorczonSAAVW23} is that substitution of variables is not performed once and for all, as in the standard $\beta$-rule. Instead, program execution happens in an \emph{environment} of available resources, and each variable occurrence is replaced when needed, by  consuming in some way the associated resource. Formally, the resource-aware semantics is defined on \emph{configurations}, that is, pairs $\Conf{\e}{\env}$ where the environment $\env$ is a finite map associating to each resource (variable), besides its value, a grade modeling its available usage, represented as shown below:\\

The goal is to provide an execution model that explicitly accounts for resource consumption. To this end, a key idea, following \citep{ChoudhuryEEW21,BianchiniDGZ@ECOOP23,TorczonSAAVW23}, is to avoid performing variable substitution once and for all, as in the standard $\beta$-reduction rule. Instead, program execution takes place within an \emph{environment} of available resources. Each occurrence of a variable is replaced only when needed, and doing so consumes the resource associated with that variable. Formally, the resource-aware semantics is defined over \emph{configurations}, which are pairs $\Conf{\e}{\env}$. Here, the environment $\env$ is a finite map that associates each resource, i.e., variable, with both its value and a grade representing the amount of usage still available, as illustrated below. 
\begin{quoting}
\begin{grammatica}
\produzione{\env}{\x_1:\Pair{\rgr_1}{\val_1},\ldots,\x_n:\Pair{\rgr_n}{\val_n}}{}\\
\produzione{\val}{\unit\mid\PairExp{\rgr}{\val_1}{\val_2}{\sgr}}{}
\end{grammatica}
\end{quoting}
The judgment has shape $\reduce{\Conf{\e}{\env}}{\rgr}{\Conf{\val}{\env'}}$, meaning that the configuration $\Conf{\e}{\env}$ produces a value $\val$ and a final environment $\env'$. The reduction relation is \emph{graded}, that is, indexed by a grade $\rgr$, meaning that the resulting value can be used (at most) $\rgr$ times.
The grade of a variable in the environment decreases, each time the variable is used, by the amount  specified in the reduction grade. Of course, this can only happen if the current grade of the variable \emph{can} be reduced  by  such an amount. Otherwise, evaluation is stuck; formally, since we choose to give the reduction relation in big-step style\footnote{The motivation for this choice is given in \cref{sect:calculus}.}, no judgment can be derived. Correspondingly, introducing a variable means adding a new resource in the environment. 

\begin{example}\label{ex:ex1}
Let us consider the following expressions:
\begin{itemize}
\item $\e_1=\Match{\xtt}{\ytt}{\ptt}{\PairExp{}{\unit}{\unit}{}}$
\item $\e_2=\Match{\xtt}{\ytt}{\ptt}{\PairExp{}{\xtt}{\unit}{}}$
\item $\e_3=\Match{\xtt}{\ytt}{\ptt}{\PairExp{}{\xtt}{\ytt}{}}$
\item $\e_4=\Match{\xtt}{\ytt}{\ptt}{\PairExp{}{\xtt}{\xtt}{}}$
\end{itemize}
to be evaluated in the environment $\env= \ptt:\Pair{1}{\val}$ with $\val=\PairExp{}{\val_1}{\val_2}{}$. Assume, as said above, that grades are natural numbers, as will be formally defined in \refItem{ex:gr-alg}{nat}.
 In order to lighten the notation, $1$ annotations are considered the default, hence omitted.  Moreover, in figures we abbreviate $\unit$ by $\un$.
In the first proof tree in \cref{fig:ex-red-nat} we show the evaluation of $\e_1$. The resource $\ptt$ is consumed, and its available amount ($1$) is ``transferred'' to both the resources $\xtt$ and $\ytt$, which are added  to  the environment\footnote{Modulo renaming, omitted here for simplicity.}, and not consumed.

\begin{figure}
\begin{small}
\begin{math}
\begin{array}{ll}
\prooftree
	\prooftree
	\justifies
	\reduce{\Conf{\ptt}{\env}}{}{\Conf{\val}{\ptt:\Pair{0}{\val}}}
	\thickness=0.08em
	\shiftright 2em
	\using \scriptstyle{\textsc{(var)}}
	 \endprooftree

	\BigSpace

	\prooftree
	\BigSpace\BigSpace\ldots
	\justifies
	\reduce{\Conf{\PairExp{}{\un}{\un}{}}{\env'}}{}{\Conf{\PairExp{}{\un}{\un}{}}{\env'}}
	\thickness=0.08em
	\shiftright 2em
	\using \scriptstyle{\textsc{(pair)}}
	 \endprooftree
\justifies
\reduce{\Conf{\Match{\xtt}{\ytt}{\ptt}{\PairExp{}{\un}{\un}{}}}{\env}}{}{\Conf{\PairExp{}{\un}{\un}{}}{\env'}}
\thickness=0.08em
\shiftright 2em
\using \scriptstyle{\textsc{(match-p)}}
 \endprooftree
\begin{array}{l}
\env{=}\EnvElem{\ptt}{1}{\val}\ \mbox{with}\ \val=\PairExp{}{\val_1}{\val_2}{}\\
\env_0{=}\EnvElem{\ptt}{0}{\val},\EnvElem{\ytt}{1}{\val_2}\\
\env'{=}\env_0,\EnvElem{\xtt}{1}{\val_1}\\
\env''{=}\env_0,\EnvElem{\xtt}{0}{\val_1}
\end{array}
\\[6ex]
\prooftree
	\prooftree
	\justifies
	\reduce{\Conf{\ptt}{\env}}{}{\Conf{\val}{\ptt:\Pair{0}{\val}}}
	\thickness=0.08em
	\shiftright 2em
	\using \scriptstyle{\textsc{(var)}}
	\endprooftree

	\BigSpace

	\prooftree
		\prooftree
		\justifies
		\reduce{\Conf{\xtt}{\env'}}{}{\Conf{\val_1}{\env''}}
		\thickness=0.08em
		\shiftright 2em
		\using \scriptstyle{\textsc{(var)}}
		\endprooftree

		\BigSpace

		\prooftree
		
		\justifies
		\reduce{\Conf{\xtt}{\env''}}{}{?}
		\thickness=0.08em
		\shiftright 2em
		\using \scriptstyle{\textsc{(var)}}
		\endprooftree

	\justifies
	\reduce{\Conf{\PairExp{}{\xtt}{\xtt}{}}{\env'}}{}{?}
	\thickness=0.08em
	\shiftright 2em
	\using \scriptstyle{\textsc{(pair)}}
	\endprooftree
\justifies
\reduce{\Conf{\Match{\xtt}{\ytt}{\ptt}{\PairExp{}{\xtt}{\xtt}{}}}{\env}}{}{?}
\thickness=0.08em
\shiftright 2em
\using \scriptstyle{\textsc{(match-p)}}
 \endprooftree
\end{array}
\end{math}
\end{small}
\caption{Examples of resource-aware evaluation (counting usages)}
\label{fig:ex-red-nat}
\end{figure}


 Evaluating  $\e_2$ and $\e_3$  works in much the same way, except that for $\e_2$  the resource $\xtt$  gets consumed during the process and the evaluation of $\e_3$ uses up all the available resources. 
Finally, the evaluation of $\e_4$ is stuck, that is, no proof tree can be constructed: indeed, when the second occurrence of $\xtt$ is found, the resource is exhausted, as shown in the second (incomplete) proof tree in \cref{fig:ex-red-nat}.
A result could be obtained, instead, if the original grade of $\ptt$ was greater than $1$, e.g., $2$, since in this case $\xtt$ (and $\ytt$) would be added  with grade $2$, or, alternatively, if the value associated to $\ptt$ in the environment was, e.g., $\val=\PairExp{2}{\val_1}{\val_2}{}$. Indeed, we expect a component of a pair $\PairExp{\rgr}{\val_1}{\val_2}{\sgr}$ to be extracted with a grade obtained by multiplying the grade of the pair with the grade of the component inside the pair, in this example $2\rmul 1$ and $1\rmul 2$, respectively.
\end{example}

Until now, we have assumed grades to be natural numbers, modeling \emph{how many times} resources are used. 
However,  a different choice of grades can model a \emph{non-quantitative}  usage,  that is, track possible \emph{modes} in which a resource can be used. A  very simple example are grades expressing  access  levels; for instance, $\private$ and $\public$, besides the grade $0$ which always expresses no usage at all. 
Such grades are ordered as follows: $0\rord\private\rord\public$, meaning that a program which does not use a variable can be safely run in an environment where this variable is available as $\private$, and a program which uses a variable as $\private$ can be safely run in an environment where it is $\public$. Of course, the converse does not hold.
Moreover, multiplication is the meet, meaning that we obtain  an access level  which is more restrictive than both.

 Note that exactly the same structure could be used to model,  e.g., modifiers $\readonly$ and $\mutable$ in an imperative setting, rather than  access  levels.  Moreover, the structure can be generalized by adding a $0$ element to any distributive lattice, as will be formalized in \refItem{ex:gr-alg}{lattice}, where we will also discuss the difference with  the  models of privacy levels in  the  literature \citep{AbelB20}.

 \begin{example}\label{ex:ex2}
 In \cref{ex:ex1},  writing $\priv$ and $\pub$ (default, omitted in the annotations) for short, we  have, e.g., for $\env= \ptt:\Pair{\pub}{\val}$ with $\val=\PairExp{\priv}{\val_1}{\val_2}{}$, that the evaluation in mode $\pub$ of $\e_1$ is analogous to that in \cref{fig:ex-red-nat}; however, the evaluation in mode $\pub$ of $\e_2$, $\e_3$, and $\e_4$ is stuck, since it needs to use the resource $\xtt$, which gets a grade $\priv=\priv\rmul\pub$, hence cannot be used in mode $\pub$ since $\pub\not\rord\priv$, as we  show in the first (incomplete) proof tree for $\e_2$ in \cref{fig:ex-red-privacy}. 
  \begin{figure}
 \begin{small}
\begin{math}
\begin{array}{l}
\begin{array}{ll}
    \prooftree
    \prooftree
\justifies
\reduce{\Conf{\ptt}{\env}}{}{\Conf{\val}{\env}}
 \thickness=0.08em
      \shiftright 2em
      \endprooftree \BigSpace

    \prooftree
    \prooftree
\justifies
\reduce{\Conf{\xtt}{\env'}}{}{?}
 \thickness=0.08em
      \shiftright 2em
      \endprooftree
\justifies
\reduce{\Conf{\PairExp{}{\xtt}{\un}{}}{\env'}}{}{?}
 \thickness=0.08em
      \shiftright 2em
      \endprooftree
\justifies
\reduce{\Conf{\Match{\xtt}{\ytt}{\ptt}{\PairExp{}{\xtt}{\un}{}}}{\env}}{}{?}
 \thickness=0.08em
      \shiftright 2em
      \endprooftree
      \end{array}
      \begin{array}{l}
\env= \EnvElem{\ptt}{\pub}{\val}\ \mbox{with}\ \val=\PairExp{\priv}{\val_1}{\val_2}{}\\
\env'=\EnvElem{\ptt}{\pub}{\val},\EnvElem{\xtt}{\priv}{\val_1},\EnvElem{\ytt}{\pub}{\val_2}
\end{array}
\\[10ex]
   \prooftree
    \prooftree
\justifies
\reduce{\Conf{\ptt}{\env}}{}{\Conf{\val}{\ptt{:}\Pair{\pub}{\val}}}
 \thickness=0.08em
      \shiftright 2em
      \endprooftree \BigSpace

    \prooftree
    \prooftree
\justifies
\reduce{\Conf{\xtt}{\env'}}{\priv}{\Conf{\val_1}{\env'}}
 \thickness=0.08em
      \shiftright 2em
      \endprooftree \BigSpace
    \prooftree
\justifies
\reduce{\Conf{\ytt}{\env'}}{\priv}{\Conf{\val_2}{\env'}}
 \thickness=0.08em
      \shiftright 2em
      \endprooftree
\justifies
\reduce{\Conf{\PairExp{}{\xtt}{\ytt}{}}{\env'}}{\priv}{\Conf{\PairExp{}{\val_1}{\val_2}{}}{\env'}}
 \thickness=0.08em
      \shiftright 2em
      \endprooftree
\justifies
\reduce{\Conf{\Match{\xtt}{\ytt}{\ptt}{\PairExp{}{\xtt}{\ytt}{}}}{\env}}{\priv}{\Conf{\PairExp{}{\val_1}{\val_2}{}}{\env'}}
 \thickness=0.08em
      \shiftright 2em
      \endprooftree
\end{array}
\end{math}
\end{small}
\caption{Examples of resource-aware evaluation ( access  levels)}
\label{fig:ex-red-privacy}
\end{figure}
On the other hand, evaluation in mode $\priv$ can be safely performed; indeed, resource $\ptt$ can be used in mode $\priv$ since $\priv\rord\pub$, as shown in the second proof tree in \cref{fig:ex-red-privacy}. 
\end{example}

In both the previous examples, reduction  gets  stuck when \emph{a needed resource is exhausted}, as will be formalized in \cref{sect:calculus}. 
In \cref{sect:nw-sem},  we 
 introduce a refined notion of reduction that 
also  gets  stuck when \emph{a resource is wasted}, that is, 
 when some resource that \emph{should have been used} is left over after program execution. 

 This cannot happen with the choices of grades considered  until now.  That is because  the grade $0$, representing no usage, is less or equal than all  grades, that is, can be overapproximated by any  amount  of usage.
In other words, there are no grades which impose that a resource \emph{should be used}. 
An example where, instead, there is such  a  constraint, are natural numbers where the considered order is equality, as formalized in \refItem{ex:gr-alg}{nat}. This means that a resource with grade, e.g., $1$, should be used \emph{exactly once}: neither more than once, nor unused. Formally, $0$ is \emph{not} less than $1$.  The grades  which impose a lower bound on usage are called \emph{non-discardable}, see \cref{def:discard} in the following. 

Let us consider what would happen in \cref{ex:ex1} with  these grades.  We expect that only $\e_3$ can be successfully evaluated, since only in this case the resources are all exhausted; the evaluations of $\e_1$ and $\e_2$, instead, would waste both $\x$ and $\y$, and only $\y$, respectively, hence  they are  not allowed.  Instead  the evaluation of $\e_4$ gets stuck due to resource exhaustion exactly as before. 
The formal definition of a reduction which is no-waste, hence  gets stuck for   $\e_1$ and $\e_2$, requires some more elaborated definitions, which will be developed in \cref{sect:nw-sem}.


\section{Algebraic preliminaries: a taxonomy of grade algebras}
\label{sect:algebraic} 

In this section we introduce the algebraic structures we will use throughout the paper. 
At the core of our work there are \emph{grades}, namely, annotations in terms and types expressing how or how much resources can be used during the computation. 
As we will see, we need some operations and relations to properly combine and compare grades in  the resource-aware semantics and type system, hence we assume grades to form an algebraic structure called  a  \emph{grade algebra} defined below. 
Such a structure is a slight variant of others in  the  literature \citep{BrunelGMZ14,GhicaS14,McBride16,Atkey18,GaboardiKOBU16,AbelB20,OrchardLE19,ChoudhuryEEW21,WoodA22}, which are all instances of ordered semirings.

\begin{definition}[Ordered Semiring] \label{def:ord-sr}
An \emph{ordered semiring}  is a tuple $\RR = \ple{\RSet,\rord,\rsum,\rmul,\rzero,\rone}$ such that: 
\begin{itemize}
\item \ple{\RSet,\rord} is a partially ordered set; 
\item \ple{\RSet,\rsum,\rzero} is a commutative monoid; 
\item \ple{\RSet,\rmul,\rone} is a monoid; 
\end{itemize}
and the following axioms are satisfied: 
\begin{itemize}
\item $\rgr\rmul(\sgr\rsum\tgr) = \rgr\rmul \sgr\rsum\rgr\rmul\tgr$ and $(\sgr\rsum\tgr)\rmul\rgr = \sgr\rmul\rgr\rsum\tgr\rmul\rgr$, for all $\rgr,\sgr,\tgr\in\RSet$; 
\item $\rgr\rmul\rzero = \rzero$ and $\rzero\rmul\rgr = \rzero$, for all $\rgr\in\RSet$; 
\item if $\rgr\rord\rgr'$ and $\sgr\rord\sgr'$ then  $\rgr\rsum\sgr\rord\rgr'\rsum\sgr'$ and $\rgr\rmul\sgr\rord\rgr'\rmul\sgr'$, for all $\rgr,\rgr'\sgr,\sgr'\in\RC\RR$; 
\end{itemize}
\end{definition}

Essentially, an ordered semiring is
a semiring  with a partial order on its carrier which makes addition and multiplication monotonic with respect to it. 
Roughly, addition and multiplication (which is not necessarily commutative) provide parallel and sequential composition of usages, 
$\rone$ models the unitary or default usage and 
$\rzero$ models  no use. 
Finally, the partial order models overapproximation in the resource usage,  giving  flexibility. For instance we can have  different usage in the branches of an if-then-else construct. 
  
In an ordered semiring  there can be elements $\rgr\rord\rzero$, which, however, make no sense in our context, since $\rzero$ models no usage. 
Hence, in a grade algebra we forbid such grades. 

\begin{definition}[Grade Algebra]\label{def:gr-alg}
An ordered semiring $\RR = \ple{\RC\RR, \rord, \rsum, \rmul, \rzero, \rone}$ is a \emph{grade algebra} if 
$\rgr\rord\rzero$ implies $\rgr = \rzero$, for all $\rgr\in\RC\RR$. 
\end{definition}

This property is not technically needed, but  it is intuitively expected to hold, and allows  to simplify some definitions. 
Moreover, it can be forced in any ordered semiring, just noting that the set $I_{\rord\rzero} = \{ \rgr\in\RC\RR \mid \rgr\rord\rzero \}$ is a two-sided ideal and so the quotient semiring $\RR/I_{\rord\rzero}$ is well-defined and is a grade algebra. 

We now give some examples of grade algebras adapted from the literature. 

\begin{example}\label{ex:gr-alg}\
\begin{enumerate}
\item\label{ex:gr-alg:nat}
The simplest way of measuring resource usage is by counting and can be done using natural numbers with their usual operations. 
We consider two grade algebras over natural numbers: 
one for \emph{bounded counting} $\NN^\le = \ple{\N,\le,+,\cdot,0,1}$, taking the natural ordering, and 
another for \emph{exact counting}   $\NN^= = \ple{\N,=,+,\cdot,0,1}$, taking equality as order, thus basically forbidding approximations of resource usage. 
\item\label{ex:gr-alg:lin}
The \emph{linearity} grade algebra $\ple{\{0,1,\infty\},\le,+,\cdot,0,1}\}$ is obtained from  $\NN^=$ above by identifying all natural numbers strictly greater than $1$ and taking as order $0\le\infty$ and $1\le\infty$; the \emph{affinity} grade algebra only differs for the order, which is $0\le 1 \le \infty$. 
\item\label{ex:gr-alg:triv}
In the trivial grade algebra $\Triv$ the carrier is a singleton set $\RC\Triv = \{\infty\}$, the partial order is the equality, addition and multiplication are defined in the trivial way and $\rzero_\Triv = \rone_\Triv = \infty$. 
 \item\label{ex:gr-alg:relas} 
The grade algebra of extended non-negative real numbers is 
$\RRPos{=}{\ple{\RPos,\le,+,\cdot,0,1}}$, where usual order and operations are extended  to $\infty$ in the expected way.
\item\label{ex:gr-alg:lattice} 
A distributive lattice $\LL = \ple{\RC\LL, \le, \lor,\land,\bot,\top}$, where  $\le$ is the order, $\lor$ and $\land$ the join and meet, and $\bot$ and $\top$ the bottom and top element, respectively, is a grade algebra. 
 These  grade algebras do not carry quantitative information, as the addition is idempotent, but rather express how/in which mode a resource can be used. They are called \emph{informational} by  \mbox{\cite{AbelB20}. } 
\item\label{ex:gr-alg:prod} 
Given the grade algebras $\RR = \ple{\RSet,\rord_\RR,\rsum_\RR,\rmul_\RR,\rzero_\RR,\rone_\RR}$ and $\RS = \ple{\SSet,\rord_\RS,\rsum_\RS,\rmul_\RS,\rzero_\RS,\rone_\RS}$,
the \emph{product} $\RR\times\RS = \ple{\{\Pair{\rgr}{\sgr}\ |\ \rgr\in\RSet\ \wedge\ \sgr\in\SSet\},\rord,\rsum,\rmul,\Pair{\rzero_\RR}{\rzero_\RS},\Pair{\rone_\RR}{\rone_\RS}}$, where operations are the pairwise application of the operations for $\RR$ and $\RS$, is a grade algebra. 
\item\label{ex:gr-alg:ext} 
Given a  grade algebra $\RR = \ple{\RSet,\rord_\RR,\rsum_\RR,\rmul_\RR,\rzero_\RR,\rone_\RR}$, 
set
$\Extend{\RR} = \ple{\RSet + \{\infty\},\rord,\rsum,\rmul,\rzero_\RR,\rone_\RR}$ where $\rord$ extends $\rord_\RR$ by adding
$\rgr\rord\infty$ for all $\rgr\in\RC{\Extend{\RR}}$ and $\rsum$ and $\rmul$ extend $\rsum_\RR$ and $\rmul_\RR$ by 
$\rgr\rsum\infty = \infty\rsum\rgr = \infty$, for all $\rgr\in\RC{\Extend\RR}$, and 
$\rgr\rmul\infty = \infty\rmul\rgr = \infty$, for all $\rgr\in\RC{\Extend\RR}$ with $\rgr\ne\rzero_\RR$, and 
$\rzero_\RR\rmul\infty = \infty\rmul\rzero_\RR = \rzero_\RR$. 
Then, $\Extend{\RR}$ is a grade algebra, where $\infty$ models \emph{unrestricted usage}. 
\item\label{ex:gr-alg:interval}
Given $\RR$ as above, set $\RC{\Int\RR} = \{ \ple{\rgr,\sgr} \in \RC\RR\times\RC\RR \mid \rgr\rord_\RR\sgr\}$, the set of \emph{intervals} between two points in $\RC\RR$, with  
$\ple{\rgr,\sgr}\rord\ple{\rgr',\sgr'}$ iff $\rgr'\rord_\RR\rgr$ and $\sgr \rord_\RR \sgr'$. 
Then, $\Int\RR = {\ple{\RC{\Int\RR},\rord,\rsum,\rmul,\ple{\rzero_\RR,\rzero_\RR},\ple{\rone_\RR,\rone_\RR}}}$ is a grade algebra, with $\rsum$ and $\rmul$ defined pointwise. 
\end{enumerate}
\end{example} 

We say that a grade algebra is \emph{trivial} if it is isomorphic to $\Triv$, that is, it contains a single point. 
 It is easy to see that a grade algebra is trivial iff $\rone\rord\rzero$, as this implies $\rgr\rord\rzero$, hence $\rgr = \rzero$, for all $\rgr\in\RC\RR$, by the axioms of ordered semiring and \cref{def:gr-alg}.

A grade algebra is still a quite wild structure, notably 
we can get $\rzero$ by summing or multiplying non-zero grades. This does not match the intuition that $\rzero$ represents no usage at all, hence a combination of non-zero usages cannot give no usage, as expressed by the following definition.

\begin{definition}\label{def:gr-alg-int-red}
Let $\RR = \ple{\RC\RR,\rord,\rsum,\rmul,\rzero,\rone}$ be a grade algebra. 
We say that 
\begin{itemize}
\item $\RR$ is \emph{integral} if $\rgr\rmul\sgr = \rzero$ implies $\rgr = \rzero$ or $\sgr = \rzero$, for all $\rgr,\sgr\in\RC\RR$; 
\item $\RR$ is \emph{reduced} if $\rgr\rsum\sgr = \rzero$ implies $\rgr = \sgr = \rzero$, for all $\rgr,\sgr\in\RC\RR$. 
\end{itemize}
\end{definition}

  The grade algebras considered as examples in the paper will be integral and reduced. Besides the modeling reasons mentioned above, the requirement to be integral  is technically useful, e.g., in \refItem{lem:promotion}{exp}.

All the grade algebras in \cref{ex:gr-alg} are reduced, provided that the parameters are reduced,  for (\ref{ex:gr-alg:prod})-(\ref{ex:gr-alg:interval}). 
Similarly, they are all integral except (\ref{ex:gr-alg:lattice}) and (\ref{ex:gr-alg:prod}). 
Indeed, in the former there can be elements different from $\bot$ whose meet is $\bot$ (e.g., disjoint subsets in the powerset lattice), 
while in the latter there are 
pairs \ple{\rgr,\rzero} and \ple{\rzero,\sgr} whose product is \ple{\rzero,\rzero} even if both $\rgr\ne\rzero$ and $\sgr\ne\rzero$. 
Fortunately, there is an easy construction making a grade algebra both reduced and integral. 
\begin{definition}\label{def:gr-alg-0}
Set $\RR = \ple{\RC\RR,\rord,\rsum,\rmul,\rzero,\rone}$ an ordered semiring. We denote by $\RR_\rzero$ the ordered semiring $\ple{\RC\RR + \{\hat\rzero\}, \rord, \rsum, \rmul, \hat\rzero, \rone}$ where we add a new point $\hat\rzero$ and extend order and operations as follows: 
\begin{quoting}
$\hat\rzero \rord\rgr$ iff $\rzero\rord\rgr$, \\
$\rgr\rsum\hat\rzero = \hat\rzero\rsum\rgr = \rgr$, \\
$\rgr\rmul\hat\rzero = \hat\rzero\rmul\rgr = \hat\rzero$, for all $\rgr \in \RC\RR + \{\hat\rzero\}$. 
\end{quoting}
\end{definition}
It is easy to check that the following proposition holds. 

\begin{proposition}\label{prop:gr-alg-0}
If $\RR = \ple{\RC\RR,\rord,\rsum,\rmul,\rzero,\rone}$ is a grade algebra, then 
$\RR_\rzero$ is a reduced and integral grade algebra. 
\end{proposition}

Applying this construction to (\ref{ex:gr-alg:lattice}) and (\ref{ex:gr-alg:prod}) we get reduced and integral grade algebras.

For the former, note that we get an informational grade algebra where the grade meaning the ``lowest'' mode of usage is different from the zero grade, meaning no usage at all. 
For instance, two  access  levels are modeled by a grade algebra where $0\rord\private\rord\public$. This is different from models of privacy levels in  the  literature \citep{AbelB20}, where $\rzero$ coincides with the lowest mode of usage. Indeed, in such models the meaning of the zero grade is ``irrelevant''. That is, the usage of the resource does not affect the result of the computation, hence we cannot infer any information about the resource by observing such result: in this sense the resource is $\private$.

For (\ref{ex:gr-alg:prod}), the result is not yet satisfactory. 
Indeed, the resulting grade algebra still has spurious elements which are difficult to interpret. 
Thus we consider the following refined  construction. 

\begin{example}\label{ex:gr-alg-smash}
Let $\RR$ and $\RS$ be non-trivial, reduced and integral grade algebras. 
The \emph{smash product} of $\RR$ and $\RS$ is 
$\RR \smprod \RS = \ple{\RC{\RR\smprod\RS},\rord,\rsum,\rmul,\hat\rzero,\ple{\rone_\RR,\rone_\RS}}$, where 
$\RC{\RR\smprod\RS} = \RC{(\RR\times\RS)_\rzero} \setminus(\RC\RR\times\{\rzero_\RS\} \cup \{\rzero_\RR\}\times\RC\RS)$, 
$\rord$, $\rsum$ and $\rmul$ are the restrictions of the order and operations of $(\RR\times\RS)_\rzero$, as in \cref{def:gr-alg-0}, to the subset $\RC{\RR\smprod\RS}$, and 
$\hat\rzero$ is the zero of $(\RR\times\RS)_\rzero$.
It is easy to see $\RR\smprod\RS$ is a non-trivial, reduced and integral grade algebra.
\end{example}

Finally, to express no-wasting, as will be done in \cref{sect:nw-sem}, we need the following definition. 

 \begin{definition}[Discardable grade and affine grade algebra]\label{def:discard}
 In a grade algebra $\RR$, a grade $\rgr$ such that $\rzero\rord\rgr$ is said \emph{discardable}.  
A grade algebra is said to be \emph{affine} if $\rzero\rord\rgr$ holds for all $\rgr\in\RC\RR$, that is, all grades are discardable.
\end{definition}

Intutively, resources with discardable grades can be safely wasted,  since they can be reduced to $\rzero$ through the approximation order. 
For instance, in the linearity grade algebra of \refItem{ex:gr-alg}{lin}, the element $\infty$ can be discarded, while the element $1$ cannot. 
Note that the affinity condition is equivalent to requiring  just $\rzero\rord\rone$ again thanks to the axioms of ordered semiring. 
Instances of affine grade algebras from \cref{ex:gr-alg} are 
bounded usage (\ref{ex:gr-alg:nat}),  the extended non-negative real numbers (\ref{ex:gr-alg:relas})  and distributive lattices (\ref{ex:gr-alg:lattice}), while 
exact usage (\ref{ex:gr-alg:nat}) and linearity (\ref{ex:gr-alg:lin}) are not affine.

The subset of discardable elements in a grade algebra\footnote{Actually in any ordered semiring.} 
forms a \emph{two-sided ideal}. 
This means that the following  two properties hold: 
 \begin{itemize}
\item if $\rzero\rord\rgr$ and $\rzero\rord\sgr$, then $\rzero\rord\rgr\rsum\sgr$,
\item if $\rzero\rord\rgr$, then $\rzero\rord\rgr\rmul\sgr$ and $\rzero\rord\sgr\rmul\rgr$, for every $\sgr\in\RSet$. 
\end{itemize}
These properties easily follow from the monotonicity of $\rsum$ and $\rmul$ and  from the fact that multiplying by $\rzero$ yields $\rzero$. 


\section{Resource-aware semantics}\label{sect:calculus}

We define, for a standard functional calculus, an instrumented semantics which keeps track of resource usage, hence, in particular, it gets stuck if some needed resource is  insufficient.  

 The syntax, shown in \cref{fig:fine-grained}, is given in a fine-grained style. 
\begin{figure}
\begin{grammatica}
\produzione{\ve}{\x \mid \RecFun{\f}{\x}{\e}\mid\unit\mid\PairExp{\rgr}{\ve_1}{\ve_2}{\sgr}\mid\Inl{\rgr}{\ve}\mid\Inr{\rgr}{\ve} 
}{value expression} \\ 
\produzione{\e}{ \Return\ve \mid \Let\x{\e_1}{\e_2}  \mid \App{\ve_1}{\ve_2}}{expression} \\ 
\seguitoproduzione{ \MatchUnit{\ve}{\e} }{}\\
\seguitoproduzione{ \Match{\x}{\y}{\ve}{\e} }{}\\
\seguitoproduzione{ \Case{\ve}{\x_1}{\e_1}{\x_2}{\e_2} }{}
\end{grammatica}
\caption{Fine-grained syntax}
\label{fig:fine-grained}
\end{figure}
 This long-standing approach \citep{LevyPT03} is used to clearly separate effect-free from effectful expressions (\emph{computations}), and to make the evaluation strategy,  relevant for the latter, explicit through the sequencing construct (\terminale{let-in}), rather than fixed a-priori. In our calculus, the computational effect is \emph{divergence}, so the effectful expressions  are possibily diverging, whereas those effect-free will be called \emph{value expressions}\footnote{They are often called just ``values'' in literature, though, as already noted in \citep{LevyPT03}, they are not values in the operational sense, that is, results of the evaluation; here we keep the two notions distinct. Values turn out to be value expressions with no free variables, except 
  those  under a lambda.}. Note that, as customary,  expressions are defined on top of value expressions, whereas the converse does not hold;  this stratification makes the technical development modular,  as will be detailed in the following.

The constructs are pretty standard: the $\unit$ constant, pairs, left and right injections, and three variants of \terminale{match} construct playing as destructors of units, pairs, and injections, respectively. 
 Instead  of standard lambda expressions and a \terminale{fix} operator for recursion, we have a  single  construct $\RecFun{\f}{\x}{\e}$, meaning a function with parameter $\x$ and body $\e$ which can recursively call itself through the variable $\f$.
Standard lambda expressions can be recovered as those where $\f$ does not occur free in $\e$, that is, when the function is non-recursive, and we will use the abbreviation $\Fun{\x}{\e}$ for such expressions. The motivation for such unique construct is that in the resource-aware semantics there is no immediate parallel substitution as in standard rules for application and \terminale{fix}, but occurrences of free variables are replaced  one at a time, when needed, by their value stored in an environment. Thus, application of a (possibly recursive) function can be nicely modeled by generalizing what  is  expected for a non-recursive one, that is, it leads to the evaluation of the body in an environment where both $\f$ and $\x$ are added as resources, as formalized in  rule \refToRule{app} in \cref{fig:semantics}.

 As anticipated, the  pair and injection constructors are decorated with a grade for each subterm, intuitively meaning  ``how many copies'' are contained in the compound term. 
For instance, taking as grades the natural numbers as in \refItem{ex:gr-alg}{nat}, a pair of shape
$\PairExp{2}{\e_1}{\e_2}{2}$ contains ``two copies'' of each component.
 In the resource-aware semantics, this is reflected by the fact that, to evaluate (one copy of) such pair, we need to obtain $2$ copies of the results of $\e_1$ and $\e_2$; correspondingly, when matching such result with a pair of variables, both are made available in the environment with grade $2$.

We will sometimes use, rather than $\MatchUnit{\e_1}{\e_2}$, the alternative syntax $\Seq{\e_1}{\e_2}$, emphasising  that there is a sequential evaluation of the two subterms.

The resource-aware semantics is formally defined in \cref{fig:semantics}.  Corresponding to the two syntactic categories, the semantics is expressed by two distinct judgments, $\redval{\ve}{\env}{\rgr}{\val}{\env'}$ in the top section, and $\reduce{\Conf{\e}{\env}}{\rgr}{\Conf{\val}{\env'}}$ in the bottom section.  Note that there is no mutual recursion: the latter is defined on top of the former.  Hence, the metarules in \cref{fig:semantics} can be equivalently seen as

\begin{itemize}
\item a single inductive definition
\item a stratified definition, where the judgment $\reduce{\Conf{\e}{\env}}{\rgr}{\Conf{\val}{\env'}}$ is inductively defined, assuming the definition of  $\redval{\ve}{\env}{\rgr}{\val}{\env'}$.
\end{itemize}
For simplicity, we use the same notation for the two judgments, and in the bottom section of \cref{fig:semantics} we write both judgments as premises, taking the first view. However, the stratified view will be useful later to allow a modular  technical  development, notably for the construction adding divergence given in \cref{sect:soundness}.

\begin{figure}
\begin{math}
\begin{array}{l}
\NamedRule{var}{ }
{ \redvalS{\x}{\AddToEnv{\env}{\x}{\sgr}{\val}}{\rgr}{\val}{\AddToEnv{\env}{\x}{\sgr'}{\val}} }
{  
\rgr\rsum \sgr' \rord \sgr }
\NamedRule{fun}{ }
{ \redvalS{\RecFun\f\x\e}{\env}{\rgr}{\RecFun\f\x\e}{\env} }
{}
\\[6ex]
\NamedRule{unit}{ }
{ \redval{\unit}{\env}{\rgr}{\unit}{\env} } 
{} 
\BigSpace
\NamedRule{pair}{
  \redval{\ve_1}{\env}{\rgr\rmul\rgr_1}{\val_1}{\env_1}
  \Space
  \redval{\ve_2}{\env_1}{\rgr\rmul\rgr_2}{\val_2}{\env_2} 
}{ \redval{\PairExp{\rgr_1}{\ve_1}{\ve_2}{\rgr_2}}{\env}{\rgr}{\PairExp{\rgr_1}{\val_1}{\val_2}{\rgr_2}}{\env_2} } 
{}
\\[4ex]
\NamedRule{inl}{
  \redval{\ve}{\env}{\rgr\rmul\sgr}{\val}{\env'} 
}{ \redval{\Inl{\sgr}{\ve}}{\env}{\rgr}{\Inl{\sgr}{\val}}{\env'} }
{} 
\BigSpace
\NamedRule{inl}{
  \redval{\ve}{\env}{\rgr\rmul\sgr}{\val}{\env'} 
}{ \redval{\Inr{\sgr}{\ve}}{\env}{\rgr}{\Inr{\sgr}{\val}}{\env'} }
{} 
\\[4ex]
\end{array}
\end{math}

\hrule

\begin{math}
\begin{array}{l}
\\
\NamedRule{ret}{\redval{\ve}{\env}{\sgr}{\val}{\env'} }
{ \reduce{\Conf{\Return\ve}{\env}}{\rgr}{\Conf{\val}{\env'}} }
{ 
\!\! \rgr\rord\sgr\ne\rzero 
} 
\ 
\NamedRule{let}{
  \begin{array}{l}
    \reduceNarrow{\ConfP{\e_1}{\env}}{\sgr}{\Conf{\val}{\env''}} \\ 
    \reduceNarrow{\ConfP{\Subst{\e_2}{\x'}{\x}}{\AddToEnv{\env''}{\x'}{\sgr}{\val}}}{\rgr}{\ConfP{\val'}{\env'}} 
  \end{array}
}{ \reduceNarrow{\ConfP{\Let{\x}{\e_1}{\e_2}}{\env}}{\rgr}{\ConfP{\val'}{\env'}} }
{\!\x'\ \mbox{fresh}} 
\\[4ex] 
\NamedRule{app}{
  \begin{array}{l}
    \redval{\ve_1}{\env}{\sgr}{\RecFun\f\x\e}{\env_1} \Space 
    \redval{\ve_2}{\env_1}{\tgr}{\val_2}{\env_2} \\
    \reduce{\Conf{\Subst{\Subst{\e}{\f'}{\f}}{\x'}{\x}}{\AddToEnv{\AddToEnv{\env_2}{\f'}{\sgr_2}{\RecFun{\f}{\x}{\e}}}{\x'}{\tgr}{\val_2}}}{\rgr}{\Conf{\val}{\env'}}
  \end{array}
}{ \reduce{\Conf{\App{\ve_1}{\ve_2}}{\env}}{\rgr}{\Conf{\val}{\env'}} }
{ \sgr_1 \rsum \sgr_2 \rord \sgr\\
  \sgr_1 \neq \rzero \\ 
  \f',\x'\ \mbox{fresh} 
}
\\[5ex]
\NamedRule{match-u}{\redval{\ve}{\env}{\sgr}{\unit}{\env''} \Space
  \reduce{\Conf{\e}{\env''}}{\rgr}{\Conf{\val}{\env'}} 
}{ \reduce{\Conf{\MatchUnit{\ve}{\e}}{\env}}{\rgr}{\Conf{\val}{\env}} } 
{  
  \sgr\ne\rzero 
}
\\[5ex]
\NamedRule{match-p}{
\begin{array}{l}
\redval{\ve}{\env}{\sgr}{\PairExp{\rgr_1}{\val_1}{\val_2}{\rgr_2}}{\env''} \\ 
  \reduce{\Conf{\Subst{\Subst{\e}{\x'}{\x}}{\y'}{\y}}{\AddToEnv{\AddToEnv{\env''}{\x'}{\sgr\rmul\rgr_1}{\val_1}}{\y'}{\sgr\rmul\rgr_2}{\val_2}}}{\rgr}{\Conf{\val}{\env'}}
  \end{array}
}{ \reduce{\Conf{\Match{\x}{\y}{\ve}{\e}}{\env}}{\rgr}{\Conf{\val}{\env'}} }
{ \sgr\ne\rzero \\ 
  \x',\y'\ \mbox{fresh}
}
\\[5ex]
\NamedRule{match-l}{
\begin{array}{l}
 \redval{\ve}{\env}{\tgr}{\Inl\sgr\val}{\env''} \\ 
 \reduce{\Conf{\Subst{\e_1}{\y}{\x_1}}{\AddToEnv{\env''}{\y}{\tgr\rmul\sgr}{\val}}}{\rgr}{\Conf{\val'}{\env'}}
  \end{array}
}{ \reduce{\Conf{\Case{\ve}{\x_1}{\e_1}{\x_2}{\e_2}}{\env}}{\rgr}{\Conf{\val'}{\env'}} } 
{   \tgr\ne\rzero \\ 
  \y\ \mbox{fresh} 
}
\\[5ex]
\NamedRule{match-r}{
\begin{array}{l}
  \redval{\ve}{\env}{\tgr}{\Inr\sgr\val}{\env''} \\ 
\reduce{\Conf{\Subst{\e_2}{\y}{\x_1}}{\AddToEnv{\env''}{\y}{\tgr\rmul\sgr}{\val}}}{\rgr}{\Conf{\val'}{\env'}}
  \end{array}
}{ \reduce{\Conf{\Case{\ve}{\x_1}{\e_1}{\x_2}{\e_2}}{\env}}{\rgr}{\Conf{\val'}{\env'}} } 
{   \tgr\ne\rzero \\ 
  \y\ \mbox{fresh} 
}
\end{array}
\end{math}
\caption{Resource-aware semantics}
\label{fig:semantics} 
\end{figure}
 Rules for value expressions just replace variables by values; the reduction cannot diverge, but is resource-consuming, hence it can get stuck.  
 
In particular, in \refToRule{var}, which is the key rule where resources are consumed, 
a variable is replaced by its associated value, 
and its grade $\sgr$ decreases to $\sgr'$, burning an amount  $\rgr$  of resource which has to be at least the reduction grade. 
The side condition  $\rgr\rsum\sgr'\rord\sgr$  ensures that the initial grade  is allowed to  consume the  $\rgr$  amount, \mbox{leaving a residual grade $\sgr'$.  }
Note that  such  a  condition could hold for different residual grades.  Hence, reduction is largely non-deterministic; it will be the responsibility of the type system to ensure that there is at least one reduction  which  does not get stuck.   

The other rules for value expressions propagate rule \refToRule{var} to subterms which are variables. In rules for ``data containers'' \refToRule{pair}, \refToRule{in-l}, and \refToRule{in-r}, the components are evaluated with the evaluation grade of the compound value expression, multiplied by that of the component. 

Whereas evaluation of value expressions may have grade $\rzero$, when they are actually \emph{used}, that is, are subterms of expressions, they should be evaluated with a non-zero grade, as required by a side condition in the corresponding rules in the bottom section of \cref{fig:semantics}. 
 
In rule \refToRule{ret}, the evaluation grade of the value expression should be enough to cover the current evaluation grade.
In rule \refToRule{let}, expressions $\e_1$ and $\e_2$ are evaluated sequentially, the latter in an environment where the local variable $\x$ has been added  as available resource, modulo renaming with a fresh variable to avoid clashes, with the value and grade obtained by the evaluation of $\e_1$.

In rule \refToRule{app}, an application $\App{\ve_1}{\ve_2}$ is evaluated by first consuming the resources needed to obtain a value from $\ve_1$ and $\ve_2$, with the former expected to be a (possibly recursive) function. Then, the function body is evaluated in an environment where the function name and the parameter have been added as available resources, modulo renaming with fresh variables. 
 The function should be produced in a ``number of copies'', that is, with a grade $\sgr$, enough to cover both all the future recursive calls ($\sgr_2$) and the current use ($\sgr_1$); in particular, for a non-recursive call, $\sgr_2$ could be $\rzero$. Instead, the current use should be non-zero since we are actually using the function.
 

 Note also that, in this rule as in others, there is no required relation between the reduction grades of some premises (in this case, $\sgr$ and $\tgr$) and that of the consequence, here $\rgr$. 
 Of course, depending on the choice of such grades, reduction could either proceed or get stuck due to resource exhaustion; the role of the type system is exactly to show that there is a choice which prevents the latter case.

Rules for match constructs, namely \refToRule{match-u}, \refToRule{match-p}, \refToRule{match-l}, and \refToRule{match-r}, all follow the same pattern. The  resources needed to obtain a value from the value expression to be matched are consumed, and then the continuation is evaluated.  In rule \refToRule{match-u}, there is no value-passing from the matching expression to the continuation, hence their evaluation grades are independent. In rule \refToRule{match-p}, instead, the continuation is evaluated in an environment where the two variables in the pattern have been added as available resources, again modulo renaming. The values associated to the two variables are  the ones   of the corresponding component of the expression to be matched, whereas the grades are the evaluation grade of  the  expression, multiplied by the grade of the component. Rules \refToRule{match-l} and \refToRule{match-r} are analogous. 

Besides the standard typing errors, an evaluation graded $\rgr$ can get stuck (formally, no judgment can be derived) when rule \refToRule{var} cannot be applied since  its side condition does  not hold. Informally, some resource (variable) is exhausted, that is, can no longer be replaced by its value.
Also note that the instrumented reduction is non-deterministic, due to rule \refToRule{var}. That is, when  a resource is needed, it can be consumed in different ways; hence, soundness of the type system will be \emph{soundness-may}, meaning that \emph{there exists} a computation which does not get stuck.
 
 We end this section by illustrating resource consumption in a non-terminating computation.  
 \begin{example}\label{ex:div}
 Consider the function $\diverge=\RecFun{\f}{\x}{\Seq{\y}{\App{\f}{\x}}}$, which clearly diverges  on any argument, by using infinitely many times the external resource $\y$. In \cref{fig:ex-div} we show the (incomplete) proof tree for the evaluation of an application $\Seq{\y}{\App{\f}{\x}}$, in an environment where $\f$ denotes the function $\diverge$, and $\y$ and $\x$ denote $\unit$. For simplicity, as in previous examples, we omit $\rone$ grades, and renaming of variables; moreover, we abbreviate by $\Triple{\rgr}{\sgr}{\tgr}$ the environment $\EnvElem{\y}{\rgr}{\unit},\EnvElem{\f}{\sgr}{\diverge},\EnvElem{\x}{\tgr}{\unit}$.
 \begin{figure}
 \begin{footnotesize}
 \begin{math}
 \begin{array}{l}
\prooftree
      \prooftree
       \justifies      
      \reduceNarrow{\ConfP{\y}{\Triple{\rgr_0}{\sgr_0}{\rone}}}{}{\ConfP{\un}{\Triple{\rgr_1}{\sgr_0}{\rone}}}
     \thickness=0.08em
      \shiftright 0em
      \using
      \endprooftree
  \!\!
\prooftree
      \prooftree
       \justifies      
      \reduceNarrow{\ConfP{\f}{\Triple{\rgr_1}{\sgr_0}{\rone}}}{\sgr_0}{\ConfP{\diverge}{\Triple{\rgr_1}{\rzero}{\rone}}}
     \thickness=0.08em
      \shiftright 0em
      \using
      \endprooftree
      \!\!\!
      \prooftree
       \justifies      
      \reduceNarrow{\ConfP{\x}{\Triple{\rgr_1}{\rzero}{\rone}}}{}{\ConfP{\un}{\Triple{\rgr_1}{\rzero}{\rzero}}}
     \thickness=0.08em
      \shiftright 0em
      \using
      \endprooftree
      \!\!\!
      \prooftree 
      \ldots
      \justifies
\reduceNarrow{\ConfP{\Seq{\y}{\App{\f}{\x}}}{\Triple{\rgr_1}{\sgr_1}{\rone}}}{}{?}
      \thickness=0.08em
      \shiftright 0em
      \using
      \endprooftree
      \justifies      
\reduceNarrow{\Conf{\App{\f}{\x}}{\Triple{\rgr_1}{\sgr_0}{\rone}}}{}{?}
      \thickness=0.08em
      \shiftright 0em
      \using
              \scriptstyle{\textsc{(app)}}
      \endprooftree
      \justifies              
\reduceNarrow{\Conf{\Seq{\y}{\App{\f}{\x}}}{\Triple{\rgr_0}{\sgr_0}{\rone}}}{}{?}
      \thickness=0.08em
      \shiftright 0em
      \using
              \scriptstyle{\textsc{(match-u)}}
      \endprooftree
\end{array}
\end{math}
\end{footnotesize}
\caption{Example of resource-aware evaluation: consumption/divergency}
\label{fig:ex-div}
\end{figure}

In this way, we can focus on the key feature that the (tentative) proof tree shows: the body of $\diverge$ is evaluated infinitely many times, in a sequence of environments, starting from the root, where the grades of the external resource $\y$, assuming that each time it is consumed by $\rone$,  are as follows:
\begin{quoting}
$\rgr_0=\rgr_1\rsum\rone$\BigSpace$\rgr_1=\rgr_2\rsum\rone$\BigSpace\ldots\BigSpace $\rgr_k=\rgr_{k+1}\rsum\rone$\BigSpace\ldots
\end{quoting}
 In the case of the resource $\f$, at each recursive call, the function must be produced with a grade which is the sum  of its current usage (assumed again to be $\rone$) and the grade which will be associated to (a fresh copy of) $\f$ in the environment, to evaluate the body. As a consequence, we also get:
\begin{quoting}
$\sgr_0=\sgr_1\rsum\rone$\BigSpace$\sgr_1=\sgr_2\rsum\rone$\BigSpace\ldots\BigSpace $\sgr_k=\sgr_{k+1}\rsum\rone$\BigSpace\ldots
\end{quoting}

Let us now see what happens depending on the underlying grade algebra, considering $\y$ (an analogous reasoning applies to $\f$).
Taking the grade algebra of natural numbers of \refItem{ex:gr-alg}{nat}, it is easy to see that the above sequence of constraints can be equivalently expressed as:
\begin{quoting}
$\rgr_1=\rgr_0-1$\BigSpace$\rgr_2=\rgr_1-1$\BigSpace\ldots\BigSpace $\rgr_{k+1}=\rgr_{k}-1$\BigSpace\ldots
\end{quoting}
Thus, in a finite number of steps, the grade of $\y$ in the environment becomes $0$, hence the proof tree cannot be continued since we can no longer extract the associated value by rule \refToRule{var}.
In other words, the computation is stuck due to resource consumption.

Assume now to take, instead, natural numbers extended with $\infty$, as defined in \refItem{ex:gr-alg}{ext}.
In this case, if we start with $\rgr_0=\infty$, intuitively meaning that $\y$ can be used infinitely many times, evaluation can proceed forever by taking $\rgr_k=\infty$ for all $k$. The same happens if we take a non-quantitative grade algebra, e.g., that of  access  levels; we can have $\rgr_k=\public$ for all $k$. However, interpreting the rules in \cref{fig:semantics} in the standard inductive way, 
the semantics we get \emph{does not} formalize such non-terminating evaluation, since we only consider judgments with a \emph{finite} proof tree. We will see in \cref{sect:soundness} how to extend the semantics to model \mbox{non-terminating computations as well.}
\end{example}


\section{Resource-aware type system}\label{sect:typesystem}
Types, defined in \cref{fig:types}, are those expected for the constructs in the syntax: functional, $\Unit$,  (tensor) product, sum, and (equi-)recursive types, obtained by interpreting the productions \emph{coinductively}, so that infinite\footnote{More precisely, \emph{regular} terms, that is, those with finitely many distinct subterms.} terms are allowed. However, they are \emph{graded}, that is, decorated with a grade, and the type subterms are graded.   Moreover, accordingly with the fact that functions are possibly recursive, arrows in functional types are decorated with a grade as well,  called  the  \emph{recursion grade} in the following, expressing the recursive usage of the function; thus, functional types decorated with $\rzero$ are non-recursive.  

\begin{figure}
\begin{grammatica}
\produzione{\tau,\sigma}{ \funType{\T}{\sgr}{\ST}\mid\Unit\mid\PairT{\T}{\ST}\mid\SumT{\T}{\ST}}{non-graded type}\\
\produzione{\T,\ST}{ \Graded{\tau}{\rgr}}{graded type}\\
\produzione{\cctx}{\VarGrade{\x_1}{\rgr_1}, \ldots, \VarGrade{\x_n}{\rgr_n}}{grade context}\\
\produzione{\Gamma,\Delta}{\VarGradeType{\x_1}{\rgr_1}{\tau_1}, \ldots, \VarGradeType{\x_n}{\rgr_n}{\tau_n}}{(type-and-grade) context}
\end{grammatica}
\caption{Types and (type-and-grade) contexts}
\label{fig:types}
\end{figure}
 In (type-and-grade) contexts, also defined in \cref{fig:types}, order is immaterial and $\x_i\neq\x_j$ for $i\neq j$; hence, they represent maps from variables  to  pairs consisting of a grade, sometimes called \emph{coeffect} \citep{PetricekOM13,PetricekOM14,BrunelGMZ14} when used in this position, and a (non-graded) type. We write $\dom{\Gamma}$ for $\{\x_1,\ldots,\x_n\}$.  Equivalently, a type-and-grade context can be seen as a pair consisting of the standard type context $\VarType{\x_1}{\tau_1}\ldots,\VarType{\x_n}{\tau_n}$, and the grade context $\VarGrade{\x_1}{\rgr_1}, \ldots, \VarGrade{\x_n}{\rgr_n}$, the latter representing a function $\gamma$ from variables into grades with finite support, that is, $\cctx(\x)\ne\rzero$ for finitely many $\x\in\Vars$.   Grade contexts will play a key role in \cref{sect:nw-sem} to express the no-waste version of the semantics.

We define the following operations on contexts:
 
\begin{itemize}
\item a partial order $\ctxord$ 
\begin{align*} 
\emptyset &\ctxord \emptyset  & \\ 
\VarGradeType{\x}{\sgr}{\tau}, \Gamma &\ctxord \VarGradeType{\x}{\rgr}{\tau}, \Delta &&\text{if $\sgr \rord \rgr$  and $\Gamma \ctxord \Delta$} \\ 
\Gamma &\ctxord \VarGradeType{\x}{\rgr}{\tau}, \Delta  &&\text{if $\x\not\in\dom{\Gamma}$  and $\Gamma \ctxord \Delta$ and $\rzero\rord\rgr$} 
\end{align*} 
\item a sum $\ctxsum$ 
\begin{align*} 
\emptyset \ctxsum \Gamma &= \Gamma  \\ 
(\VarGradeType{\x}{\sgr}{\tau}, \Gamma) \ctxsum ( \VarGradeType{\x}{\rgr}{\tau}, \Delta) &= \VarGradeType{\x}{\sgr \rsum \rgr}{\tau}, (\Gamma \ctxsum\Delta) \\ 
(\VarGradeType{\x}{\sgr}{\tau}, \Gamma) \ctxsum \Delta &= \VarGradeType{\x}{\sgr}{\tau}, (\Gamma \ctxsum \Delta)  
  &\text{if $\x \notin \dom{\Delta}$} 
\end{align*} 
\item a scalar multiplication $\ctxmul$ 
\begin{align*}  
\sgr \ctxmul \emptyset = \emptyset 
  && 
\sgr \ctxmul (\VarGradeType{\x}{\rgr}{\tau},\Gamma) =  \VarGradeType{\x}{\sgr \rmul \rgr}{\tau}, (\sgr {\ctxmul} \Gamma) 
\end{align*} 
\end{itemize}  
These operations on type-and-grade contexts are obtained by lifting the corresponding operations  
 on grade contexts, which are the pointwise extension of those on grades,
to handle types. 
In this step, the addition becomes partial since a variable in the domain of both contexts is required to have the same type.

\begin{remark}\label{rem:order}
It is important to note that, when overapproximating contexts with the partial order, new resources can be added only with a discardable grade. For instance, in the linearity grade algebra of \refItem{ex:gr-alg}{lin}, $\Gamma\not\ctxord(\Gamma,\VarGradeType{\y}{1}{\sigma})$. In other words, the partial order models \emph{no-waste} approximation.
\end{remark}

In \cref{fig:typing}, we give the typing rules,  which are \emph{parameterized} on the underlying grade algebra. 
 As for instrumented reduction, the resource-aware type system is formalized by two judgments, $\IsWFExp{\Gamma}{\ve}{\T}$ and $\IsWFExp{\Gamma}{\e}{\T}$, for value expressions and expressions, respectively. However, differently from reduction, the two judgments are mutually recursive, due to rule \refToRule{t-fun}, hence the metarules in \cref{fig:typing} define a unique judgment which is their union. 
We only comment  on  the most significant points. 
\begin{figure}
\begin{math}
\begin{array}{l}
\NamedRule{t-sub-v}{
  \IsWFExp{\Gamma}{\ve}{\Graded{\tau}{\rgr}} 
}{ \IsWFExp{\Delta}{\ve}{\Graded{\tau}{\sgr}} } 
{ \sgr \rord \rgr\\
 \Gamma \ctxord \Delta
}
\BigSpace
\NamedRule{t-var}{ }
{ \IsWFExp{\VarGradeType{\x}{\rgr}{\tau}}{\x}{\Graded{\tau}{\rgr}} } 
{} 
\\[4ex]
\NamedRule{t-fun}{ 
  \IsWFExp{\Gamma,\VarGradeType{\f}{\sgr}{\funType{\Graded{\tau_1}{\rgr_1}}{\sgr}{\Graded{\tau_2}{\rgr_2}}},\VarGradeType{\x}{\rgr_1}{\tau_1}}{\e}{\Graded{\tau_2}{\rgr_2}} 
}{ \IsWFExp{\rgr\ctxmul\Gamma}{\RecFun{\f}{\x}{\e}}{\Graded{(\funType{\Graded{\tau_1}{\rgr_1}}{\sgr}{\Graded{\tau_2}{\rgr_2}})}{\rgr}} } 
{} 
\BigSpace

\NamedRule{t-unit}{}
{ \IsWFExp{\emptyset}{\unit}{\Graded{\Unit}{\rgr}} }
{}

\BigSpace
\\[4ex]
\NamedRule{t-pair}{
  \IsWFExp{\Gamma_1}{\ve_1}{\GradedInd{\tau}{1}{\rgr}} 
  \Space 
  \IsWFExp{\Gamma_2}{\ve_2}{\GradedInd{\tau}{2}{\rgr}} 
}{ \IsWFExp{\rgr\ctxmul(\Gamma_1 \ctxsum \Gamma_2)}{\PairExp{\rgr_1}{\ve_1}{\ve_2}{\rgr_2}}{\Graded{(\PairT{\GradedInd{\tau}{1}{\rgr}}{\GradedInd{\tau}{2}{\rgr}})}{\rgr} } } 
{} 
\\[4ex]
\NamedRule{t-inl}{
  \IsWFExp{\Gamma}{\ve}{\GradedInd{\tau}{1}{\rgr}}
}{ \IsWFExp{\rgr\ctxmul\Gamma}{\Inl{\rgr_1}{\ve}}{\Graded{(\SumT{\GradedInd{\tau}{1}{\rgr}}{\GradedInd{\tau}{2}{\rgr}})}{\rgr}} }
{}
\BigSpace
\NamedRule{t-inr}{
  \IsWFExp{\Gamma}{\ve}{\GradedInd{\tau}{2}{\rgr}}
}{ \IsWFExp{\rgr\ctxmul\Gamma}{\Inr{\rgr_2}{\ve}}{\Graded{\SumT{(\GradedInd{\tau}{1}{\rgr}}{\GradedInd{\tau}{2}{\rgr}})}{\rgr}} } 
{} 
\\[4ex]
\end{array}
\end{math}

\hrule

\begin{math}
\begin{array}{l}
\\
\NamedRule{t-sub}{
  \IsWFExp{\Gamma}{\e}{\Graded{\tau}{\rgr}} 
}{ \IsWFExp{\Delta}{\e}{\Graded{\tau}{\sgr}} } 
{ \sgr \rord \rgr\\
 \Gamma \ctxord \Delta
}
\Space
\NamedRule{t-ret}{ 
  \IsWFExp{\Gamma}{\ve}{\Graded\tau\rgr} 
}{ \IsWFExp{\Gamma}{\Return\ve}{\Graded\tau\rgr} }
{ \rgr\ne\rzero } 
\\[4ex]
\NamedRule{t-let}{
  \IsWFExp{\Gamma_1}{\e_1}{\Graded{\tau_1}{\rgr_1}} 
  \ \
  \IsWFExp{\Gamma_2, \VarGradeType{\x}{\rgr_1}{\tau_1}}{\e_2}{\Graded{\tau_2}{\rgr_2}} 
}{ \IsWFExp{\Gamma_1\ctxsum\Gamma_2}{\Let{\x}{\e_1}{\e_2}}{\Graded{\tau_2}{\rgr_2}} } 
{} 
\\[4ex]
\NamedRule{t-app}{
  \IsWFExp{\Gamma_1}{\ve_1}{\Graded{(\funType{\Graded{\tau_1}{\rgr_1}}{\sgr}{\Graded{\tau_2}{\rgr_2}})}{(\rgr\rsum\rgr\rmul\sgr)}} 
  \
  \IsWFExp{\Gamma_2}{\ve_2}{\Graded{\tau_1}{\rgr\rmul\rgr_1}} 
}{ \IsWFExp{\Gamma_1 \ctxsum \Gamma_2}{\App{\ve_1}{\ve_2}}{\Graded{\tau_2}{\rgr\rmul\rgr_2}} }
{\rgr\neq\rzero}
\\[4ex]
\NamedRule{t-match-u}{
  \IsWFExp{\Gamma_1}{\ve}{\Graded{\Unit}{\rgr}} 
  \Space
  \IsWFExp{\Gamma_2}{\e}{\T}
}{ \IsWFExp{\Gamma_1 \ctxsum \Gamma_2}{\MatchUnit{\ve}{\e}}{\T} } 
{\rgr\ne\rzero} 
\\[4ex]
\NamedRule{t-match-p}{
  \IsWFExp{\Gamma_1}{\ve}{\Graded{(\PairT{\GradedInd{\tau}{1}{\rgr}}{\GradedInd{\tau}{2}{\rgr}})}{\rgr}}
  \Space
  \IsWFExp{\Gamma_2,\VarGradeType{\x}{\rgr\rmul\rgr_1}{\tau},\VarGradeType{\y}{\rgr\rmul\rgr_2}{\tau_2}}{\e}{\T}
}{ \IsWFExp{\Gamma_1\ctxsum\Gamma_2}{\Match{\x}{\y}{\ve}{\e}}{\T} } 
{ \rgr\ne\rzero } 
\\[4ex]
\NamedRule{t-match-in}{
  \IsWFExp{\Gamma_1}{\ve}{\Graded{(\SumT{\GradedInd{\tau}{1}{\rgr}}{\GradedInd{\tau}{2}{\rgr}})}{\rgr}}\Space
  \IsWFExp{\Gamma_2,\VarGradeType{\x}{\rgr\rmul\rgr_1}{\tau_1}}{\e_1}{\T}\Space
  \IsWFExp{\Gamma_2,\VarGradeType{\x}{\rgr\rmul\rgr_2}{\tau_2}}{\e_2}{\T}
}{ \IsWFExp{\Gamma_1\ctxsum\Gamma_2}{\Case{\ve}{\x}{\e_1}{\x}{\e_2}}{\T} }
{ \rgr\ne\rzero } 
\end{array}
\end{math}
\caption{Typing rules}
\label{fig:typing} 
\end{figure}
In rule \refToRule{t-sub-v} and \refToRule{t-sub},   the context can be made more general, and the grade of the type more specific. This means that, on one hand, variables can get less constraining grades. For instance, assuming  affinity grades as in \refItem{ex:gr-alg}{lin}, an expression which can be typechecked assuming to use a given variable at most once (grade 1) can be typechecked with no constraints (grade $\omega$). On the other hand, an expression can get a more constraining grade. For instance, an expression of grade $\omega$ can be used where a grade 1 is required.

If we take $\rgr=\rone$, then rule \refToRule{t-var} is the standard rule for variable in graded systems, where the grade context is the map where the given variable is used once, and no other  variable  is used.  Here, more generally, the variable can get an arbitrary grade $\rgr$, provided that  the context is multiplied by $\rgr$.  The same ``local promotion'' can be applied in  the other structural rules for value expressions. 

In rule \refToRule{t-fun}, a (possibly recursive) function can get a graded functional type, provided that its body can get the result type in the context enriched by assigning the functional type to the function name, and the parameter type to the parameter.  As mentioned, we expect  the recursion grade  $\sgr$ to be $\rzero$ for a non-recursive function;  for a recursive function, we expect  $\sgr$ to be  an ``infinite'' grade, in a sense which will be clarified in  \cref{ex:div-type} below. 

In the rules in the bottom section, when a value expression is used as subterm of an expression, its grade is required to be non-zero, since it is evaluated in the computation, hence its resource consumption should be taken into account. 

In rule \refToRule{t-app}, the function should be produced with a grade which is the sum of that required for the current usage ($\rgr$) and that corresponding to the recursive calls: the latter are the grade required for a single usage multiplied by  the recursion grade  ($\sgr$). For a non-recursive function ($\sgr=\rzero$), the rule turns out to be as expected for a usual application.

\begin{example}\label{ex:div-type}

As an example of type derivation, we consider the function ${\diverge=\RecFun{\f}{\x}{\Seq{\y}{\App{\f}{\x}}}}$ introduced in \cref{ex:div}.
 In \cref{fig:ex-div-type}, we show a (parametric) proof tree deriving for $\diverge$  a type of the shape $\Graded{(\funType{\Graded{\Unit}{\rgr_1}}{\sgr}{\Graded{\Unit}{\rgr_2}})}{}$, in a context providing the external resource $\y$. 
\begin{figure}
\begin{math}
\begin{array}{l}
\prooftree
	\prooftree
		\prooftree
			\prooftree	
			\justifies
			\IsWFExp{\VarGradeType{\y}{\rgr'}{\Un}}{\y}{\Graded{\Un}{\rgr'}}
			\thickness=0.08em
			\shiftright 0em
			\using
			\scriptstyle{\textsc{(t-var)}}
			\endprooftree
			\prooftree
				\prooftree	
				\justifies
				\IsWFExp{\VarGradeType{\f}{\rgr\rsum\rgr\rmul\sgr}{\tau}}{\f}{\Graded{\tau}{\rgr\rsum\rgr\rmul\sgr}}				
				\thickness=0.08em
				\shiftright 0em
				\using
				\scriptstyle{\textsc{(t-var)}}
				\endprooftree
				\prooftree	
				\justifies
				\IsWFExp{\VarGradeType{\x}{\rgr\rmul\rgr_1}{\Un}}{\x}{\Graded{\Un}{\rgr\rmul\rgr_1}}
				\thickness=0.08em
				\shiftright 0em
				\using
				\scriptstyle{\textsc{(t-var)}}
				\endprooftree
			\justifies
			\IsWFExp{\VarGradeType{\f}{\rgr\rsum\rgr\rmul\sgr}{\tau},\VarGradeType{\x}{\rgr\rmul\rgr_1}{\Graded{\Un}{}}}{\App{\f}{\x}}{\Graded{\Un}{\rgr\rmul\rgr_2}}
			\thickness=0.08em
			\shiftright 0em
			\using
			\scriptstyle{\meta{\rgr\neq\rzero} \ \textsc{(t-app)}}
			\endprooftree		
		\justifies
		\IsWFExp{\VarGradeType{\y}{\rgr'}{\Un},\VarGradeType{\f}{\rgr\rsum\rgr\rmul\sgr}{\tau},\VarGradeType{\x}{\rgr\rmul\rgr_1}{\Graded{\Un}{}}}{\Seq{\y}{\App{\f}{\x}}}{\Graded{\Un}{\rgr\rmul\rgr_2}}
		\thickness=0.08em
		\shiftright 0em
		\using
		\scriptstyle{\meta{\rgr'\neq\rzero} \ \textsc{(t-match-u)}}
		\endprooftree
    \justifies
		\IsWFExp{\VarGradeType{\y}{\rgr'}{\Un},\VarGradeType{\f}{\sgr}{\tau},\VarGradeType{\x}{\rgr_1}{\Graded{\Un}{}}}{\Seq{\y}{\App{\f}{\x}}}{\Graded{\Un}{\rgr_2}}
		\thickness=0.08em
		\shiftright 0em
    \using
	\scriptstyle{\meta{(\rgr\rsum\rgr\rmul\sgr)\rord\sgr\quad\rgr\rmul\rgr_1\rord\rgr_1\quad\rgr_2\rord\rgr\rmul\rgr_2 } }
	\endprooftree
	\justifies
	\IsWFExp{\VarGradeType{\y}{\rgr'}{\Un}}{\RecFun{\f}{\x}{\Seq{\y}{\App{\f}{\x}}}}{\Graded{\tau}}{}
	\thickness=0.08em
	\shiftright 0em
	\using
	\scriptstyle{\textsc{(t-fun)}} 
         \hspace{9ex}
        \displaystyle{\tau=\funType{\Graded{\Un}{\rgr_1}}{\sgr}{\Graded{\Un}{\rgr_2}}}
	\endprooftree
      \hspace{-12pt}
\end{array}
\end{math}
\caption{Example of type derivation: recursive function}
\label{fig:ex-div-type}
\end{figure}
 Consider, first of all, the condition that the recursion grade $\sgr$ should satisfy, that is $(\rgr\rsum\rgr\rmul\sgr)\rord\sgr$, for some $\rgr\neq\rzero$, meaning that it should be enough to cover the recursive call in the body and all the further recursive calls.\footnote{Note that $\rgr$ is arbitrary, since there is no sound default grade, and  it is  only required to be non-zero since the function \mbox{is actually used.}}
Assuming the grade algebra of natural numbers of \refItem{ex:gr-alg}{nat}, there is no grade $\sgr$ satisfying this condition. In other words, the type system correctly rejects the function since its application would  get stuck due to resource consumption, as illustrated in \cref{ex:div}.  On the other hand, for, e.g., $\sgr=\infty$,  with natural numbers extended as in \refItem{ex:gr-alg}{ext}, $(\rgr\rsum\rgr\rmul\sgr)\rord\sgr$ would hold for any $\rgr\neq 0$, hence the function would be well-typed. 
Moreover, there would be no constraints on the parameter and return type, since the conditions $ \rgr\rmul\rgr_1\rord\rgr_1$ and $\rgr_2\rord\rgr\rmul\rgr_2$ would be always satified taking $\rgr=1$.
Assuming the grade algebra of  access  levels  introduced before \cref{ex:ex2},  where $\rone=\public$, for $\sgr=\public$ the condition is satisfied analogously, again with no constraints on $\rgr_1$ \mbox{and $\rgr_2$.} For $\sgr=\private$, instead, it only holds for $\rgr=\private$, hence the condition $\rgr_2\rord\rgr\rmul\rgr_2$ prevents the return type of the function to be $\public$, accordingly with the intuition that a function used in $\private$ mode cannot return a $\public$ result. 
In such cases, the type system correctly accepts the function since its application to a value never gets stuck.  Finally note that, to type an application of the function, e.g., to derive that  
$\App{\diverge\,}{\un}$ has type $\Graded{\Un}{\rgr_2}$, $\diverge$ should get grade $\rone\rsum\sgr$, hence the grade of the external resource $\y$ should be $(\rone\rsum\sgr)\rmul\rgr'$, that is, it should be usable infinitely many times as well. 
\end{example}

 A similar reasoning applies in general; namely, for recursive calls in a function's body we get a condition of shape $(\rgr\rsum\rgr\rmul\sgr)\rord\sgr$, with $\rgr\neq\rzero$, forcing the grade $\sgr$ of the function to be ``infinite''.  \cite{BrunelGMZ14} and the subsequent work of \cite{AbelDE23} assume the existence of fixed-point operators whose typing rules explicitly impose a restriction analogous to ours. In our system, instead, this restriction arises from the application of the typing rules. As already observed by \cite{BrunelGMZ14}, the condition $(\rgr\rsum\rgr\rmul\sgr)\rord\sgr$, with $\rgr\neq\rzero$, forces the grade $\sgr$ of the function to be ``infinite''.   
This happens regardless of whether the recursive function is actually always diverging, as in the example above, or terminating on some/all arguments. 
On the other hand, in the latter case the resource-aware semantics does terminate, as expected, provided that the initial amount of resources is sufficient to cover the (finite number of) recursive calls.  This behavior is perfectly reasonable,  since the type system is a static overapproximation of  the  evaluation. We will show an example  of this terminating resource-aware evaluation in \cref{sect:examples} (\cref{fig:term-rec}).

\medskip

The following lemmas show that the promotion rule, usually explicitly stated in graded type systems, is admissible in our system (\cref{lem:promotion}), and also a converse holds for value expressions (\cref{lem:demotion}). 
Note that we can promote an expression only using a non-zero grade, to ensure that non-zero constraints on grades in typing rules for expressions are preserved.\footnote{Without assuming the grade algebra to be integral, we would need to use grades which are not zero-divisors. }
These lemmas also show that we can assign to a value expression any grade provided that the context is  appropriately adjusted:
by demotion we can always derive $\rone$ (taking $\rgr = \rone$) and then by promotion \mbox{any grade. } 

\begin{lemma}[Promotion] \label{lem:promotion} \
\begin{enumerate}
\item\label{lem:promotion:val} 
If $\IsWFExp{\Gamma}{\ve}{\Graded{\tau}{\rgr}}$ then 
$\IsWFExp{\sgr\ctxmul\Gamma}{\ve}{\Graded{\tau}{\sgr\rmul\rgr}}$
\item\label{lem:promotion:exp} 
If $\IsWFExp{\Gamma}{\e}{\Graded{\tau}{\rgr}}$ and $\sgr \neq \rzero$, then 
$\IsWFExp{\sgr\ctxmul\Gamma}{\e}{\Graded{\tau}{\sgr\rmul\rgr}}$. 
\end{enumerate}
\end{lemma}
\begin{proof}
By induction on the typing rules. We show only some cases.
\begin{description}

\item [\refToRule{t-sub-v}]
We have  $\IsWFExp{\Delta}{\ve}{\Graded{\tau}{\sgr'}}$, $\IsWFExp{\Gamma}{\ve}{\Graded{\tau}{\rgr}}$, $\rgr\rord\sgr'$ and $\Delta \ctxord\Gamma$.
By induction hypothesis  $\IsWFExp{\sgr\rmul\Delta}{\ve}{\Graded{\tau}{\sgr\rmul\sgr'}}$. By monotonicity  $\sgr\rmul\rgr\rord\sgr\rmul\sgr'$ and $\sgr\ctxmul\Delta\ctxord\sgr\ctxmul\Gamma $, so, by  \refToRule{t-sub} we get $\IsWFExp{\sgr\rmul\Gamma}{\e}{\Graded{\tau}{\sgr\rmul\rgr}}$, that is, the thesis.
\item [\refToRule{t-var}] By rule \refToRule{t-var} we get the thesis.
\item [\refToRule{t-fun}] 
We have 
 $\Gamma=\rgr\ctxmul\Gamma'$, $\ve={\RecFun{\f}{\x}{\e'}}$, 
$\tau={\funType{\Graded{\tau_1}{\rgr_1}}{\sgr'}{\Graded{\tau_2}{\rgr_2}}} $ and
$\IsWFExp{\Gamma',\VarGradeType{\f}{\sgr'}{\funType{\GradedInd{\tau}{1}{\rgr}}{\sgr'}{\GradedInd{\tau}{2}{\rgr}}},\VarGradeType{\x}{\rgr_1}{\tau_1}}{\e'}{\GradedInd{\tau}{2}{\rgr}}$. By rule \refToRule{t-fun}, $\IsWFExp{(\sgr\rmul\rgr)\ctxmul\Gamma'}{\ve}{\Graded{\tau}{\sgr\rmul\rgr}}$.
From $(\sgr\rmul\rgr)\ctxmul\Gamma'=\sgr\rmul(\rgr\ctxmul\Gamma')=\sgr\ctxmul\Gamma$ \mbox{we get the thesis.}
\item [\refToRule{t-pair}, \refToRule{t-inl} and \refToRule{t-inr}]  Similar to \refToRule{t-fun}.  
\item [\refToRule{t-sub}]
We have  $\IsWFExp{\Delta}{\e}{\Graded{\tau}{\sgr'}}$, $\IsWFExp{\Gamma}{\e}{\Graded{\tau}{\rgr}}$, $\rgr\rord\sgr'$ and $\Delta \ctxord\Gamma$.
By induction hypothesis  $\IsWFExp{\sgr\rmul\Delta}{\e}{\Graded{\tau}{\sgr\rmul\sgr'}}$. By monotonicity  $\sgr\rmul\rgr\rord\sgr\rmul\sgr'$ and $\sgr\ctxmul\Delta\ctxord\sgr\ctxmul\Gamma $, so, by \refToRule{t-sub} we get $\IsWFExp{\sgr\rmul\Gamma}{\e}{\Graded{\tau}{\sgr\rmul\rgr}}$, that is, the thesis.
\item [\refToRule{t-app}]
We have  $\IsWFExp{\Gamma_1\ctxsum\Gamma_2}{\App{\ve_1}{\ve_2}}{\Graded{\tau_2}{\rgr'\rmul\rgr_2}}$.
By induction hypothesis  $\IsWFExp{\sgr\rmul\Gamma_1}{\ve_1}{\Graded{(\funType{\GradedInd{\tau}{1}{\rgr}}{\sgr'}{\GradedInd{\tau}{2}{\rgr}})}{\sgr\rmul(\rgr'\rsum\rgr'\rmul\sgr')}}$ and $\IsWFExp{\sgr\rmul\Gamma_2}{\ve_2}{\Graded{\tau}{\sgr\rmul\rgr'\rmul\rgr_1}}$. Since $\sgr\neq\rzero$ and $\rgr'\neq\rzero$ and, since the algebra is integral,  $\sgr\rmul\rgr\neq\rzero$. By rule \refToRule{t-app}, $ \IsWFExp{\sgr\rmul(\Gamma_1 \ctxsum \Gamma_2)}{\App{\ve_1}{\ve_2}}{\Graded{\tau_2}{\sgr\rmul\rgr'\rmul\rgr_2}}$, that is, the thesis.
\item [\refToRule{t-match-p}]
By induction hypothesis  $\IsWFExp{\sgr\rmul\Gamma_1}{\ve}{\Graded{(\PairT{\Graded{\tau_1}{\rgr_1}}{\GradedInd{\tau}{2}{\rgr}})}{\sgr\rmul\sgr'}}
$ and $\IsWFExp{\sgr\rmul\Gamma_2,\VarGradeType{\x}{(\sgr\rmul\sgr')\rmul\rgr_1}{\tau},\VarGradeType{\y}{(\sgr\rmul\sgr')\rmul\rgr_2}{\tau_2}}{\e}{\Graded{\tau}{\sgr\rmul\rgr}}$. Since $\sgr\neq\rzero$ and $\sgr'\neq\rzero$ and, since the algebra is integral, $\sgr\rmul\sgr'\neq\rzero$. By rule \refToRule{t-match-p}, $\IsWFExp{\sgr\rmul(\Gamma_1\ctxsum\Gamma_2)}{\Match{\x}{\y}{\ve}{\e}}{\Graded{\tau}{\sgr\rmul\rgr}}$, that is, the thesis.
\end{description}
\end{proof}

\begin{lemma}[Demotion] \label{lem:demotion} \
If $\IsWFExp{\Gamma}{\ve}{\Graded{\tau}{\sgr\rmul\rgr}}$ then
$ \sgr\ctxmul \Gamma'\ctxord\Gamma$
and $\IsWFExp{\Gamma'}{\ve}{\Graded{\tau}{\rgr}}$,  for some $\Gamma'$ .
\end{lemma}
\begin{proof}
By case analysis on $\ve$. We show only some cases.
\begin{description}
\item[$\ve = \x$] By \cref{lem:inversion}(\ref{lem:inversion:x}) $\VarGradeType{\x}{\rgr'}{\tau}\ctxord\Gamma$ and $\sgr\rmul\rgr\rord\rgr'$. Let $\Gamma'=\VarGradeType{\x}{\rgr}{\tau}$. 
By monotonicity of $\rmul$ and, by transitivity of $\ctxord$, $(\sgr\rmul\rgr)\ctxmul\VarGradeType{\x}{\rone}{\tau} = \sgr\ctxmul\Gamma' \ctxord\Gamma$. By \refToRule{t-var} $\IsWFExp{\Gamma'}{\x}{\Graded{\tau}{\rgr}}$.
\item[$\ve = \PairExp{\rgr_1}{\ve_1}{\ve_2}{\rgr_2}$] By \cref{lem:inversion}(\ref{lem:inversion:pair}) $\rgr'\ctxmul(\Gamma_1\ctxsum\Gamma_2)\ctxord\Gamma$ and $\sgr\rmul\rgr\rord\rgr'$ and $\IsWFExp{\Gamma_1}{\ve_1}{\GradedInd{\tau}{1}{\rgr}}$ and ${\IsWFExp{\Gamma_2}{\ve_2}{\GradedInd{\tau}{2}{\rgr}}}$.   Let $\Gamma'=\rgr\ctxmul(\Gamma_1\ctxsum\Gamma_2)$.
By monotonicity of $\rmul$ and by transitivity of $\ctxord$ we have $\sgr\rmul\Gamma'\ctxord\Gamma$. By \refToRule{t-pair} ${\IsWFExp{\Gamma'}{\ve}{\Graded{\tau}{\rgr}}}$. 
\end{description}
\end{proof}

In \cref{fig:typing-conf} we give the typing rules for environments and configurations. 
\begin{figure}
\begin{math}
\begin{array}{l}

\NamedRule{t-env}{
\IsWFExp{\Gamma_i}{\val_i}{\tau_i^\rone}  \Space \forall i \in 1..n
}
{
\IsWFEnv{\Gamma}{\EnvElem{\x_1}{\rgr_1}{\val_1},\ldots,\EnvElem{\x_n}{\rgr_n}{\val_n}}{\Delta}
}{
\Gamma = \VarGradeType{\x_1}{\rgr_1}{\tau_1},\ldots,\VarGradeType{\x_n}{\rgr_n}{\tau_n}\\
(\rgr_1\ctxmul\Gamma_1 \ctxsum \ldots \ctxsum \rgr_n\ctxmul\Gamma_n\ctxsum\Delta)\ctxord \Gamma
}
\\[3ex]
\NamedRule{t-conf}{
\IsWFExp{\Delta}{\E}{\T} \BigSpace
\IsWFEnv{\Gamma}{\env}{\Delta}
}
{
\IsWFConf{\Gamma}{\E}{\env}{\T}
}{ \produzioneinline{\E}{\ve\mid\e}
}
\end{array}
\end{math}
\caption{Typing rules for environments and configurations}
\label{fig:typing-conf}
\end{figure}

The typing judgement for environments has shape $\IsWFEnv{\Gamma}{\env}{\Delta}$. In rule \refToRule{t-env},  $\Gamma$ is the context  \emph{extracted} from the environment, in the sense that, if typechecking succeeds, each variable will get exactly its grade, and the type of its value, in $\env$. In other words, $\Gamma$ gives type and grade information about the resources (variables) provided by the environment. 
 The context $\Gamma$ should be enough to cover both the sum of the contexts needed to typecheck each provided resource, and a \emph{residual} context $\Delta$ left for the program (last side condition).
 The typing judgment for configurations has shape $\IsWFConf{\Gamma}{\E}{\env}{\T}$, with $\E$ either value expression or expression. In rule \refToRule{t-conf}, the residual context obtained by typechecking the environment should be exactly that required by the program.

\begin{remark}\label{rem:t-env}
Note that, in this way, rule \refToRule{t-conf} imposes \emph{no-waste}.
 For instance, taking the linearity grade algebra of \refItem{ex:gr-alg}{lin}, $\Gamma$ cannot offer  resources which are unused in both the expression and the codomain of the environment. 
Hence, e.g., $\Conf{\x}{\EnvElem{\x}{\infty}{\unit},\EnvElem{\y}{1}{\unit}}$, is ill-typed, since the linear resource $\y$ is wasted, even though in the resource-aware semantics such configuration safely reduces to $\Conf{\unit}{\EnvElem{\x}{\infty}{\unit},\EnvElem{\y}{1}{\unit}}$. 
 In other words, the type system checks that available (non-discardable) resources are fully consumed, whereas in the instrumented semantics there is no such check.
 This mismatch between type system and semantics can be eliminated in two ways:
 \begin{itemize}
 \item as  done by  \cite{BianchiniDGZ@OOPSLA23},  with the  more permissive typing rule for configurations shown below:
 \[
 \NamedRule{t-conf}{
\IsWFExp{\Delta}{\E}{\T} \BigSpace
\IsWFEnv{\Gamma}{\env}{\widehat{\Gamma}}
}
{
\IsWFConf{\Gamma}{\E}{\env}{\T}
}{\produzioneinline{\E}{\ve\mid\e}\\
\Delta\ctxsum\widehat{\Gamma}\ctxsum\Theta\ctxord\Gamma 
}
 \]
The additional context $\Theta$ allows the environment to contain resources which will be wasted; e.g., the configuration $\Conf{\x}{\EnvElem{\x}{\infty}{\unit},\EnvElem{\y}{1}{\unit}}$ becomes well-typed with $\Theta=\VarGradeType{\y}{1}{\Unit}$.  
 
\item Conversely, we can refine the instrumented semantics to check no-waste as well; e.g., the configuration $\Conf{\x}{\EnvElem{\x}{\infty}{\unit},\EnvElem{\y}{1}{\unit}}$ is stuck. This is the approach we will develop in \cref{sect:nw-sem}.
\end{itemize}
\end{remark}


\section{Programming examples and discussions}
\label{sect:examples} 
In this section, for readability, we  use the surface syntax and, moreover,  assume some standard additional constructs and syntactic conventions. Notably, we generalize  (tensor)  product types to tuple types, with the obvious extended syntax, and sum types  to variant types, written ${\VariantT{\lab_1}{\T_1}+\ldots+\VariantT{\lab_n}{\T_n}}$  for some \emph{tags} $\lab_1$, \ldots, $\lab_n$.  Injections are generalized to tagged variants $\Variant{\lab}{\rgr}{\e}$, and the corresponding match constructs are of  shape
$\texttt{match}\ \e\ \texttt{with}\ \Variant{\lab_1}{}{\x_1}\ \texttt{or} \ldots\ \texttt{or}\ \Variant{\lab_n}{}{\x_n}$. We write just $\lab$ as an abbreviation for an addend $\VariantT{\lab}{\Graded{\Unit}{\rzero}}$ in a variant type,   and also for the corresponding tagged variants $\Variant{\lab}{\rzero}{\unit}$ and patterns $\Variant{\lab}{}{\x}$ in a matching construct. 
Moreover, we will use type and function definitions (that is, synonyms  for function and type expressions), and, as customary, represent (equi-)recursive types
by equations.  Finally, we will omit $\rone$ annotations as in the previous examples, and we will consider $\rzero$ as default, hence omitted, as  recursion grade  (that is, functions are by default non-recursive). 

\paragraph*{Natural numbers and lists}
The encoding of  Booleans,  natural numbers, lists of natural numbers, and optional natural numbers, is given at the top of \cref{fig:boolNatList} followed by the definition of some standard functions. 
\begin{figure}
\begin{lstlisting} [basicstyle=\ttfamily\footnotesize]
Bool = true $+$ false 
Nat = zero $+$ succ:Nat
NatList = empty + cons:(Nat $\otimes$ NatList)
OptNat = none $+$ some:Nat

not: Bool -> Bool
not = \b.match b with true -> false or false -> true 

even: Nat ->$_\infty$ Bool
even = rec ev.\n.match n with zero -> true or succ m -> not (ev m)
 
_+_: Nat ->$_\infty$ Nat -> Nat
_+_ = rec sum.\n.\m. match n with zero -> m or succ x -> succ (sum x m) 
        
double: Nat$^2$ -> Nat
double n = n + n

_*_: Nat ->$_\infty$ Nat$^\infty$ -> Nat  
_*_ = rec mult.\n.\m.match n with zero -> zero or succ x -> (mult x m) + m 

length : NatList ->$_\infty$ Nat 
length = rec len.
         \ls.match ls with empty -> zero 
             or cons ls1 -> match ls1 with <_,tl> -> succ (len tl)   
              
get : NatList ->$_\infty$ Nat -> OptNat
get = rec g.
         \ls.\i.match ls with empty -> none 
                or cons ls1 -> 
                  match ls1 with <hd,tl> -> match i with zero -> some hd 
                                            or succ j -> g j tl         
 \end{lstlisting}
\caption{ Examples of type and function definitions}\label{fig:funNatList}\label{fig:boolNatList}
\end{figure}
 Assume, firstly, the grade algebra of natural numbers with bounded usage, \refItem{ex:gr-alg}{nat}, extended with $\infty$, as in \refItem{ex:gr-alg}{ext}, needed as annotation of recursive functions, as has been illustrated in \cref{ex:div-type}; we discuss below what happens taking exact usage instead.

 As a first comment, note that in types of recursive functions the  recursion grade needs to be  $\infty$, as expected. 
On the other hand, most function parameters are graded $1$, since they are used at most once in each branch of the function's body. The second parameter of  the function \lstinline{*}, instead, needs to be graded $\infty$. Indeed, it is used in the body of the function both as argument of sum, and \emph{inside} the recursive call. Hence, its grade $\rgr$ should satisfy the equation $(\rone\rsum\rgr)\leq\rgr$, analogously to what happens for the  recursive grade; compare with the parameter of function \lstinline{double}, which is used twice as well, and can be graded $2$.
In the following alternative definitions:

\begin{lstlisting}[basicstyle=\ttfamily\footnotesize]
double: Nat$^1$ -> Nat
double n = n * succ succ zero

double: Nat$^\infty$ -> Nat
double n = succ succ zero * n
\end{lstlisting}
the parameter needs to be graded differently depending on  whether it  is  used as either first or second argument of \lstinline{*}.  In other words, the resource-aware type system captures, as expected, non-extensional properties. 

Assume now the grade algebra of natural numbers with exact usage, again extended with $\infty$. Interestingly enough, the functions \lstinline{length} and \lstinline{get} above are no longer typable. 

In \lstinline{length}, this is due to the fact that, when the list is non-empty, the head is unused,  whereas it should be used exactly once as the tail, since the grade of a pair is ``propagated'' to both the components. 
The function would be typable assuming for lists the type \lstinline{NatList=empty+cons:(Nat$^\rzero\otimes$NatList)}, which, however, would mean to essentially handle a list as a natural number.  

Function \lstinline{get}, analogously, cannot be typed since, in the last line, only one component of a non-empty list (either the head or the tail) is used in  a  branch of the match, whereas both should be used exactly once. 
Both functions could be typed, instead, grading with $\infty$ the list parameter; this would mean to allow an arbitrary usage in the body. These examples suggest that, in a grade algebra with exact usage, such as that of natural numbers, or the simpler linearity grade algebra, see \refItem{ex:gr-alg}{lin}, there is often no middle way between imposing severe limitations on code, to ensure linearity (or, in general, exactness), and allowing code which is essentially non-graded.

\paragraph*{ Cartesian  product} 
The product type we consider in \cref{fig:types} is the \emph{tensor} product, also called \emph{multiplicative}, following Linear Logic terminology \citep{Girard87}. 
In the destructor construct for such  a  type, both components are simultaneously extracted, each one with a grade which is (a multiple of) that of the pair, see the semantic\footnote{Typing rules follow the same pattern.} rule \refToRule{match-p} in \cref{fig:semantics}. Thus, as shown in the examples above, programs which discard the use of either component cannot be typed in a non-affine grade algebra. Correspondingly, the resource consumption for constructing a (multiplicative) pair is the \emph{sum} of those for constructing the two components, corresponding to a sequential evaluation, see rule \refToRule{pair} in \cref{fig:semantics}. The  \emph{Cartesian}  product, instead, 
formalized in \cref{fig:cartesian}, has one destructor for each component, so that which is not extracted is discarded. Correspondingly, the resource consumption for constructing  a   pair is \emph{an upper bound} of those for constructing the two components, corresponding in a sense to a non-deterministic evaluation. 
The \lstinline{get} example, rewritten using the constructs of the  Cartesian  product, becomes typable even in a non-affine grade algebra. 
In an affine grade algebra, programs can always be rewritten replacing the cartesian product with the tensor, and conversely; in particular,  $\Proj{i}{\ve}$ can be encoded as 
$\Match{\x_1}{\x_2}{\ve}{\x_i}$,  even though, as said above, resources are consumed differently (sum versus upper bound). An interesting remark is that record/object calculi, where generally width subtyping is allowed, 
meaning that components can be discarded at runtime, and object construction happens by sequential evaluation of the fields, need to be modeled by multiplicative product and affine grades. In future work we plan to investigate object calculi which are \emph{linear}, or, more generally, use resources in \mbox{an exact way.}

\begin{figure}
\begin{grammatica}
\produzione{\ve}{\ldots\mid\APairExp{\rgr_1}{\ve_1}{\ve_2}{\rgr_2}\mid \Fst\ve \mid \Snd\ve }{} \\ 
\produzione{\tau,\sigma}{ \ldots\mid\APairT{\T}{\ST}}{}\\
\\
\end{grammatica}

\begin{math}
\begin{array}{l}
\NamedRule{apair}{
  \redval{\ve_1}{\env}{\rgr\rmul\rgr_1}{\val_1}{\env'} 
  \Space
  \redval{\ve_2}{\env}{\rgr\rmul\rgr_2}{\val_2}{\env'} 
}{ \redval{\APairExp{\rgr_1}{\ve_1}{\ve_2}{\rgr_2}}{\env}{\rgr}{\APairExp{\rgr_1}{\val_1}{\val_2}{\rgr_2}}{\env'} }
{} 
\Space 
\NamedRule{proj}{ 
  \redval{\ve}{\env}{\sgr}{\APairExp{\rgr_1}{\val_1}{\val_2}{\rgr_2}}{\env'}
}{ \reduce{\Conf{\Proj{i}{\ve}}{\env}}{\rgr}{\Conf{\val_i}{\env'}} }
{ i\in\{1,2\}\\
\rgr\rord \sgr\rmul\rgr_i  }
\\[3ex]
\NamedRule{t-apair}{
  \IsWFExp{\Gamma}{\ve_1}{\Graded{\tau_1}{\rgr_1}} 
  \Space 
  \IsWFExp{\Gamma}{\ve_2}{\Graded{\tau_2}{\rgr_2}} 
}{ \IsWFExp{\rgr\ctxmul \Gamma}{\APairExp{\rgr_1}{\ve_1}{\ve_2}{\rgr_2}}{\Graded{(\APairT{\Graded{\tau_1}{\rgr_1}}{\Graded{\tau_2}{\rgr_2}})}{\rgr} } } 
{} 
\Space 
\NamedRule{t-proj}{
  \IsWFExp{\Gamma}{\ve}{\Graded{\APairT{(\GradedInd{\tau}{1}{\rgr}}{\GradedInd{\tau}{2}{\rgr}})}{\rgr}} 
}{ \IsWFExp{\Gamma}{\Proj{i}{\ve}}{\Graded{\tau_i}{\rgr\rmul\rgr_i}} }
{ \rgr\ne\rzero }
\end{array}
\end{math}
\caption{Cartesian product}\label{fig:cartesian}
\end{figure}
 
\paragraph*{Terminating recursion} As anticipated, even though the type system can only derive, for recursive functions,  recursion  grades which are ``infinite'', their calls which terminate in standard semantics can terminate also in resource-aware semantics, provided that the initial amount of the function resource is enough to cover the (finite number of) recursive calls, as shown in \cref{fig:term-rec}.

For brevity, we write $\zero$, $\Succ{}{}$, and $\tru$ for \terminale{zero}, \terminale{succ}, and \terminale{true}, respectively, $\CaseNat{\ve}{\e_1}{\x}{\e_2}$ for $\terminale{match}\ \ve\ \terminale{with} \ \zero\,  \rightarrow \e_1 \ \terminale{or} \ \Succ{}{\x} \rightarrow \e_2$, and $\Triple{\rgr}{\sgr}{\val}$ for the environment ${\EnvElem{\ev}{\rgr}{\even},\EnvElem{\n}{\sgr}{\val}}$.
\begin{figure}
 \begin{scriptsize}
 \begin{math}
 \begin{array}{l}

\prooftree
	\prooftree
	\justifies      
	\reduce{\ConfP{\ev}{\Triple{1}{1}{\zero}}}{1}{\ConfP{\even}{\Triple{0}{1}{\zero}}}
	\thickness=0.08em
	\shiftright 0em
	\endprooftree

	\BigSpace
         
	\prooftree
	\justifies      
	\reduce{\ConfP{\n}{\Triple{0}{1}{\zero}}}{}{\ConfP{\zero}{\Triple{0}{0}{\zero}}}
	\thickness=0.08em
	\shiftright 0em
	\endprooftree

	\BigSpace
         
	\prooftree
	\reduce{\ConfP{\tru}{\Triple{0}{0}{\zero}}}{}{\ConfP{\tru}{\Triple{0}{0}{\zero}}}
	\justifies      
	\reduce{\ConfP{\CaseNat{\zero}{\tru}{\n}{\App{\nottt}{(\App{\ev}{\n})}}}{\Triple{0}{0}{\zero}}}{}{\ConfP{\tru}{\Triple{0}{0}{\zero}}}
	\thickness=0.08em
	\shiftright 0em
	\endprooftree
\justifies
\reduce{\ConfP{\App{\ev}{\n}}{\Triple{1}{1}{\zero}}}{}{\ConfP{\tru}{\Triple{0}{0}{\zero}}}
\thickness=0.08em
\shiftright 0em
\endprooftree

 \\[7ex]

\prooftree
	\prooftree
	\justifies      
	\reduceNarrow{\ConfP{\ev}{\Triple{n}{1}{\Succ{n}{\zero}}}}{n}{\ConfP{\even}{\Triple{n{-}1}{1}{\zero}}}
	\thickness=0.08em
	\shiftright 0em
	\endprooftree

	\ 
		\prooftree
	\justifies      
	\reduceNarrow{\ConfP{\n}{\Triple{n{-}1}{1}{\zero}}}{}{\ConfP{\zero}{\Triple{{n{-}1}}{0}{\zero}}}
	\thickness=0.08em
	\shiftright 0em
	\endprooftree

	\ 
	\prooftree
	 \reduceNarrow{\Conf{\App{\nottt}{(\App{\ev}{\n})}}{\Triple{{n{-}1}}{1}{\zero}}}{}{\Conf{\overline{\bvar}}{\Triple{0}{0}{\zero}}}
	\justifies      
	\reduceNarrow{\ConfP{\CaseNat{\zero}{\tru}{\n}{\App{\nottt}{(\App{\ev}{\n})}}}{\Triple{n{-}1}{1}{\zero}}}{}{\ConfP{\bvar}{\Triple{0}{0}{\zero}}}
	\thickness=0.08em
	\shiftright 0em
	\endprooftree
\justifies
\reduceNarrow{\ConfP{\App{\ev}{\n}}{\Triple{n}{1}{\Succ{n}{\zero}}}}{}{\ConfP{\bvar}{\Triple{0}{0}{\zero}}}
\thickness=0.08em
\shiftright 0em
\endprooftree

\end{array}
\end{math}
\end{scriptsize}
\caption{Example of resource-aware evaluation: terminating recursion. (Labels of rules are omitted.)}
\label{fig:term-rec}
\end{figure}
In the top part of the figure we show the proof tree for the evaluation of $\App{\ev}{\n}$ in the base case, that is, in an environment where the value of $\n$  is \lstinline{zero}. 
In this case, both $\ev$ and $\n$ can be graded $1$, since they are used only once. In the  bottom part, we show the proof tree in an environment where the value of $\n$ is $\Succ{n}{\zero}$, for $n\neq 0$.
Here, $\bvar$ stands for either \lstinline{true} or \lstinline{false}, and $\overline{\bvar}$ for its complement.

\paragraph*{Processing data from/to files} The following example illustrates how we can simultaneously track  access  levels, as introduced in \cref{ex:ex2},  and linearity information.  Linearity grades guarantee the correct use of files, whereas  access  levels  are used to ensure that data are handled without leaking information. 
We combine the two grade algebras with the smash product of \cref{ex:gr-alg-smash}. 
So there are five grades: $0$ (meaning unused), $\comb{1}{\priv}$ and $\comb{1}{\pub}$ (meaning used linearly in either $\priv$ or $\pub$ mode), and $\comb{\omega}{\priv}$, $\comb{\omega}{\pub}$ (meaning used an arbitrary number of times in either $\priv$ or $\pub$ mode).  The partial order is graphically shown in  \cref{fig:algebra}. The neutral element for multiplication is $\comb{1}{\pub}$, which therefore will be omitted. 

\begin{figure*}
\begin{minipage}[t]{.4\textwidth}
\begin{footnotesize}
\centering
\begin{tikzpicture}
\node (OmegaPb) at (0ex,20ex) {$\comb{\omega}{\pub}$};
\node (OmegaPr) at (0ex,10ex) {$\comb{\omega}{\priv}$};
\node (OnePb) at (12ex,10ex) {$\comb{1}{\pub}$};
\node (OnePr) at (12ex,0ex) {$\comb{1}{\priv}$};
\node (Zero) at (0ex,0ex) {$0$};
\draw [arrows={-latex}] (Zero) -- (OmegaPr);
\draw [arrows={-latex}] (OnePr) -- (OnePb);
\draw [arrows={-latex}] (OnePr) -- (OmegaPr);
\draw [arrows={-latex}] (OnePb) -- (OmegaPb);
\draw [arrows={-latex}] (OmegaPr) -- (OmegaPb);
\end{tikzpicture}
\end{footnotesize}
\caption{Grade algebra of  access  levels
and linearity}\label{fig:algebra}
\end{minipage}
\hspace{1em}
\begin{minipage}[t]{.5\textwidth}
\vspace{-7.5em}
\begin{lstlisting}[mathescape=true,basicstyle=\ttfamily\footnotesize]  
Result = success$\SumT$failure
OptChar = none$\SumT$some:Char$^{\comb{\omega}{\priv}}$
fnType= Char$^{\comb{\omega}{\priv}}\rightarrow$ (ok:Char$^{\comb{\omega}{\priv}}\SumT$ error)

open:$\!\!\!$String$^{\comb{\omega}{\pub}}$  $\rightarrow$ FileHandle
read:$\!\!\!$FileHandle $\!\!\!\rightarrow\!\!\!$ (OptChar$^{\comb{\omega}{\priv}} \otimes\!\!\!$ FileHandle)
write:$\!\!\!$(Char$^{\comb{\omega}{\priv}} \otimes\!\!\!$ FileHandle) $\!\!\!\rightarrow\!\!\!$ FileHandle
close:$\!\!\!$FileHandle $\rightarrow$ Unit$^0$ 
\end{lstlisting}
\caption{Types of data and filesystem interface}\label{fig:data}
\end{minipage}
\end{figure*}
In \cref{fig:data} are the types of the data, the processing function and the functions of a filesystem interface, assuming types \lstinline{Char}, \lstinline{String}, and \lstinline{FileHandle} to be given. 
The type \Q@Result@ indicates success or failure. The type \Q@OptChar@ represents the presence or absence of a \Q@Char@ and is used in the function reading from a file;
\Q@fnType@ is the type of a function processing a private \Q@Char@ and returning either a private \Q@Char@ or \Q@error@.  The signatures of the functions of the filesystem interface specify that file handlers are used in a linear way. Hence, after opening a file and doing a number of read or write operations, \mbox{the file must be closed.} 

In the code of \cref{fig:server} we use, rather than $\MatchUnit{\e_1}{\e_2}$, the 
alternative syntax $\Seq{\e_1}{\e_2}$  mentioned in \cref{sect:calculus}.  
\begin{figure*}
\begin{lstlisting}[mathescape=true,numbersep=5pt,xleftmargin=+2em,
numbers=left,basicstyle=\ttfamily\small]
fileRW:FileHandle $\rightarrow_{\comb{\omega}{\pub}}$ FileHandle $\rightarrow$ fnType$^{\comb{\omega}{\pub}}\rightarrow$ Result$^{\comb{\omega}{\pub}}$=
rec fileRdWr.
   \inF.\outF.\fn.
      match (read inF) with <oC,inF1> ->
         match oC with none ->
            close inF1;close outF;success
         or (some c) ->
            match (fn c) with (ok c1) ->
               let outF1 = write <c1,outF> in
                  fileRdWr inF1 outF1 fn
            or error ->
               close inF1;close outF;failure
 \end{lstlisting} 
\caption{Processing data from/to files}\label{fig:server}
\end{figure*}
The function \Q@fileRW@
takes as parameters an input and an output file handler and a function that processes the character read from the input file. The result indicates whether all the characters of the input file have been successfully processed and written in the output file or there was an error in processing some character. 

The function starts by reading a character from the input file. 
If \Q@read@ returns \Q@none@, then all the characters from the input file have been read and so both files are closed and the function returns \Q@success@ (lines 5---6).  Closing the files is necessary in order to typecheck this branch of the match.
If \Q@read@ returns a character  (lines 7---12), then the processing function \Q@fn@ is applied to that character. Then, if  \Q@fn@  returns a character, then this result is written to the output file using the file handler passed as a parameter and, finally, the function  is recursively called  with the file handlers returned by the  \Q@read@ and \Q@write@ functions as arguments.  If, instead,  \Q@fn@  returns \Q@error@,  then both files are closed and the function returns \Q@failure@ (lines 11---12).  

Observe that, given the order of  \cref{fig:algebra}, we have $0\not\rord\comb{1}{\pub}$. Therefore the variables of type \Q@FileHandle@, which have grade $\comb{1}{\pub}$, must be used exactly once in every branch of the matches in their scope.

A call to \Q@fileRW@ could be: \Q@fileRW (open "inFile") (open "outFile") (rec f.\x.x)@.\\
Note that , with 
\begin{lstlisting}
  write: (Char$^{\comb{\omega}{\pub}}\otimes$ FileHandle) $\rightarrow$ FileHandle
\end{lstlisting} 
the function \Q@fileRW@ would not be well-typed, since a $\priv$ character cannot be the input of \Q@write@. On the other hand, using subsumption, we can apply  \Q@fileRW@ to a processing function  with type
\begin{lstlisting}
  Char$^{\comb{\omega}{\priv}}\rightarrow$ (ok:Char$^{\comb{\omega}{\pub}}\SumT$ error)
\end{lstlisting} 
Finally,  in the type of \Q@fileRW@, the grade of the first arrow says that the function  is recursive and it is internally used in an unrestricted way. The function could also be typed with: 
\begin{lstlisting}
  FileHandle $\rightarrow_{\comb{\omega}{\priv}}$ FileHandle $\rightarrow$ fnType$^{\comb{\omega}{\pub}}\rightarrow$ Result$^{\comb{\omega}{\priv}}$
\end{lstlisting} 
However, in this case its final result would be private and therefore less usable.


\section{No-waste semantics}\label{sect:nw-sem}
 This section is structured as follows. We first informally describe no-waste usage in \cref{sect:nw}.   In \cref{sect:gr-graph} we give some preliminary definitions and results related to \emph{grade graphs}, needed in the following sections.  In \cref{sect:nws} we define the no-waste refinement of the resource-aware semantics, and examples of reduction are shown in \cref{sect:nw-ex}. 
Finally, in \cref{sect:nw-res} we prove the related results.  Notably,  the refined semantics actually ensures no-waste (\cref{theo:nw-sem}), garbage in the environment is discardable (\cref{prop:GID}), and well-typed configurations are no-waste (\cref{theo:no-waste-typing}).

\subsection{No-waste usage}\label{sect:nw}
In the resource-aware semantics in \cref{fig:semantics}, a program cannot consume more resources than those available. That is, reduction is stuck if the current availability of a resource is not enough to cover the required usage by the program, and this situation is expected to be prevented by the type system. 

However, the converse property does not hold, that is, resources can be \emph{wasted}, even though, see \cref{rem:t-env} at the end of \cref{sect:typesystem}, the type system prevents this situation as well. For instance,  assuming linearity grades, the configuration ${\Conf{\x}{\EnvElem{\x}{1}{\unit},\EnvElem{\y}{1}{\unit}}}$, where both $\x$ and $\y$ are provided as linear, reduces (with grade $1$) to $\Conf{\unit}{\EnvElem{\y}{1}{\unit}}$, by wasting the resource $\y$.  In this section, we provide a refinement of resource-aware semantics where \emph{resources cannot be wasted}, hence the configuration above is stuck. In this way,  resource-aware soundness including no-waste can be stated and proved.

To formally express that ``resources cannot be wasted'', we will rely on the \emph{no-waste relation} $\NoWaste{\cctx}{\env}$, meaning that the resources provided by $\env$ are (transitively) used with no-waste by $\cctx$. 
The formal definition of the no-waste relation will be provided in \cref{def:nw} (\cref{sect:nws}); the following examples illustrate its expected behaviour.

\begin{example}\label{ex:nw1}
In the simple case when values in the environment are closed, as in the example above, there is no need of transitivity; for instance, $\NoWaste{\VarGrade{\x}{1},\VarGrade{\y}{1}}{\EnvElem{\x}{1}{\unit},\EnvElem{\y}{1}{\unit}}$ holds, whereas $\NoWaste{\VarGrade{\x}{1},\VarGrade{\y}{0}}{\EnvElem{\x}{1}{\unit},\EnvElem{\y}{1}{\unit}}$ does not hold,  since for linearity grades $0\not\rord 1$. 
\end{example}

On the other hand, when values have free variables, dependencies in the environment should be taken into account. To this end, as shown in \cref{fig:annotated}, the no-waste semantics is instrumented by \emph{annotating} values with a grade context ``declaring'' the \emph{direct usage} of resources in the value, so that transitive usage can be computed. 

\begin{figure}
\begin{grammatica}
\produzione{\env}{\EnvElem{\x_1}{\rgr_1}{\bv_1},\ldots,\EnvElem{\x_n}{\rgr_n}{\bv_n}}{environment}\\
\produzione{\bv}{\bvPair{\cctx}{\val}}{(annotated) value}\\
\produzione{\cctx}{\VarGrade{\x_1}{\rgr_1}, \ldots, \VarGrade{\x_n}{\rgr_n}}{grade context}
\\
\produzione{\E}{\ve\mid\e}{value expression or expression}
\end{grammatica}
\caption{Environments and (annotated) values}
\label{fig:annotated}
\end{figure}


\begin{example}\label{ex:nw2}
For instance, $\NoWaste{\VarGrade{\x}{1}}{\EnvElem{\x}{1}{\bvPair{\VarGrade{\y}{1},\VarGrade{\z}{1}}{(\Fun{\_}{\PairExp{}{\y}{\z}{}})},\EnvElem{\y}{1}{\bvPair{\emptyset}{\unit}},\EnvElem{\z}{1}{\bvPair{\emptyset}{\unit}}}}$ holds, since $\VarGrade{\x}{1}$ directly uses $\x$, and transitively $\y$ and $\z$, exactly once. Note that the relation of no-waste usage, being based on the partial order, still allows approximation which is no-waste, see \cref{rem:order}.
Hence, for instance, $\NoWaste{\VarGrade{\x}{1}}{\EnvElem{\x}{1}{\bvPair{\emptyset}{\unit}},\EnvElem{\y}{\infty}{\bvPair{\emptyset}{\unit}}}$ holds, since the resource $\y$, which is not used by the grade context, is discardable. 
\end{example}

 Set $\fv{\val}$  to be  the free variables of $\val$, defined in the standard way. 
We assume environments to be \emph{closed}, that is, for each $\EnvElem{\x}{\rgr}{\bvPair{\cctx}{\val}}\in\env$, $\fv{\val}\subseteq\dom{\env}$. 

 The grade context annotating a value is meant to be a ``guess'' about the resource consumption which would be actually triggered when the value is used.   Reduction will be non-stuck when the guess is correct (and no standard typing errors occur); as expected, the type system will check such correctness. We only  assume that, in an annotated value  $\bvPair{\cctx}{\val}$, $\cctx$ is \emph{honest} for $\val$,  that is,   it does not declare a resource consumption which could not possibly be triggered,  as formally defined below:
\begin{definition}[Honest]
A grade context $\cctx$ is honest for $\E$ if  $\rzero\not\rord\cctx(\x)$ implies $\x\in\fv{\E}$.
\end{definition}
For instance,  $\bvPair{\VarGrade{\y}{1},\VarGrade{\z}{1}}{(\Fun{\_}{\y})}$ is an ill-formed value,  since $\z$ is not a free variable of the term,  whereas  the grade context says that it is used exactly once.  The honesty condition is intuitively reasonable,  and  allows us to prove the important property that garbage,  that is, resources which are not reachable from the program,  can be discarded (\cref{prop:GID}).

The no-waste relation $\NoWaste{\cctx}{\env}$ informally introduced by the above examples will be formalized in \cref{sect:nws}. However, such formalization requires some technical definitions and results which will be provided in \cref{sect:gr-graph}. Notably, to express the property, the first step is the notion of the \emph{grade graph} $\Graph{\env}$ extracted from an environment $\env$. Intuitively, $\Graph{\env}$ expresses the
(direct) dependencies among variables in $\env$, that is,  for each pair $\x$ and $\y$, how $\y$ is used in 
(the grade contex annotating) the value associated to $\x$.
For instance, considering
the first environment of \cref{ex:nw2}, that is, $\env={\EnvElem{\x}{1}{\bvPair{\VarGrade{\y}{1},\VarGrade{\z}{1}}{(\Fun{\_}{\PairExp{}{\y}{\z}{}})},\EnvElem{\y}{1}{\bvPair{\emptyset}{\unit}},\EnvElem{\z}{1}{\bvPair{\emptyset}{\unit}}}}$,
we expect to have $\Graph{\env}(\x,\y)=\Graph{\env}(\x,\z)=1$, and $\Graph{\env}$ to give $0$ on all the other pairs.

\subsection{Grade graphs: an algebraic digression}\label{sect:gr-graph}

As anticipated, the aim of this section is to provide the technical definitions and results needed in the following  sections.  In detail:
\begin{enumerate}
\item Definition of \emph{grade graphs} (\cref{def:gr-graph}).
\item Grade graphs are a special case of \emph{grade  matrices } (\cref{def:gr-matrix}). The  \emph{reflexive and transitive closure $\Gstar{\graph}$} of an (acyclic) grade graph $\graph$ (\cref{def:closure}) is a grade matrix.
\item \emph{Fixed-point characterization} of the closure: $\Gstar\graph=\Id  \rsum \Gstar\graph\rmul \graph$ (\cref{cor:GStarId}).
\item \emph{Induction principle}: for each $\cctx$, $\cctx\ctxmul\GraphStar{}$ is the least fixed point of the transformation $\phi_{\cctx}(\cctx') = \cctx\rsum \cctx'\rmul\graph$ (\cref{prop:implication}).
\end{enumerate}

The key points where these notions will be used  are the following:
\begin{itemize}
\item The no-waste relation is defined in terms of $\Gstar{\graph}$
(\cref{def:nw}). 
\item Resource balance (\cref{theo:balance}) relies on the fixed point characterization, and the  result that well-typed configurations and results are no-waste (\cref{theo:no-waste-typing}) relies on the induction principle.
\end{itemize}

\subsubsection{Grade graphs} Recall that we denote by $\Vars$ the set of variables and that, for a function $\fun{f}{X}{\RC\RR}$, the support of $f$ is the set 
$\supp{f} = \{ x \in X \mid f(x) \ne \rzero \}$. 

\begin{definition}[Grade graph]  \label{def:gr-graph} 
A \emph{grade graph} is a function $\fun{\graph}{\Vars \times \Vars}{\RC\RR}$ with finite support. 
\end{definition}
The finite support condition ensures that 
$\graph(x,y)\ne\rzero$ for finitely many pairs $\ple{x,y}\in\Vars\times\Vars$. 
As the name suggests, 
these functions can be represented by finite graphs where nodes are variables and edges are labeled by grades, 
analogously to functions from variables into grades with finite support which are represented by grade contexts. 
More precisely, given a grade graph $\graph$, 
we have an edge from $\x$ to $y$ with label $\rgr$ exactly when $\graph(\x,\y) = \rgr$. 
Intuitively, such an edge means that $\x$ depends on ``$\rgr$ copies'' of $\y$; for $\rgr=\rzero$, $\x$ does \emph{not} depend on $\y$, and the edge can be omitted.  
Finally, again following graph terminology, we will call \emph{source} in $\graph$ a node (variable) $\y$ such that $\graph(\x,\y)=\rzero$ for all $\x$.
We will write $\Nodes\graph$ for the set of \emph{non-source} nodes, i.e., those with at least one incoming edge: 
\[ \Nodes\graph = \{ \y\in\Vars \mid  \graph(\x,\y)\ne\rzero \text{ for some } \x \in \Vars \} \]
Since $\graph$ has finite support, the set $\Nodes\graph$ is clearly finite, as its   cardinality is bounded \mbox{by that of $\supp\graph$. }

\begin{example}\label{ex:gradeGraph}
Consider the grade graph $\graph$ with support $\{\x,\y,\z\}$ defined by:
\begin{quoting}
 $\graph(\x,\y)=2\quad\graph(\x,\z)=1\quad\graph(\y,\x)=0\quad\graph(\y,\z)=1\quad\graph(\z,\x)=\graph(\z,\y)=0 $
\end{quoting}
We have $\Nodes\graph=\{\y,\z\}$. This grade graph could correspond, for instance, to the environment:
\begin{quoting} 
$\EnvElem{\x}{\_}{\bvPair{\VarGrade{\y}{2},\VarGrade{\z}{1}}{(\Fun{\_}{\PairExp{}{\y}{\z}{}})}},\EnvElem{\y}{\_}{\bvPair{\VarGrade{\z}{1}}{(\Fun{\_}{\z})}},\EnvElem{\z}{\_}{\bvPair{\emptyset}{\unit}}$
\end{quoting}
where we left the grades unspecified as they are not relevant.
\end{example}

In order to model transitive dependencies, we have to consider paths in a grade graph. 
The next definition introduces path-related notions, adapting them from the standard graph-theoretic ones. 

\begin{definition}[Paths, cycles and acyclic grade graphs] \label{def:paths} 
A \emph{path} is a sequence $\p=\x_1\ldots \x_n$, with $n\geq 1$; $\x_1, \x_n$ are denoted $\Start{\p}$ and $\End{\p}$, respectively.
We set $\length\p = n-1$ the \emph{length} of the path $\p$.\footnote{Note that the length counts the number of edges  which is the number of nodes minus $1$. } 
A path $\x_1\ldots \x_n$ is a \emph{cycle} if $n>1$ and $\x_1 = \x_n$. 
A grade graph $\graph$ induces a function $\fun{\graph}{\Vars^+}{\RR}$ inductively defined as follows: 
\begin{quoting}
$\weight{\graph}{\x} = \rone$\\
$\weight{\graph}{\x\cdot\p}=\graph(\x,\Start{\p})\rmul\weight{\graph}{\p}$
\end{quoting}
Then, $\graph$ is \emph{acyclic} if, for all cycles $\p$,  $\weight{\graph}{\p}=\rzero$. 
\end{definition}

\begin{example}\label{ex:paths}
For the grade graph of \cref{ex:gradeGraph} the paths having non zero grade are:
\begin{quoting}
 $\graph(\x)=\graph(\y)=\graph(\z)=1\quad\graph(\x\y)=2\quad\graph(\x\z)=1\quad\graph(\x\y\z)=2\quad\graph(\y\z)=1$
\end{quoting}
As expected, $\graph$ is acyclic.
\end{example}

We will write $\Paths\x\y$ for the set of paths $\p$ such that 
$\Start\p = \x$ and $\End\p = \y$. The following are immediate consequences of the above definitions. 

\begin{proposition}\label{prop:paths}
Let $\graph$ be a grade graph, $\x,\y\in\Vars$. 
Set 
\mbox{$\Pathsnz\graph\x\y {=} \{ \p\in\Paths\x\y \mid \weight\graph\p {\ne}\rzero \}$. }
\begin{enumerate}
\item\label{prop:paths:0}
If $\p = \x_1\cdots\x_{n+1}\in\Paths\x\y$ with $n\geq 1$ and there is $j \in 2..n+1$ such that $\x_j\notin\Nodes\graph$, then $\weight\graph\p = \rzero$. 
\item\label{prop:paths:1}
If $\x\ne\y$ and $\y\notin\Nodes\graph$, then $\Pathsnz\graph\x\y = \emptyset$. 
\item\label{prop:paths:02}
If $\graph$ is acyclic and $n$ is the cardinality of $\Nodes\graph$, then 
$\length\p \geq n+1$ implies $\weight\graph\p = \rzero$. 
\item\label{prop:paths:2}
If $\graph$ is acyclic, then the set 
$\Pathsnz\graph\x\y$ is finite. 
\end{enumerate}
\end{proposition}

\subsubsection{Reflexive and transitive closure}
For acyclic grade graphs we can define the reflexive and transitive closure, as detailed below. 

\begin{definition}[Reflexive and transitive closure of a grade graph]\label{def:closure}
Let $\graph$ be an acyclic grade graph. 
Its \emph{reflexive and transitive closure}, denoted $\Gstar{\graph}$, is defined as follows: 
\begin{quoting}
for all $\x,\y\in\Vars$, 
$\displaystyle \Gstar{\graph}(\x,\y) = \sum_{\p\in\Paths\x\y} \weight{\graph}{\p}$
\end{quoting}
\end{definition}  
The definition of $\Gstar{\graph}(x,y)$, for each pair of nodes $x,y$, requires a summation over an infinite  set. 
However, \refItem{prop:paths}{2} ensures that only finitely many addends are different from $\rzero$, hence the summation is well-defined.

\begin{example}\label{ex:graphStar}
Given $\graph$ of \cref{ex:gradeGraph} and \cref{ex:paths}, $\Gstar{\graph}$ is:
\begin{quoting}
$
\begin{array}{c}
\Gstar{\graph}(\x,\x)=\Gstar{\graph}(\y,\y)=\Gstar{\graph}(\z,\z)=1\\
\begin{array}{lllll}
 \Gstar{\graph}(\x,\y)& = &\graph(\x\y)+\graph(\x\z\y)& =& 2+0=2\\
  \Gstar{\graph}(\x,\z)& = & \graph(\x\z)+\graph(\x\y\z)& =& 1+2=3\\
    \Gstar{\graph}(\y,\x)& = & \graph(\y\x)+\graph(\y\z\x)& =& 0+0=0\\
    \Gstar{\graph}(\y,\z)& = & \graph(\y\z)+\graph(\y\x\z)& =& 1+0=1\\
    \Gstar{\graph}(\z,\x)& = & \graph(\z\x)+\graph(\z\y\x)& =& 0+0=0\\
    \Gstar{\graph}(\z,\y)& = & \graph(\z\y)+\graph(\z\x\y)& =& 0+0=0\\
\end{array}
\end{array}
$
\end{quoting}
Since $\graph$ is acyclic, by \refItem{prop:paths}{02} we can consider only paths of length strictly less than 4.
\end{example}

Note that the reflexive and transitive closure of a grade graph $\graph$ is \emph{neither acyclic, nor a grade graph}. 
Indeed, we have $\Gstar\graph(\x,\x) = \rone$ for all $\x\in\Vars$, hence 
cycles of the form $\x\cdot\x$ have a non-zero grade  and $\supp{\Gstar\graph}$ is infinite. 

 However, $\Gstar\graph$ is an instance of a more general notion, which we call \emph{grade matrix}. 
Whereas for grade graphs we require a function $\fun\mtx{\Vars\times\Vars}{\RC\RR}$ to have finite support (in the graph view, finitely many non-isolated nodes), for grade matrices we require a weaker constraint, namely, 
that every row has finite support (every node has finitely many non-zero outgoing edges), as formalized below. 

Let $\fun\mtx{\Vars\times\Vars}{\RC\RR}$ be a function. 
For every variable $\x\in\Vars$, we define the function $\fun{\Row\mtx\x}{\Vars}{\RC\RR}$ as follows: 
$\Row\mtx\x(\y) = \mtx(\x,\y)$. 

\begin{definition}[Grade matrix]  \label{def:gr-matrix} 
A \emph{grade matrix} is a function $\fun\mtx{\Vars\times\Vars}{\RC\RR}$ such that, for every variable $\x\in\Vars$, 
$\Row\mtx\x$ has a finite support, i.e., it is a grade context. 
\end{definition}
  The grade context $\Row\graph\x$ 
is the row expressing the direct dependencies of $\x$.
Note that grade contexts turn out to be grade  matrices  with only one row, and 
grade graphs are grade  matrices  with a finite number of non-zero entries.
As expected, $\Gstar\graph$ is a grade matrix. 
Indeed, thanks to \refItem{prop:paths}{1}, for all variables $\x\in\Vars$, we have 
$\supp{\Row{\Gstar\graph}{\x}} \subseteq \Nodes\graph\cup\{\x\}$, which thus is finite. 

Grade matrices support the standard operations of \emph{addition} and \emph{(row-by-column) multiplication}, as formally detailed below.

Addition is the pointwise extension of the addition on grades. It is immediate to see that the result is still a grade matrix, and that on grade contexts we get the addition as previously defined.

Row-by-column multiplication is a summation of multiplications.  For finite  matrices,   this summation has a finite number of addends. We show that this still holds for  matrices  where each row has a finite number of non-zero entries, leading, in turn, to a grade matrix. To this end,
consider first the multiplication of a grade matrix $\mtx$ by a single row, i.e., a grade context $\cctx$,  on the left, giving the grade context $\cctx\ctxmul\mtx$ defined by\footnote{Vector-matrix multiplication is also used by \cite{VollmerMEO25}.}:
\begin{quoting}
 for all $\x\in\Vars$, $\displaystyle (\cctx\rmul\mtx)(\x) = \sum_{\y\in\Vars} \cctx(\y)\cdot\mtx(\y,\x)$  
\end{quoting}

The infinite summation above has only finitely many non-zero addends, as $\cctx$ has a finite support, hence it is well-defined. 
Moreover, it is easy to see that
$\cctx\rmul\mtx = \sum_{\x\in\supp\cctx} \cctx(\x)\rmul\Row\mtx\x$, hence it is a well-defined grade context. 

Now, given the grade matrices $\mtx_1,\mtx_2$, we set
\begin{quoting}
for all $\x,\y\in\Vars$, 
$\displaystyle (\mtx_1\cdot\mtx_2)(\x,\y) = (\Row{{\mtx_1}}{\x}\rmul \mtx_2)(\y) = \sum_{z\in\Vars} \mtx_1(\x,z)\rmul\mtx_2(z,\y)$ 
\end{quoting} 
This is obviously well-defined because  we have 
$\Row{(\mtx_1\rmul\mtx_2)}{\x} = \Row{{\mtx_1}}{\x}\rmul\mtx_2$ which has a finite support, being a grade context. 
It is easy to see that this operation is associative and monotone and has as neutral element the identity grade matrix $\Id$ defined by 
\begin{quoting}
 $\Id(\x,\y)=
 \begin{cases}
 \rone&\mbox{if}\ \x=\y\\
 \rzero&\mbox{otherwise} 
 \end{cases}
 $
\end{quoting}

Altogether, taking as ordering the pointwise extension of the ordering on grades, grade matrices with addition carry an ordered commutative monoid structure, and with multiplication another ordered monoid structure.

\subsubsection{Fixed point characterization and induction principle}

 
The main theorems on grade graphs (the fixed point characterization of the closure and the induction principle) rely on the fact that the closure can be expressed by a finite summation (\cref{lem:finitestar}).
To prove this we need some preliminary results.

 The following lemma states that, being the graph acyclic, multiplying a row with the matrix expressing transitive dependencies we get zero on the component corresponding to the variable $\x$ itself. 
Equivalently, the matrix obtained by multiplying the graph $\graph$ with the matrix of transitive dependencies has $\rzero$ on the diagonal. 

\begin{lemma}\label{lem:nocycles}
If $\graph$ is acyclic, then 
$(\graph\rmul\Gstar\graph)(\x,\x) = (\Row{\graph}{\x}\rmul\Gstar{\graph})(\x)=\rzero$.
\end{lemma}

Intuitively, this happens because  the grade matrix $\graph\rmul\Gstar\graph$  considers transitive dependencies   through non-empty paths and,  since the graph is acyclic, there are no non-empty paths from a variable to itself.  

Given a grade graph $\graph$, in the following we will write 
$\graph^n$  to denote the multiplication of $\graph$ with itself $n$ times. 
Moreover, we will denote by $\Paths[n]\x\y$ the subset of $\Paths\x\y$ of those paths of length $n$. 

\begin{proposition}\label{prop:graph-iterate}
Let $\graph$ be a grade graph. 
For all $n\in\N$ and $\x,\y\in\Vars$, we have 
$\graph^n(\x,\y) = \sum_{\p\in\Paths[n]\x\y} \weight\graph\p$. 
\end{proposition}
Combining \cref{prop:graph-iterate} and \refItem{prop:paths}{02}, 
it is immediate to see that, if $n$ is the cardinality of $\Nodes\graph$, then 
$\graph^k$ is the zero grade matrix, for all $k \geq n+1$. 
As a consequence we get the following result. 

\begin{lemma}\label{lem:finitestar}
Let $\graph$ be an acyclic grade graph and let $n$ be the cardinality of $\Nodes\graph$. 
Then the following holds: 
\[ \Gstar\graph = \sum_{i=0}^{n}\graph^{i} \] 
\end{lemma}  
\begin{proof}
Using \cref{prop:graph-iterate}, we get the following 
\[
\Gstar\graph(\x,\y) 
  = \sum_{\p\in\Paths\x\y)} \weight\Graph\p 
  = \sum_{k\in\N} \sum_{\p\in\Paths[k]\x\y} \weight\graph\p 
  = \sum_{k\in\N} \graph^k(\x,\y) 
\]
As already noticed, by \refItem{prop:paths}{02}, we have that 
$\graph^k(\x,\y) = \rzero$ for all $k \geq n+1$, hence we have the thesis. 
\end{proof}

Considering the grade graph $\graph$ of \cref{ex:gradeGraph},  $\graph^0=\Id$, $\graph^1=\graph$
and in $\graph^2=\graph\rmul\graph$ the only non zero grade is 
$\graph^2(\x,\z)=2$. Indeed, as seen in \cref{ex:paths}, the only path of length 2 with non zero grade is $\x\y\z$ and 
$\graph(\x\y\z)=2$. Finally,  $\Gstar{\graph}=\graph^0+\graph^1+\graph^2$ coincides with
$\Gstar{\graph}$ as described in \cref{ex:graphStar}.

 We can now prove the main results of this section. The first one states that the reflexive and transitive closure behaves similarly to the Kleene closure, by satisfying an analogous equation.

\begin{corollary} \label{cor:GStarId} 
Let $\graph$ be an acyclic grade graph. 
Then, $\Gstar\graph = \Id  \rsum \Gstar\graph\rmul \graph = \Id \rsum \graph\rmul\Gstar\graph$. 
\end{corollary}
\begin{proof}
Let $n$ be the cardinality of $\Nodes\graph$. 
The thesis is immediate by \cref{lem:finitestar}: 
\[
\Id \rsum \Gstar\graph\rmul\graph 
  = \graph^0 \rsum \sum_{k\in 0..n} \graph^{k+1}
  = \sum_{k\in0..n} \graph^k \rsum \graph^{n+1}
  = \Gstar\graph 
\]
because $\graph^{n+1}$ is the zero grade matrix. 
\end{proof}

 The second main result is the following induction principle, allowing to show that an upper bound holds for $\cctx\ctxmul\GraphStar{}$ by proving that it holds for $\cctx\ctxsum(\cctx'\ctxmul\graph)$.

\begin{proposition}\label{prop:implication}
Given an acyclic graph $\graph$, and two grade contexts $\cctx,\cctx'$:
\begin{quoting}
   $\cctx\ctxsum(\cctx'\ctxmul\graph)\ctxord\cctx'$ implies $\cctx\ctxmul\GraphStar{}\ctxord\cctx'$.
\end{quoting}
\end{proposition}
\begin{proof} 
By a straightforward induction on $k\in\N$, we can prove that 
$\cctx_1\rmul \left(\sum_{i\in 0..k} \graph^k\right) \rsum \cctx_2\rmul\graph^{k+1} \rord \cctx_2$, 
for all $k \in \N$. 
Taking $n$ to be the cardinality of $\Nodes\graph$, by \cref{lem:finitestar},  we get the thesis because 
$\Gstar\graph = \sum_{i\in 0..n} \graph^i$ and $\graph^{n+1}$ is the zero matrix. 
\end{proof} 

 Note that the above property means that,  for every grade context $\cctx$, the monotone function  on grade contexts 
$\phi_{\cctx}(\cctx') = \cctx\rsum \cctx'\rmul\graph$ 
has a least pre-fixed point given by $\cctx\rmul\Gstar\graph$. 
Actually, this is also a fixed point thanks to \cref{cor:GStarId}.

\subsection{No-waste semantics}\label{sect:nws}

In this section we provide the no-waste refinement of the resource-aware semantics. First of all, we define the grade context and grade graph extracted from an environment. 
Recall from \cref{fig:annotated} that in the no-waste setting values in the environment are annotated with a grade context. 
\begin{definition}\label{def:extract-graph}
Given an environment $\env=\EnvElem{\x_1}{\rgr_1}{\bvPair{\cctx_1}{\val_1}},\ldots,\EnvElem{\x_n}{\rgr_n}{\bvPair{\cctx_n}{\val_n}}$
\begin{itemize}
\item the grade context $\CCTX{\env}$ is $\VarGrade{\x_1}{\rgr_1}, \ldots, \VarGrade{\x_n}{\rgr_n}$
\item the grade graph $\Graph{\env}$ has support $\{\x_1, \ldots, \x_n\}$ and is defined  by $\Row{(\Graph{\env})}{{\x_i}}=\cctx_i$, for $i\in 1..n$.
 \end{itemize}
\end{definition}
Intuitively, the grade context $\CCTX{\env}$ expresses the availability of each resource (variable) in the environment, whereas the grade graph $\Graph{\env}$ collects the direct dependencies among such variables.

In the following, we will assume acyclic environments, in the sense that there is no dependency path from a variable to itself. 
 Formally, this means that we assume that the grade graph $\Graph\env$ is acyclic (see \cref{def:paths}).  
It is important to note this \emph{does not} prevent provided resources to be recursive; indeed, this is handled in the semantics by adding a resource, under a fresh name, for each recursive call,  as seen in rule \refToRule{app} in \cref{fig:semantics} and the example in \cref{fig:term-rec}.  
We discuss  at the end of the section  whether allowing cyclic environments rather than modeling recursion, essentially, by unfolding, makes some relevant difference, and how to adjust the technical treatment. 

As the reader can expect, the reflexive and transitive closure $\GraphStar{\env}$ of the grade graph extracted from an environment $\env$ expresses the \emph{(full) dependencies} among resources in $\dom{\env}$. Indeed, for each $\x,\y\in\dom{\env}$, a dependency path from $\x$ to $\y$ expresses \emph{one} transitive (graded) usage of $\y$ from (the value associated to) $\x$; thus, the sum of  the weights of  such paths represents the whole usage of $\y$ from $\x$.

We can finally define the relation between grade contexts and environments modeling no-waste usage. 

\begin{definition}[No-waste usage]\label{def:nw}
The  grade  context $\cctx$ \emph{(transitively) uses $\env$ with no waste}, written $\NoWaste{\cctx}{\env}$, if  $\cctx\rmul{\GraphStar{\env}}\rord\CCTX{\env}$.
\end{definition}

The relation $\NoWaste{\cctx}{\env}$ says that the resources provided by $\env$ are (transitively) used with no-waste by $\cctx$. 
Indeed,  for each resource $\x$, $(\cctx\rmul{\GraphStar{\env}})(\x)$ gives the (transitive) usage of $\x$ from $\cctx$, obtained as the sum of the usages of $\x$ from any variable $\y$, each one multiplied by the direct usage of $\x$ in $\cctx$. On the other hand, $\CCTX{\env}(\x)$ gives the availability of $\x$ in $\env$.
As previously illustrated in \cref{ex:nw1,ex:nw2}, the relation of no-waste usage only allows the approximation given by the partial order; in cases where the partial order is the equality, as in the grade algebra of exact counting defined in \refItem{ex:gr-alg}{nat}, usage and availability are required to coincide.

For example,  consider  the environment 
\begin{quoting} 
$\env=\EnvElem{\x}{1}{\bvPair{\VarGrade{\y}{2},\VarGrade{\z}{1}}{(\Fun{\_}{\PairExp{}{\y}{\z}{}})}},\EnvElem{\y}{4}{\bvPair{\VarGrade{\z}{1}}{(\Fun{\_}{\z})}},\EnvElem{\z}{5}{\bvPair{\emptyset}{\unit}}$
\end{quoting}
whose grade graph $\Graph{\env}$  is defined in \cref{ex:gradeGraph} with reflexive and transitive closure $\GraphStar{\env}$  presented in \cref{ex:graphStar}. The grade context $\cctx=\VarGrade{\x}{1},\VarGrade{\y}{2}$ uses $\env$ with no waste, since $\cctx\rmul{\GraphStar{\env}}=\CCTX{\env}$. Indeed 
\begin{itemize}
\item $\x$ is used once only in $\cctx$
\item $\y$ is used twice in $\cctx$ and twice through $\x$
\item $\z$ is not used in $\cctx$ but it is used through $\x$ three times, one directly and two by the use of $\y$, and twice by $\y$ since each use of $\y$ in $\cctx$ uses $\z$ once.
\end{itemize}

We can now formally define the no-waste semantics. The reduction judgement for expressions has shape $\NWreduce{\Conf{\e}{\env}}{\rgr}{\Conf{\bv}{\env'}}$, and is defined in \cref{fig:nw-sem} on top of an analogous one for value expressions. As noted for the previous semantics in \cref{fig:semantics}, for simplicity we use the same notation for the two judgments, and in the bottom section of \cref{fig:nw-sem} we write both judgments as premises. 

In the reduction rules, $\env_1\envSum\env_2$ denotes the \emph{sum} of environments, defined as follows, analogously to the sum of contexts in \cref{sect:typesystem}:
\begin{align*} 
\emptyset \ctxsum \env &= \env \\ 
(\EnvElem{\x}{\sgr}{\bv}, \env_1) \envSum ( \EnvElem{\x}{\rgr}{\bv}, \env_2) &= \EnvElem{\x}{\sgr \rsum \rgr}{\bv}, (\env_1 \envSum\env_2) \\ 
(\EnvElem{\x}{\sgr}{\bv}, \env_1) \envSum \env_2 &= \EnvElem{\x}{\sgr}{\bv}, (\env_1 \envSum \env_2)  
  &\text{if $\x \notin \dom{\env_2}$} 
\end{align*} 
We use an analogous notation for grade contexts. Note that,  similarly to the sum of  contexts, sum of environments is partial, since a variable in the
domain of both is required to have \mbox{the same associated value. }
\begin{figure}
\begin{small}
\begin{math}
\begin{array}{l}
\NamedRule{var}{ }
{ \red{\x}{\env,\EnvElem{\x}{\sgr}{\bv}}{\rgr}{\bv}{\env,\EnvElem{\x}{\sgr'}{\bv}} }
{\bv=\bvPair{\cctx}{\val}\\
\rgr\rsum\sgr'\rord \sgr\\
\NoWaste{\rgr\rmul\cctx}{\env,\EnvElem{\x}{\sgr'}{\bv}}
}
\\[5ex]
\NamedRule{fun}{ }
{ \red{\RecFun\f\x\e}{\env}{\rgr}{\bvPair{\cctx}{(\RecFun{\f}{\x}{\e})}}{\env}}
{
\NoWaste{\rgr\rmul\cctx}{\env}
}
\
\NamedRule{unit}{ }
{ \red{\unit}{\env}{\rgr}{\bvPair{\cctx}{\unit}}{\env} } 
{\NoWaste{\rgr\rmul\cctx}{\env}
} 
\\[4ex]
\NamedRule{pair}{
  \red{\ve_1}{\env_1}{\rgr\rmul\rgr_1}{\bvPair{\cctx_1}{\val_1}}{\env_1'}
  \Space
  \red{\ve_2}{\env_2}{\rgr\rmul\rgr_2}{\bvPair{\cctx_2}{\val_2}}{\env_2'} 
}{ \red{\PairExp{\rgr_1}{\ve_1}{\ve_2}{\rgr_2}}{\env}{\rgr}{\bvPair{\cctx}{\PairExp{\rgr_1}{\val_1}{\val_2}{\rgr_2}}}{\env_1'\envSum\env_2'} } 
{\cctx=\rgr_1\rmul\cctx_1\ctxsum\rgr_2\rmul\cctx_2\\
\env_1\envSum\env_2\ctxord\env
}
\\[4ex]
\NamedRule{inl}{
  \red{\ve}{\env}{\rgr\rmul\sgr}{\bvPair{\cctx}{\val}}{\env'} 
}{ \red{\Inl{\rgr}{\ve}}{\env}{\rgr}{\bvPair{\sgr\rmul\cctx}{(\Inl{\rgr}{\val})}}{\env'} }
{} 
\BigSpace
\NamedRule{inr}{
  \red{\ve}{\env}{\rgr\rmul\sgr}{\bvPair{\cctx}{\val}}{\env'} 
}{ \red{\Inr{\rgr}{\ve}}{\env}{\rgr}{\bvPair{\sgr\rmul\cctx}{(\Inr{\rgr}{\val})}}{\env'} }
{} 
\\[4ex]
\end{array}
\end{math}

\hrule

\begin{math}
\begin{array}{l}
\\
\NamedRule{ret}{\red{\ve}{\env}{\sgr}{\bv}{\env'} }
{ \NWreduce{\Conf{\Return\ve}{\env}}{\rgr}{\Conf{\bv}{\env'}} }
{ 
  \rgr\rord\sgr\ne\rzero 
} 
\\[4ex]
\NamedRule{let}{
  \begin{array}{l}
    \NWreduce{\Conf{\e_1}{\env_1}}{\sgr}{\Conf{\bv_1}{\env_1'}} \\ 
    \NWreduce{\Conf{\Subst{\e_2}{\x'}{\x}}{\AddToEnv{(\env_1'\envSum\env_2)}{\x'}{\sgr}{\bv_1}}}{\rgr}{\Conf{\bv}{\env'}}
  \end{array}
}{ \NWreduce{\Conf{\Let{\x}{\e_1}{\e_2}}{\env}}{\rgr}{\Conf{\bv}{\env'}} }
{\env_1\envSum\env_2\ctxord\env\\
\x'\ \mbox{fresh}
} 

\\[4ex]

\NamedRule{app}{
\begin{array}{l}
\NWreduce{\Conf{\ve_1}{\env_1}}{\sgr}{\Conf{\bvPair{\cctx_1}{\RecFun{\f}{\x}{\e}}}{\env_1'}}\Space
\NWreduce{\Conf{\ve_2}{\env_2}}{\tgr}{\Conf{\bv_2}{\env_2'}}\\
\NWreduce{\Conf{\Subst{\Subst{\e}{\x'}{\x}}{\f'}{\f}}{\AddToEnv{\AddToEnv{(\env_1'\envSum\env_2')}{\x'}{\tgr}{\bv_2}}{\f'}{\sgr_2}{\bvPair{\cctx_1}{\RecFun{\f}{\x}{\e}}}}}{\rgr}{\Conf{\bv}{\env'}}
\end{array}}
{
\NWreduce{\Conf{\App{\ve_1}{\ve_2}}{\env}}{\rgr}{\Conf{\bv}{\env'}}
}{
\sgr_1 \rsum \sgr_2 \rord \sgr\\
\sgr_1 \neq \rzero\\
\env_1\envSum\env_2\ctxord\env\\
  \f',\x'\ \mbox{fresh} 
}
\\[4ex]
\NamedRule{match-u}{
\NWreduce{\Conf{\ve}{\env_1}}{\sgr}{\Conf{\bvPair{\cctx}{\unit}}{\env_1'}}\Space
\NWreduce{\Conf{\e}{(\env_1'\envSum\env_2)}}{\rgr}{\Conf{\bv}{\env'}}
}
{
\NWreduce{\Conf{\MatchUnit{\ve}{\e}}{\env}}{\rgr}{\Conf{\bv}{\env'}}
}
{
  \sgr\ne\rzero \\
\env_1\envSum\env_2\ctxord\env
}
\\[4ex]
\NamedRule{match-p}{
\begin{array}{l}
\NWreduce{\Conf{\ve}{\env_1}}{\sgr}{\Conf{\bvPair{\cctx}{\PairExp{\rgr_1}{\val_1}{\val_2}{\rgr_2}}}{\env_1'}} \\
\NWreduce{\Conf{\Subst{\Subst{\e}{\x'}{\x}}{\y'}{\y}}{\AddToEnv{\AddToEnv{\env}{\x'}{\sgr\rmul\rgr_1}{\bvPair{\cctx_1}{\val_1}}}{\y'}{\sgr\rmul\rgr_2}{\bvPair{\cctx_2}{\val_2}}}}{\rgr}{\Conf{\bv}{\env'}}
\end{array}
}
{
\NWreduce{\Conf{\Match{\x}{\y}{\ve}{\e}}{\env}}{\rgr}{\Conf{\bv}{\env'}}
}
{  \sgr\ne\rzero \\
\cctx\ctxord\rgr_1\rmul\cctx_1\ctxsum\rgr_2\rmul\cctx_2\\
\env_1\envSum\env_2\ctxord\env\\
\x,\y\ \mbox{fresh}
}
\\[5ex]
\NamedRule{match-l}{
\begin{array}{l}
\NWreduce{\Conf{\ve}{\env_1}}{\tgr}{\Conf{\bvPair{\cctx}{(\Inl{\sgr}{\val_1})}}{\env_1'}} \\
\NWreduce{\Conf{\Subst{\e_1}{\y}{\x_1}}{\AddToEnv{(\env_1'\envSum\env_2)}{\y}{\sgr\rmul\rgr}{\bvPair{\cctx_1}{\val_1}}}}{\rgr}{\Conf{\bv}{\env'}}
\end{array}
}
{
\NWreduce{\Conf{\Case{\ve}{\x_1}{\e_1}{\x_2}{\e_2}}{\env}}{\tgr}{\Conf{\bv}{\env'}}
}
{
  \tgr\ne\rzero \\
\cctx\ctxord\sgr\ctxmul\cctx_1\\
\env_1\envSum\env_2\ctxord\env\\
\y\ \mbox{fresh}
}
\\[4ex]
\NamedRule{match-r}{
\begin{array}{l}
\NWreduce{\Conf{\ve}{\env_1}}{\sgr}{\Conf{\bvPair{\cctx}{(\Inr{\sgr}{\val_1})}}{\env_1'}} \\
\NWreduce{\Conf{\Subst{\e_2}{\y}{\x_2}}{\AddToEnv{(\env_1'\envSum\env_2)}{\y}{\sgr\rmul\rgr}{\bvPair{\cctx_1}{\val_1}}}}{\tgr}{\Conf{\bv}{\env'}}
\end{array}
}
{
\NWreduce{\Conf{\Case{\ve}{\x_1}{\e_1}{\x_2}{\e_2}}{\env}}{\tgr}{\Conf{\bv}{\env'}}
}
{
  \tgr\ne\rzero \\
\cctx\ctxord\sgr\ctxmul\cctx_1\\
\env_1\envSum\env_2\ctxord\env\\
\y\ \mbox{fresh}
}
\end{array}
\end{math}
\end{small}
\caption{No-waste semantics}
\label{fig:nw-sem}
\end{figure}

The distinguishing features with respect to the previous resource-aware semantics are the following:

\begin{enumerate}
\item Configurations obtained as results, of shape $\Conf{\bvPair{\cctx}{\val}}{\env}$, are enforced to be \emph{no-waste}, that is, the resources left in $\env$ should be only those needed (modulo approximation) by the (grade context $\cctx$ annotating the) value $\val$.
\item When reducing a term with more than one subterm, the available resources are \emph{split} among subterms, analogously to what happens in resource-aware type systems. 
\end{enumerate}
In this way, a configuration $\Conf{\E}{\env}$ can be successfully reduced only if it is possible to inductively split $\env$, until reaching the leaves of the proof tree, so that, in every leaf, the the part of environment occurring in every leaf is used with no waste, that is, a no-waste result is obtained. In this way, altogether the whole $\env$ is used with no waste. 

Examples illustrating when successful reduction is possibile and when it is not are given in \cref{sect:nw-ex}. 

Rules which are axioms, that is, \refToRule{var}, \refToRule{fun}, and \refToRule{unit}, are designed to accomplish (1). The most significant case is \refToRule{var}, where a variable occurrence is replaced by the associated value.

As in the previous rule in \cref{fig:semantics}, its current amount should be enough to cover that required by the reduction grade (second side condition).  However, there is the additional side condition that the result should be no-waste; since it is produced in $\rgr$ copies, this condition should hold for a grade context obtained by multiplying by $\rgr$ that declared in the annotated value. 
 Note that this condition implies that  the residual amount $\sgr'$ \emph{should be discardable}, that is, $\rzero\rord\sgr'$. Indeed, set $\env'=\env,\EnvElem{\x}{\sgr'}{\bv}$, we know that $(\rgr\rmul\cctx)\rmul{\GraphStar{\env'}}\rord\CCTX{\env'}$, hence, in particular $((\rgr\rmul\cctx)\rmul{\GraphStar{\env'}})(\x)\rord\sgr'$. Since $\cctx=\Row{\graph}{\x}$, by \cref{lem:nocycles} we get $\cctx\rmul{\GraphStar{\env'}}(\x)=\rzero$, hence $(\rgr\rmul\cctx)\rmul{\GraphStar{\env'}}(\x)=\rzero$ as well.

With this formulation, a resource can be used only if its current amount is enough to cover the required use, and no resource amount would be wasted. 
An alternative formulation could even \emph{remove garbage}, see \cref{prop:GID} in the following. 

The side condition requiring the result to be no-waste is also added to the other axioms \refToRule{fun} and \refToRule{unit}.  Since no resource is consumed, the condition is checked on the initial environment. 

 Rules handling constructs with more than one subterm, that is,  \refToRule{pair}, \refToRule{let}, \refToRule{app}, and all the \refToRule{match} rules, are designed to accomplish (2).
For instance, in rule \refToRule{pair}, the environment, rather than  being  sequentially used by the two components as in the previous rule in  in \cref{fig:semantics},  is split in two parts, used to evaluate the first and the second component, respectively. 
A similar splitting is done in the other mentioned rules.

Besides the standard typing errors, and resource exhaustion, as in the previous resource-aware semantics, reduction can get stuck if  some available amount would be wasted. For instance, considering linearity grades, the configuration $\Conf{\x}{\EnvElem{\x}{1}{\bvPair{\emptyset}{\unit}},\EnvElem{\y}{1}{\bvPair{\emptyset}{\unit}}}$ is stuck. Indeed, rule \refToRule{var} is not applicable since there is no way to obtain a no-waste result. Formally, since $\cctx=\emptyset$, it should be $\NoWaste{\emptyset}{\EnvElem{\y}{1}{\bvPair{\emptyset}{\unit}}}$; this does not hold, since, for the resource $\y$, the usage is $0$, the availability is $1$, and $0\not\leq 1$. On the other hand, reduction would be possibile with an environment $\EnvElem{\x}{1}{\bvPair{\emptyset}{\unit}},\EnvElem{\y}{\omega}{\bvPair{\emptyset}{\unit}}$, and the final result would be $\Conf{\bvPair{\emptyset}{\unit}}{\EnvElem{\y}{\omega}{\bvPair{\emptyset}{\unit}}}$, or even $\Conf{\bvPair{\emptyset}{\unit}}{\emptyset}$ in a semantics removing garbage.

When the environment is split in two parts,  there are generally many possible ways to do the splitting; depending on such choices, when reaching the base cases, the corresponding axiom could either be applicable or not.  This is an additional source of non-determinism besides those already present in the previous semantics. 
Hence, as already mentioned, soundness of the type system will be \emph{soundness-may}\footnote{The terminology of \emph{may} and \emph{must} properties is very general and comes originally from \cite{DeNicolaH84}; the specific names soundness-may and soundness-must  were  introduced  by  \cite{DagninoBZD20,Dagnino22} in the context of  big-step semantics.}, meaning that, for a well-typed configuration, \emph{there exists} a computation which does not get stuck; in particular, there  exists  a way to split the environment at each node such that the proof tree can be \mbox{successfully completed.} 

\subsection{Examples of reduction}\label{sect:nw-ex}
Consider the configuration $\Conf{\App{\y}{\PairExp{3}{\x}{\x}{5}}}{\AddToEnv{\EnvElem{\y}{1}{\bvPair{\VarGrade{\x}{ 1}}{\funEx}}}{\x}{ 9 }{\bvPair{\emptyset}{\un}}}$, where $\funEx=\Fun{\z}{\Seq{\x}{\z}}$, with the grade algebra for bounded counting of \refItem{ex:gr-alg}{nat}.
The configuration can be reduced, as we show in  \cref{fig:ex-red-1}. In this figure and the following we write $\un$ for $\unit$ to save space, and $\bv=\bvPair{\emptyset}{\PairExp{3}{\un}{\un}{5}}$. 
\begin{figure}[h]
\begin{small}
\begin{math}
\begin{array}{lc}
\mathcal{D}_1= &
\prooftree
\prooftree
\Space
\justifies
\red{\x}{\EnvElem{\x}{3}{\bvPair{\emptyset}{\un}}}{3}{\bvPair{\emptyset}{\un}}{\emptyset}
	\thickness=0.08em
	\shiftright 2em
	\using \scriptstyle{\textsc{(var)}}
\endprooftree
\BigSpace
\prooftree
\Space
\justifies
\red{\x}{\EnvElem{\x}{5}{\un}}{5}{\bvPair{\emptyset}{\un}}{\emptyset}
	\thickness=0.08em
	\shiftright 2em
	\using \scriptstyle{\textsc{(var)}}
\endprooftree
\justifies
\red{\PairExp{3}{\x}{\x}{5}}{\EnvElem{\x}{ 8}{\un}}{1}{\bv}{ \emptyset }\Space
	\thickness=0.08em
	\shiftright 2em
	\using \scriptstyle{\textsc{(pair)}}
\endprooftree
\\[7ex]
\mathcal{D}_2= &

\prooftree
\prooftree
\justifies
\red{\x}{ \EnvElem{\x}{1}{\bvPair{\emptyset}{\un}} }{1}{\bvPair{\emptyset}{\un}}{\emptyset}
	\thickness=0.08em
	\shiftright 2em
	\using \scriptstyle{\textsc{(var)}}
\endprooftree
\BigSpace
\prooftree
\justifies
\red{ z}{\EnvElem{ z}{1}{\bv}}{1}{\bv}{\emptyset}
	\thickness=0.08em
	\shiftright 2em
	\using \scriptstyle{\textsc{(var)}}
\endprooftree
\justifies
\red{\Subst{\Subst{\e}{\x'}{\y}}{\funEx'}{\funEx}}{\EnvElem{\x}{5}{\bvPair{\emptyset}{\un}},\EnvElem{\x'}{1}{\bvPair{\emptyset}{\PairExp{3}{\un}{\un}{5}}}}{\rgr}{\Conf{\bv}{\env'}}
\justifies
\red{\Seq{\x}{ z}}{\AddToEnv{\EnvElem{\x}{1}{\bvPair{\emptyset}{\un}}}{ z}{1}{\bv}}{1}{\bv }{\emptyset}
	\thickness=0.08em
	\shiftright 2em
	\using \scriptstyle{\textsc{(match-u)}}
\endprooftree

\end{array}
\end{math}
\begin{math}
\begin{array}{cl}
\mathcal{D}=&
\prooftree
\prooftree
\Space
\justifies
\red{\y} { \EnvElem{\y}{ 1}{\bvPair{\VarGrade{\x}{1}}{ \funEx}}, \EnvElem{\x}{1}{\bvPair{\emptyset}{\un}} }{1}{ \bvPair{\VarGrade{\x}{1} }{\funEx}}{\EnvElem{\x}{1}{\bvPair{\emptyset}{\un}} }
\thickness=0.08em
	\shiftright 2em
	\using \scriptstyle{\textsc{(var)}}
	 \endprooftree
\BigSpace
\mathcal{D}_1
\BigSpace
\mathcal{D}_2
\justifies
\red{\App{\y}{\PairExp{3}{\x}{\x}{5}}}{\AddToEnv{\EnvElem{\y}{1}{\bvPair{\VarGrade{\x}{1}}{\funEx}}}{\x}{ 9}{\bvPair{\emptyset}{\un}}}{1}{ \bv }{\emptyset}
	\thickness=0.08em
	\shiftright 2em
	\using \scriptstyle{\textsc{(app)}}
	 \endprooftree
\end{array}
\end{math}
\end{small}
\caption{Example of  non-wasting  reduction}
\label{fig:ex-red-1}
\end{figure}
The same reduction is derivable by considering the grade algebra for exact counting of \refItem{ex:gr-alg}{nat}. On the contrary, the configuration ${\Conf{\App{\y}{\PairExp{3}{\x}{\x}{5}}}{\AddToEnv{\EnvElem{\y}{ 2}{\bvPair{\VarGrade{\x}{1}}{\funEx}}}{\x}{9}{\bvPair{\emptyset}{\un}}}}$ illustrates how the grade algebra affects the semantics.
Indeed, with bounded counting this configuration can be reduced,  as shown  at  the top section of \cref{fig:ex-red-2}, with $\mathcal{D}_1$ and $\mathcal{D}_2$  as in \cref{fig:ex-red-1}.  
\begin{figure}[h]
\begin{small}
\begin{math}
\begin{array}{cl}
\mathcal{D}=&
\prooftree
\prooftree
\Space
\justifies
\red{\y}{ \EnvElem{\y}{2}{ \bvPair{\VarGrade{\x}{1}}{\funEx} }, \EnvElem{\x}{1}{\bvPair{\emptyset}{\un}} }{1}{\bvPair{\VarGrade{\x}{1}}{\funEx}}{ \EnvElem{\y}{1}{\bvPair{\VarGrade{\x}{1}}{\funEx}}, \EnvElem{\x}{1}{\bvPair{\emptyset}{\un}}}
\thickness=0.08em
	\shiftright 2em
	\using \scriptstyle{\textsc{(var)}}
	 \endprooftree
\BigSpace
\mathcal{D}_1
\BigSpace
\mathcal{D}_2
\justifies
\red{\App{\y}{\PairExp{3}{\x}{\x}{5}}}{\AddToEnv{\EnvElem{\y}{ 2}{\bvPair{\VarGrade{\x}{1}}{\funEx}}}{\x}{9}{\bvPair{\emptyset}{\un}}}{1}{ \bv }{\EnvElem{\y}{1}{\bvPair{\VarGrade{\x}{1}}{\funEx}}} 
	\thickness=0.08em
	\shiftright 2em
	\using \scriptstyle{\textsc{(app)}}
	 \endprooftree
\\[6ex]
\end{array}
\end{math}
\end{small}

\hrule

\begin{small}
\begin{math}
\begin{array}{lc}
\\
\mathcal{D}'_1= &
\prooftree
\prooftree
\justifies
\red{\x}{\EnvElem{\y}{1}{\bvPair{\VarGrade{\x}{1}}{\funEx}},\EnvElem{\x}{3}{\bvPair{\emptyset}{\un}}}{3}{\bvPair{\emptyset}{\un}}{\EnvElem{\y}{1}{\bvPair{\VarGrade{\x}{1}}{\funEx}}}
	\thickness=0.08em
	\shiftright 2em
	\using \scriptstyle{\textsc{(var)}}
\endprooftree
\BigSpace
\prooftree
\Space
\justifies
\red{\x}{\EnvElem{\x}{5}{\un}}{5}{\bvPair{\emptyset}{\un}}{\emptyset}
	\thickness=0.08em
	\shiftright 2em
	\using \scriptstyle{\textsc{(var)}}
\endprooftree
\justifies
\red{\PairExp{3}{\x}{\x}{5}}{\EnvElem{\y}{1}{\bvPair{\VarGrade{\x}{1}}{\funEx}},\EnvElem{\x}{8}{\un}}{1}{ \bv }{\EnvElem{\y}{1}{\bvPair{\VarGrade{\x}{1}}{\funEx}}}\Space
	\thickness=0.08em
	\shiftright 2em
	\using \scriptstyle{\textsc{(pair)}}
\endprooftree
\end{array}
\end{math}
\begin{math}
\begin{array}{cl}
\mathcal{D'}=&
\prooftree
\prooftree
\Space
\justifies
\red{\y}{ \EnvElem{\y}{1}{ \bvPair{\VarGrade{\x}{1}}{\funEx} }, \EnvElem{\x}{1}{\bvPair{\emptyset}{\un}} }{1}{\bvPair{\VarGrade{\x}{1}}{\funEx}}{ \EnvElem{\x}{1}{\bvPair{\emptyset}{\un}}}
\thickness=0.08em
	\shiftright 2em
	\using \scriptstyle{\textsc{(var)}}
	 \endprooftree
\BigSpace
\mathcal{D'}_1
\BigSpace
\mathcal{D}_2
\justifies
\red{\App{\y}{\PairExp{3}{\x}{\x}{5}}}{\AddToEnv{\EnvElem{\y}{ 2}{\bvPair{\VarGrade{\x}{1}}{\funEx}}}{\x}{9}{\bvPair{\emptyset}{\un}}}{1}{ \bv }{\EnvElem{\y}{1}{\bvPair{\VarGrade{\x}{1}}{\funEx}}} 
	\thickness=0.08em
	\shiftright 2em
	\using \scriptstyle{\textsc{(app)}}
	 \endprooftree
\end{array}
\end{math}
\end{small}

\caption{Example of  wasting  reduction}
\label{fig:ex-red-2}
\end{figure}

On the other hand, if we consider exact counting, the configuration is stuck, since  there is no way to split the environment $\AddToEnv{\EnvElem{\y}{2}{\bvPair{\VarGrade{\x}{1}}{\funEx}}}{\x}{ 9 }{\bvPair{\emptyset}{\un}}$ to obtain a proof tree. We consider, for instance, two possibilities. 

If  we use both the copies of $\y$ to reduce the function, as shown  at  the top section of \cref{fig:ex-red-2}, the computation is stuck when we try to apply rule \refToRule{var} in $\mathcal{D}$, since  $\NoNoWaste{1\rmul(\VarGrade{\x}{1})}{\EnvElem{\y}{1}{\bvPair{\VarGrade{\x}{1}}{\funEx}},\EnvElem{\x}{1}{\bvPair{\emptyset}{\un} }}$ as $0\not\leq 1$. 
That is,  we cannot consume only one copy of $\y$ since the remaining grade would be non-discardable. 

 If, instead, we  allocate one copy of $\y$ to reduce the function and one to reduce its argument, as shown at the bottom section of \cref{fig:ex-red-2}, the computation is also stuck, when in $\mathcal{D}'_1$ we try to apply one of the instances of rule  \refToRule{var}. Indeed, here $\y$ is not used and its grade is non-discardable, that is, we would waste $\y$.

\subsection{Results}\label{sect:nw-res} 
As expected, no-waste semantics can be seen as a refinement of resource-aware semantics. 
More precisely, given a no-waste reduction, there exist infinitely many resource-aware reductions  obtained by  adding a ``wasted'' portion of environment. This is formalized below, where
$\Erase{\env}$ denotes the environment obtained from $\env$ by removing annotations from values.
\newcommand{\wasted}{\hat{\env}}
\begin{theorem}\label{theo:nw-refines}
If $\red{\E}{\env_1}{\rgr}{\bvPair{\cctx}{\val}}{\env_2}$ then, for all $\wasted$ such that $\Erase{\env_1}\envSum\wasted$ is defined, $\reduce{\Conf{\E}{\Erase{\env_1}\envSum\wasted}}{\rgr}{\Conf{\val}{\Erase{\env_2}\envSum\wasted}}$. 
\end{theorem}
\begin{proof}
Let us write $\OKSum{\env}{\env'}$ if $\env$ and $\env'$ are non-conflicting, that is, $\env\envSum\env'$ is defined. First of all note that, if $\red{\E}{\env}{\rgr}{\bvPair{\cctx}{\val}}{\env'}$, and $\OKSum{\Erase{\env}}{\wasted}$ holds, then $\OKSum{\Erase{\env'}}{\wasted}$ holds as well, since variables added in $\env'$ are fresh.\\
By induction on the no-waste reduction. We show one base and one inductive case. 
\begin{description}

\item[\refToRule{var}] We have $\red{\x}{\env,\EnvElem{\x}{\sgr}{\bv}}{\rgr}{\bv}{\env,\EnvElem{\x}{\sgr'}{\bv}}$, with $\bv=\bvPair{\cctx}{\val}$ and $\rgr\rsum\sgr'\rord \sgr$. Take an arbitrary $\wasted$ such that 
$\OKSum{\Erase{\env,\EnvElem{\x}{\sgr}{\bv}}}{\wasted}$. \\
If $\x\not\in\dom{\wasted}$, then the side condition $\rgr\rsum\sgr'\rord \sgr$ of rule \refToRule{var} in \cref{fig:semantics} holds, hence $\reduce{\Conf{\x}{(\Erase{\env},\EnvElem{\x}{\sgr}{\val}) \envSum\wasted}}{\rgr}{\Conf{\val}{\Erase{\env},\EnvElem{\x}{\sgr'}{\val} \envSum\wasted}}$.\\
If ${\wasted}(\x)=\Pair{\rgr'}{\val}$,  then the side condition $\rgr\rsum\sgr'\rsum\rgr'\rord\sgr\rsum\rgr'$ of rule \refToRule{var} in \cref{fig:semantics} holds, hence $\reduce{\Conf{\x}{(\Erase{\env},\EnvElem{\x}{\sgr}{\val}) \envSum\wasted}}{\rgr}{\Conf{\val}{\Erase{\env},\EnvElem{\x}{\sgr'}{\val} \envSum\wasted}}$.
In both cases we get the thesis. 
\item[\refToRule{pair}] We have 
\begin{quoting}
(1) $\red{\ve_1}{\env_1}{\rgr\rmul\rgr_1}{\bvPair{\cctx_1}{\val_1}}{\env_1'}$\\
(2) $\red{\ve_2}{\env_2}{\rgr\rmul\rgr_2}{\bvPair{\cctx_2}{\val_2}}{\env_2'}$\\
(3) ${\red{\PairExp{\rgr_1}{\ve_1}{\ve_2}{\rgr_2}}{\env}{\rgr}{\bvPair{\cctx_1\ctxsum\cctx_2}{\PairExp{\rgr_1}{\val_1}{\val_2}{\rgr_2}}}{\env_1'\envSum\env_2'}}$\\
(4) $\env_1\envSum\env_2\ctxord\env$
\end{quoting}
From (4) we get that $\OKSum{\erase{\env_1}}{\erase{\env_2}}$ holds. 
Hence, from (1) by inductive hypothesis, taking as additional environment $\erase{\env_2}$, we get 
\begin{quoting}
$\red{\ve_1}{\erase{\env_1}\envSum\erase{\env_2}}{\rgr\rmul\rgr_1}{\val_1}{\erase{\env'_1}\envSum\erase{\env_2}}$
\end{quoting}
Since $\OKSum{\erase{\env'_1}}{\erase{\env_2}}$, 
from (2) by inductive hypothesis, taking as additional environment $\erase{\env'_1}$, we get
\begin{quoting}
$\red{\ve_2}{\erase{\env_2}\envSum\erase{\env'_1}}{\rgr\rmul\rgr_1}{\val_1}{\erase{\env'_2}\envSum\erase{\env'_1}}$
\end{quoting}
Hence, we can apply rule \refToRule{pair} in \cref{fig:semantics}:
\begin{quoting}
${\reduce{\Conf{\PairExp{\rgr_1}{\ve_1}{\ve_2}{\rgr_2}}{\env}}{\rgr}{\Conf{\PairExp{\rgr_1}{\val_1}{\val_2}{\rgr_2}}}{\erase{\env_1'}\envSum\erase{\env_2'}}}$
\end{quoting}
and we get the thesis.

\end{description}
\end{proof}

We state now the main result of this section: the reduction relation is no-waste.
We say that \emph{$\Conf{\E}{\env}$} is no-waste, written $\NWConf{\E}{\env}$, if $\NoWaste{\cctx}{\env}$ holds for some $\cctx$ honest for $\E$.

\begin{theorem}[Semantics is no-waste]\label{theo:nw-sem}
If $\red{\E}{\env}{\rgr}{\bvPair{\cctx'}{\val}}{\env'}$ then $\NWConf{\E}{\env}$, and $\NoWaste{\rgr\ctxmul\cctx'}{\env'}$.
 \end{theorem}

Here, $\NoWaste{\rgr\ctxmul\cctx'}{\env'}$  states that, as required, \emph{only no-waste results are produced}; since the final result is produced in $\rgr$ copies, the no-waste condition holds for a grade context obtained by multiplying by $\rgr$ that declared in the annotated value.

 In addition, the theorem provides a \emph{characterization of configurations which do not get stuck due to wasting errors}. In other words, configurations such that $\NWConf{\E}{\env}$ does not hold are stuck, since they cannot be reduced to a no-waste result. 

 Let us consider the configuration $\Conf{\x}{\EnvElem{\x}{1}{\bvPair{\VarGrade{\y}{1},\VarGrade{\z}{1}}{(\Fun{\_}{\PairExp{1}{\y}{\z}{1}})},\EnvElem{\y}{1}{\bvPair{\emptyset}{\unit}},\EnvElem{\z}{1}{\bvPair{\emptyset}{\unit}}}}$.  Assuming linearity grades, this  configuration is non-stuck, since  there exists a grade context (a way to assign grades to the variables), notably, $\VarGrade{\x}{1}$, which uses with no-waste the environment. Indeed, the three resources need to be used with grade $1$, and such grade context  uses, with grade $1$, directly $\x$, and transitively $\y$ and $\z$. 
On the contrary, the configuration $\Conf{\x}{\EnvElem{\x}{1}{\bvPair{\VarGrade{\y}{1}}{(\Fun{\_}{\y})},\EnvElem{\y}{1}{\bvPair{\emptyset}{\unit}},\EnvElem{\z}{1}{\bvPair{\emptyset}{\unit}}}}$ is stuck,  since \emph{there is no grade context} such that the no-waste condition holds. Indeed, since $\z$ is no longer transitively used, it is should be directly used in the grade context,  but this would not be honest for the value expression $\x$.

Of course we expect, and will prove, that well-typed configurations are no-waste; however, note that the characterization above is purely semantic, and, as usual, it can hold for ill-typed configurations. For instance, assuming to have  booleans and conditional, which can be encoded as shown in \cref{fig:boolNatList}, the configuration
$\Conf{\If{\True}{\x}{\unit}}{\EnvElem{\x}{1}{\bvPair{\emptyset}{\unit}}}$ is no-waste,  since $\NoWaste{\VarGrade{\x}{1}}{\EnvElem{\x}{1}{\unit}}$ holds, and indeed can be safely reduced to $\Conf{\bvPair{\emptyset}{\unit}}{\emptyset}$. However, this configuration is ill-typed, since the else alternative does not use $\x$.

\cref{theo:nw-sem} is an immediate corollary of the following one.

\begin{theorem}[Resource balance]\label{theo:balance}
If $\red{\E}{\env}{\rgr}{\bvPair{\cctx'}{\val}}{\env'}$ then there exist $\cctx,\used,\added$, with   $\cctx$ honest for $\E$,   $\cctx(\x)=\rzero$ for each $\x\not\in\dom{\env}$, 
and $\added(\x)=\rzero$ for all $\x\in\dom{\env}$,  such that
the following conditions hold.
\begin{enumerate}
\item\label{theo:nw-sem:1} $\rgr\ctxmul\cctx'\ctxmul\GraphStar{\env'}\ctxord\CCTX{\env'}$
\item\label{theo:nw-sem:2}  $\cctx\ctxmul\GraphStar{\env}\ctxord\CCTX{\env}$
\item\label{theo:nw-sem:3} $\CCTX{\env'}\ctxsum\used\ctxord\CCTX{\env}\ctxsum\added$
\item\label{theo:nw-sem:4} $\cctx\ctxmul\GraphStar{\env} \ctxsum\added\rord\rgr\ctxmul\cctx'\ctxmul\GraphStar{\env'}\ctxsum\used$
\end{enumerate}
\end{theorem}  
Notably, \cref{theo:nw-sem:1}  shows that the result is no-waste, and 
 $\cctx$ in \cref{theo:nw-sem:2} is the grade context
 used to prove that the configuration is no-waste. The grade contexts 
$\used$ and $\added$ denote, respectively, \emph{consumed} and \emph{added} resources during the computation, which are used to express balance (in)equations in \cref{theo:nw-sem:3,theo:nw-sem:4}.   Added resources are always fresh, hence   $\added(\x)=\rzero$ for all $\x\in\dom{\env}$  is expected to hold. 
The inequation of \cref{theo:nw-sem:3} states that remaining resources plus those consumed during the computation cannot exceed the sum of initially available resources and those  added  by the evaluation. 
This is analogous to a resource-balance result proved  by \cite{ChoudhuryEEW21},  but here it is expressed in purely semantic terms.
 \cref{theo:nw-sem:4} provides a strict version formulated using the more precise information carried by the grade graphs of the environments and is the key property to prove the no-waste result.  
 
The no-waste property implies, in particular, that any resource in the environment which is not reachable from the program (garbage) is discardable, as formalized below. As mentioned before, such  a  resource could be discarded in a semantics removing garbage.

Let us denote by $\reach{}{\env}{}$ the \emph{reachability relation} in $\env$, that is, the reflexive and transitive closure of the relation $\reachOne{}{\env}{}$ defined by:
\begin{quoting}
$\reachOne{\x}{\env}{\y}$ if $\env(\x)=\Pair{\_}{\bvPair{\cctx}{\val}}$ and $\y\in\fv{\val}$
\end{quoting}

The following lemma states that, as expected, the existence  of  a non-discardable path implies reachability.
\begin{lemma}\label{lem:GID}
If $\rzero\not\rord\Graph{\env}(\p)$, then $\reach{\Start{\p}}{\env}{\End{\p}}$.
\end{lemma}
\begin{proof}
By induction on the length of $\p$. 
If $\p=\x$, then $\reach{\x}{\env}{\x}$ by reflexivity. Otherwise, $\p=\x\cdot\p'$, hence $\env(\x)=\Pair{\_}{\bvPair{\cctx}{\val}}$. From $\rzero\not\rord\Graph{\env}(\x\cdot\p)$ we get
$\rzero\not\rord\graph(\x,\Start{\p'})$ and $\rzero\not\rord\Graph{\env}(\p')$. 
Being $\cctx$ honest for $\val$, we get $\Start{\p'}\in\fv{\val}$, 
hence $\reachOne{\x}{\env}{\Start{\p'}}$, and we get the thesis by applying the inductive hypothesis to $\rzero\not\rord\Graph{\env}(\p')$, and then by transitivity. 

\end{proof}
We write $\reach{X}{\env}{\y}$ if $\reach{\x}{\env}{\y}$ holds for some $\x\in X$.
\begin{proposition}[No-waste $\Rightarrow$ Garbage is Discardable]\label{prop:GID}
If $\NoWaste{\E}{\env}$ holds:
\begin{quoting}
$\notReach{\fv{\E}}{\env}{\x}$ implies $\rzero\rord\CCTX{\env}(\x)$. 
\end{quoting}
\end{proposition}
\begin{proof}
Set $\rgr=\CCTX{\env}(\x)$.
We prove that $\rzero\not\rord\rgr$ implies $\reach{\fv{\E}}{\env}{\x}$. \\
For some $\cctx$ honest for $\E$, we have $\cctx\rmul{\GraphStar{\env}}(\x)\rord\rgr$. Since $\rzero\not\rord\rgr$, it is necessarily $\rzero\not\rord\cctx\rmul{\GraphStar{\env}}(\x)$, that is, $\rzero\not\rord\sum_{\y\in\Vars} \cctx(\y)\cdot\GraphStar{\env}(\y,\x)$. Hence $\rzero\not\rord\cctx(\y)\cdot\GraphStar{\env}(\y,\x)$ for some $\y$, and this implies $\rzero\not\rord\cctx(\y)$ and $\rzero\not\rord\GraphStar{\env}(\y,\x)$, hence there is a path $\p$ from $\y$ to $\x$ with $\rzero\not\rord\Graph{\env}(\p)$.
Being $\cctx$ honest for $\E$,  we get $\y\in\fv{\E}$, hence $\reachOne{\x}{\env}{\y}$, and by \cref{lem:GID} and transitivity we get the thesis. 
\end{proof}

Let us consider the type system  in \cref{sect:typesystem},  with the only difference that, as shown in \cref{fig:nw-typing-conf}, annotations in values are taken into account. 
That is,  rule \refToRule{t-val} checks that annotations  are (overapproximated by) the direct usage obtained by typechecking the value.  
Moreover, since annotated values are not value expressions, we have to add a typing rule for results.  Here $\CCTX{\Gamma}$ is the grade context
extracted from $\Gamma$, i.e., $\CCTX{\Gamma}=\VarGrade{\x_1}{\rgr_1}, \ldots, \VarGrade{\x_n}{\rgr_n}$ for $\Gamma = \VarGradeType{\x_1}{\rgr_1}{\tau_1},\ldots,\VarGradeType{\x_n}{\rgr_n}{\tau_n}$. 
\begin{figure}
\begin{math}
\begin{array}{l}

\NamedRule{t-val}{
\IsWFExp{\Gamma}{\val}{\Graded{\tau}{\rone}}
}{
\IsWFExp{\rgr\rmul\Gamma}{\bvPair{\cctx}{\val}}{\Graded{\tau}{\rgr}}
}{
\cctx\ctxord\CCTX{\Gamma} 
}

\\[4ex]

\NamedRule{t-env}{
\IsWFExp{\Gamma_i}{\bv_i}{\GradedInd{\tau}{i}{\rgr}} \Space \forall i \in 1..n
}
{
\IsWFEnv{\Gamma}{\EnvElem{\x_1}{\rgr_1}{\bv_1},\ldots,\EnvElem{\x_n}{\rgr_n}{\bv_n}}{\Delta}
}{
\Gamma = \VarGradeType{\x_1}{\rgr_1}{\tau_1},\ldots,\VarGradeType{\x_n}{\rgr_n}{\tau_n}\\
\Gamma_1 \ctxsum \ldots \ctxsum \Gamma_n\ctxsum\Delta\ctxord \Gamma
}

\\[4ex]

\NamedRule{t-conf}{
\IsWFExp{\Delta}{\E}{\T} \BigSpace
\IsWFEnv{\Gamma}{\env}{\Delta}
}
{
\IsWFConf{\Gamma}{\E}{\env}{\T}
}{ }
\BigSpace\NamedRule{t-res}{
\IsWFExp{\Delta}{\bv}{\T} \BigSpace
\IsWFEnv{\Gamma}{\env}{\Delta}
}
{
\IsWFConf{\Gamma}{\bv}{\env}{\T}
}{ }

\end{array}
\end{math}
\caption{Typing rules for (annotated) values, environments, configurations, and results}
\label{fig:nw-typing-conf}
\end{figure}

 We  prove now that well-typed configurations  and results  are no-waste (\cref{theo:no-waste-typing}).  
To this end, the key property is \cref{prop:env-typing-star}, stating that, given a well-typed environment, the grade context extracted from the residual context uses the environment with no waste.
In other words, the resources left in the environment should be only those needed (modulo approximation) by the residual context. Thanks to the induction principle,  proved  in \cref{prop:implication}, this property 
is implied by the following   \cref{prop:env-typing},  related to grade graph rather than its closure, which can be proved by routine reasoning on the typing rules.  

\begin{proposition}\label{prop:env-typing}
If $\IsWFEnv{\Gamma}{\env}{\Delta}$ then $\CCTX{\Delta}\ctxsum\CCTX{\env}\ctxmul\Graph{\env}\ctxord\CCTX{\env}$.
\end{proposition}

\begin{proposition}\label{prop:env-typing-star}
If $\IsWFEnv{\Gamma}{\env}{\Delta}$ then $\CCTX{\Delta}\ctxmul\GraphStar{\env}\ctxord\CCTX{\env}$.
\end{proposition}
\begin{proof}
It is an immediate consequence of \cref{prop:implication,prop:env-typing}.
\end{proof}

The next lemma shows that the grade context extracted from the typing context of a well-typed expression is always honest for it.  

\begin{lemma}\label{lem:fv}
If $\IsWFExp{\Gamma}{\E}{\T}$ and $\rzero\not\rord\CCTX{\Gamma}(\x)$ then $\x\in\fv{\E}$.
\end{lemma}
\begin{proof}
By induction on the typing rules.
\end{proof}

\begin{theorem}[Well-typedness implies no-waste]\label{theo:no-waste-typing}\
\begin{enumerate}
\item \label{theo:no-waste-typing:exp}
If $\IsWFConf{\Gamma}{\E}{\env}{\T}$ then $\NWConf{\E}{\env}$.
\item \label{theo:no-waste-typing:val}
If $\IsWFConf{\Gamma}{\bvPair{\cctx}{\val}}{\env}{\Graded{\tau}{\rgr}}$ then $\NoWaste{\rgr\ctxmul\cctx}{\env}$.
\end{enumerate}
\end{theorem}
\begin{proof}
\begin{enumerate}
\item
We have to show that $\NoWaste{\cctx}{\env}$ for some $\cctx$ honest for $\E$.
To derive $\IsWFConf{\Gamma}{\E}{\env}{\T}$,
we have necessarily applied rule \refToRule{t-conf}, hence $\IsWFExp{\Delta}{\E}{\T}$ and $\IsWFEnv{\Gamma}{\env}{\Delta}$ hold for some $\Delta$. By \cref{prop:env-typing-star} we have $\CCTX{\Delta}\ctxmul\GraphStar{\env}\ctxord\CCTX{\env}$, that is, $\NoWaste{\CCTX{\Delta}}{\env}$. Moreover, $\CCTX{\Delta}$ is honest for $\E$ by \cref{lem:fv}.

\item 
To derive $\IsWFConf{\Gamma}{\bvPair{\cctx}{\val}}{\env}{\Graded{\tau}{\rgr}}$,
we have necessarily applied rule \refToRule{t-res}, hence $\IsWFExp{\Delta}{\bvPair{\cctx}{\val}}{\Graded{\tau}{\rgr}}$ and $\IsWFEnv{\Gamma}{\env}{\Delta}$ hold for some $\Delta$. 
Moreover, to derive $\IsWFExp{\Delta}{\bvPair{\cctx}{\val}}{\Graded{\tau}{\rgr}}$ we have necessarily applied rule \refToRule{t-val}, hence $\Delta=\rgr\rmul\Delta'$ and $\cctx\ctxord\CCTX{\Delta'}$. Moreover, by definition $\CCTX{\Delta}=\rgr\rmul\CCTX{\Delta'}$, hence $\rgr\rmul\cctx\ctxord\CCTX{\Delta}$. By \cref{prop:env-typing-star} we have $\CCTX{\Delta}\ctxmul\GraphStar{\env}\ctxord\CCTX{\env}$, hence $(\rgr\rmul\cctx)\ctxmul\GraphStar{\env}\ctxord\CCTX{\env}$, that is, $\NoWaste{\rgr\ctxmul\cctx}{\env}$. 
\end{enumerate}
\end{proof}

We end this section by briefly discussing the acyclicity assumption on enviroments. Technically, this allows us to give \cref{def:closure}, where an infinite sum of  path grades can be reduced to a finite one thanks to  \refItem{prop:paths}{2}. 
However, infinite sums would be well-defined as well by assuming suitable conditions on the grade algebra \citep{Golan03,EsikK04,EsikL05,FengJ09}. 
A more serious problem is that  the proof of \cref{prop:implication} relies on acyclicity; that is, there is no guarantee that the condition ensured by well-typedness ($\CCTX{\Delta}\ctxsum\CCTX{\env}\ctxmul\Graph{\env}\ctxord\CCTX{\env}$) implies the no-waste condition ($\CCTX{\Delta}\ctxmul\GraphStar{\env}\ctxord\CCTX{\env}$). In a sense, the meaning of no-waste in presence of cyclic environments is unclear. Consider a cyclic resource which is ``garbage'', that is, not reachable from the program, as in, e.g., $\Conf{\val}{\EnvElem{\x}{\rgr}{\bvPair{\VarGrade{\x}{1}}{\Fun{z}{\x}}}}$, with $\val$ closed. This configuration would be well-typed for an arbitrary grade $\rgr$, since the resource is ``auto-consuming''. We will investigate meaning and formalization of no-waste for cyclic resources in future work.


\section{Type soundness}\label{sect:soundness}

In this section, we prove our main result: soundness of the type system with respect to the  no-waste  semantics\footnote{ Thanks to \cref{theo:nw-refines}, this implies soundness with respect to the semantics in \cref{fig:semantics} as well.}.
That is, for well-typed expressions there is a computation which is  not stuck for any reason, including  resource  exhaustion and waste. 
Note that this is a \emph{may} flavour of soundness \citep{DeNicolaH84,DagninoBZD20,Dagnino22}, which is the only one we can prove in this context, because the semantics is non-deterministic. 
We analyse separately type soundness for value expressions and expressions. 
 In the following, we will say that two judgements 
$\IsWFEnv{\Gamma_1}{\env_1}{\Delta_1}$ and $\IsWFEnv{\Gamma_2}{\env_2}{\Delta_2}$ are \emph{consistent} if 
in their derivations 
the values associated with variables in the common domain of $\env_1$ and $\env_2$  are typed with the same type and typing context. 
Formally, this means that 
for all $\x\in\dom{\env_1}\cap\dom{\env_2}$, there are $\bv = \bvPair\cctx\val$, $\tau$ and $\Delta$ such that 
$\env_1(\x) = \Pair{\_}{\bv}$, $\env_2(\x) = \Pair{\_}{\bv}$, 
$\Gamma_1(\x) = \Pair{\_}{\tau}$, $\Gamma_2(\x) = \Pair{\_}{\tau}$, and 
in the derivations of both judgements there is a premise 
$\IsWFExp{\Delta}{\val}{\Graded\tau\rone}$. 
Similarly, we say that two typing judgements for configurations are consistent if so are the judgements of the environments they contain. 
 
\paragraph*{Type soundness for value expressions}
Since  reduction of  value expressions cannot diverge, type soundness means that  well-typed value expressions reduce  to a value, \mbox{as stated below.}

\begin{theorem}[Soundness for value expressions]\label{theo:soundness-ve}
If $\IsWFConf{\Gamma}{\ve}{\env}{\Graded{\tau}{\rgr}}$, then $\NWreduce{\Conf{\ve}{\env}}{\rgr}{\Conf{\bv}{\env'}}$ for some $\bv$, $\env'$.
\end{theorem}
This is a corollary of the following extended formulation which can be proved inductively on the structure of value expressions, relying on standard lemmas.
\begin{theorem}[Soundness for value expressions (extended)]\label{theo:adjustV}
If $\IsWFConf{\Gamma}{\ve}{\env}{\Graded\tau\rgr}$ then
$\NWreduce{\Conf\ve{\env}}{\rgr}{\Conf\bv{\env'}}$ and 
$\IsWFConf{\Gamma'}{\bv}{\env'}{\Graded\tau\rgr}$ 
 consistent with $\IsWFConf\Gamma\ve\env{\Graded\tau\rgr}$,   
for some $\bv$, $\env'$, $\Gamma'$.
\end{theorem}

 Note that, even though reduction of value expressions just performs substitution, in a resource-aware semantics this is a significant event, since it implies consuming some amount of resources.  \cref{theo:soundness-ve} states that no resource exhaustion  or waste  can happen. 

\paragraph*{Adding divergence}
For expressions, instead,  big-step semantics, as those defined in \cref{fig:semantics} and  \cref{fig:nw-sem}, suffer  from  the long-known drawback \citep{CousotC92,LeroyG09} that non-terminating and stuck computations are  indistinguishable, since in both cases no finite proof tree of a judgment can be constructed.   
This is an issue for our aim: to prove that for a well-typed expression  there is a computation which does not get stuck, that is, either produces a value or diverges.  
 To solve this problem, we extend the big-step semantics to explicitly model diverging computations, proceeding as follows, with $\produzioneinline{\conf}{\Conf{\e}{\env}}$.
\begin{itemize}
\item the shape of the judgment  for expressions  is generalized to $\NWreduce{\conf}{\rgr}{\res}$, where the \emph{result}  $\res$ is either a pair consisting of a value and a final environment,  or \emph{divergence} ($\infty$); 
\item this judgment is defined through a \emph{generalized inference system}, shown in \cref{fig:div}, consisting of the \emph{rules} from \refToRule{ret} to \refToRule{match-u/match-u-div}, and the \emph{corule} \refToRule{co-div} 
(differences with respect to the previous semantics in the bottom section of \cref{fig:nw-sem} are emphasized in grey\footnote{Recall that, since the evaluation judgment is stratified, premises involving the judgment for value expressions can be equivalently considered as side conditions.}). 
\end{itemize}
The key point here is that, in generalized inference systems, 
rules are interpreted in an \emph{essentially coinductive}, rather than inductive, way.
For details on generalized inference systems 
we refer to  the works by \cite{AnconaDZ@esop17,Dagnino19}; 
here, for the reader's convenience, we provide a self-contained presentation, instantiating general definitions on our specific case.   

\begin{figure}
\begin{small}
\begin{grammatica}
\produzione{\res}{\Conf{\bv}{\env}}{result}\\
\seguitoproduzione{\infty}{}
\\[2ex]
\end{grammatica}
\hrule
\begin{math}
\begin{array}{l}
\\
\NamedRule{ret}{\redval{\ve}{\env}{\sgr}{\bv}{\env'} }
{ \NWreduce{\Conf{\Return\ve}{\env}}{\rgr}{\Conf{\bv}{\env'}} }
{ 
  \rgr\rord\sgr\ne\rzero 
} 
\BigSpace
\meta{\NamedRule{let-div1}{
    \NWreduce{\Conf{\e_1}{\env}}{\sgr}{\infty}
}{ \NWreduce{\Conf{\Let{\x}{\e_1}{\e_2}}{\env}}{\rgr}{\infty} }
{} }
\\[4ex]

\NamedRule{let/let-div2}{
  \begin{array}{l}
    \NWreduce{\Conf{\e_1}{\env_1}}{\sgr}{\Conf{\bv_1}{\env_1'}} \\ 
    \NWreduce{\Conf{\Subst{\e_2}{\x'}{\x}}{\AddToEnv{(\env_1'\envSum\env_2)}{\x'}{\sgr}{\bv_1}}}{\rgr}{\meta{\res}}
  \end{array}
}{ \NWreduce{\Conf{\Let{\x}{\e_1}{\e_2}}{\env}}{\rgr}{\meta{\res}} }
{\env_1\envSum\env_2\ctxord\env\\
\x'\ \mbox{fresh}
} 

\\[4ex]

\NamedRule{app/app-div}{
\begin{array}{l}
\NWreduce{\Conf{\ve_1}{\env_1}}{\sgr}{\Conf{\bvPair{\cctx_1}{\RecFun{\f}{\x}{\e}}}{\env_1'}}\Space
\NWreduce{\Conf{\ve_2}{\env_2}}{\tgr}{\Conf{\bv_2}{\env_2'}}\\
\NWreduce{\Conf{\Subst{\Subst{\e}{\x'}{\x}}{\f'}{\f}}{\AddToEnv{\AddToEnv{(\env_1'\envSum\env_2')}{\x'}{\tgr}{\bv_2}}{\f'}{\sgr_2}{\bvPair{\cctx_1}{\RecFun{\f}{\x}{\e}}}}}{\rgr}{\meta{\res}}
\end{array}}
{
\NWreduce{\Conf{\App{\ve_1}{\ve_2}}{\env}}{\rgr}{\meta{\res}}
}{
\sgr_1 \rsum \sgr_2 \rord \sgr\\
\sgr_1 \neq \rzero\\
\env_1\envSum\env_2\ctxord\env\\
  \f',\x'\ \mbox{fresh} 
}

\\[4ex]

\NamedRule{match-u/match-u-div}{
\NWreduce{\Conf{\ve}{\env_1}}{\sgr}{\Conf{\bvPair{\cctx}{\unit}}{\env_1'}}\Space
\NWreduce{\Conf{\e}{(\env_1'\envSum\env_2)}}{\rgr}{\meta{\res}}
}
{
\NWreduce{\Conf{\MatchUnit{\ve}{\e}}{\env}}{\rgr}{\meta{\res}}
}
{
  \sgr\ne\rzero \\
\env_1\envSum\env_2\ctxord\env
}

\\[4ex]

\NamedRule{match-p/match-p-div}{
\begin{array}{l}
\NWreduce{\Conf{\ve}{\env_1}}{\sgr}{\Conf{\bvPair{\cctx_1\ctxsum\cctx_2}{\PairExp{\rgr_1}{\val_1}{\val_2}{\rgr_2}}}{\env_1'}} \\
\NWreduce{\Conf{\Subst{\Subst{\e}{\x'}{\x}}{\y'}{\y}}{\AddToEnv{\AddToEnv{\env}{\x'}{\sgr\rmul\rgr_1}{\bvPair{\cctx_1}{\val_1}}}{\y'}{\sgr\rmul\rgr_2}{\bvPair{\cctx_2}{\val_2}}}}{\rgr}{\meta{\res}}
\end{array}
}
{
\NWreduce{\Conf{\Match{\x}{\y}{\ve}{\e}}{\env}}{\rgr}{\meta{\res}}
}
{  \sgr\ne\rzero \\
\env_1\envSum\env_2\ctxord\env\\
\x,\y\ \mbox{fresh}
}

\\[5ex]

\NamedRule{match-l/match-l-div}{
\begin{array}{l}
\NWreduce{\Conf{\ve}{\env_1}}{\tgr}{\Conf{\bvPair{\cctx_1}{\Inl{\sgr}{\val_1}}}{\env_1'}} \\
\NWreduce{\Conf{\Subst{\e_1}{\y}{\x_1}}{\AddToEnv{(\env_1'\envSum\env_2)}{\y}{\sgr\rmul\rgr}{\bvPair{\cctx_1}{\val_1}}}}{\rgr}{\meta{\res}}
\end{array}
}
{
\NWreduce{\Conf{\Case{\ve}{\x_1}{\e_1}{\x_2}{\e_2}}{\env}}{\tgr}{\meta{\res}}
}
{
  \tgr\ne\rzero \\
\env_1\envSum\env_2\ctxord\env\\
\y\ \mbox{fresh}
}

\\[4ex]

\NamedRule{match-r/match-r-div}{
\begin{array}{l}
\NWreduce{\Conf{\ve}{\env_1}}{\sgr}{\Conf{\bvPair{\cctx_1}{\Inr{\rgr}{\val_1}}}{\env_1'}} \\
\NWreduce{\Conf{\Subst{\e_2}{\y}{\x_2}}{\AddToEnv{(\env_1'\envSum\env_2)}{\y}{\sgr\rmul\rgr}{\bvPair{\cctx_1}{\val_1}}}}{\tgr}{\meta{\res}}
\end{array}
}
{
\NWreduce{\Conf{\Case{\ve}{\x_1}{\e_1}{\x_2}{\e_2}}{\env}}{\tgr}{\meta{\res}}
}
{
  \tgr\ne\rzero \\
\env_1\envSum\env_2\ctxord\env\\
\y\ \mbox{fresh}
}

\\[4ex]
\meta{\NamedCoRule{co-div}{}{\NWreduce{\Conf{\e}{\env}}{\rgr}{\infty}}{}}
\end{array}
\end{math}
\end{small}

\caption{Adding divergence}
\label{fig:div}
\end{figure}

Rules in \cref{fig:div} handle divergence propagation.  Notably, for each rule in the bottom section of \cref{fig:nw-sem}, we add a divergence propagation rule for each of the possibly diverging premises. 
 The only rule with two possibly diverging premises is \refToRule{let}. Hence, divergence propagation for an expression $\Let{\x}{\e_1}{\e_2}$ is obtained by two (meta)rules: 
\refToRule{let-div1}  
when $\e_1$ diverges, and \refToRule{let-div2} when $\e_1$ converges and $\e_2$ diverges; in \cref{fig:div}, for brevity, this second metarule is merged with \refToRule{let}, using the metavariable $\res$.
All  the  other rules have only one possibly diverging premise, so one divergence propagation rule is added and merged with the original metarule, analogously to \refToRule{let-div2}.

In generalized inference systems, infinite proof trees are allowed. 
Hence, judgments $\NWreduce{\conf}{\rgr}{\infty}$ can be derived, as desired, even though there is no axiom introducing them, 
thus distinguishing diverging computations (infinite proof tree) from stuck computations (no proof tree).
However, a purely coinductive interpretation would allow the derivation of spurious judgements \citep{CousotC92,LeroyG09,AnconaDZ17}. 
To address this issue, generalized inference systems may include \emph{corules}, written with a thick line, only \refToRule{co-div} in our case, 
which refine the coinductive interpretation, filtering out some (undesired) infinite derivations. 
Intuitively, the meaning of \refToRule{co-div} is to allow infinite derivations only for divergence (see \cref{ex:div-infinite} below). 
Formally, we have the following definition instantiated from  those by \cite{AnconaDZ@esop17,Dagnino19}.  

\begin{definition}\label{def:deriv-gen}
A judgment $\NWreduce{\conf}{\rgr}{\res}$ is derivable in the generalized inference system in \cref{fig:div}, written $\gen{\conf}{\rgr}{\res}$,  if it has an infinite proof tree constructed using the rules where, moreover, each node has a \emph{finite} proof tree constructed using the rules \emph{plus the corule}. 
\end{definition}


\begin{example}\label{ex:div-infinite}
Let us consider again the expression  $\Seq{\y}{\App{\f}{\x}}$ of \cref{ex:div}. 
Now, its non-terminating evaluation in the environment $\EnvElem{\y}{\infty}{\unit},\EnvElem{\f}{\infty}{\diverge},\EnvElem{\x}{1}{\unit}$, abbreviated $\Triple{\infty}{\infty}{1}$ using the previous convention, is formalized by the infinite proof tree in \cref{fig:ex-div-infinite}, where instantiations of (meta)rules \refToRule{match-u} and \refToRule{app} have been replaced by those of the corresponding divergence propagation rule. 
It is immediate to see that each node in such infinite proof tree has a  finite proof tree constructed using the rules plus the corule: 
the only nodes which have no finite proof tree constructed using the rules are those in the infinite path, of shape  either $\NWreduce{\Conf{\Seq{\y}{\App{\f}{\x}}}{\Triple{\infty}{\infty}{1}}}{}{\infty}$ or $\NWreduce{\Conf{\App{\f}{\x}}{\Triple{\infty}{\infty}{1}}}{}{\infty}$, and such judgments are directly derivable by the corule. 
On the other hand, infinite proof trees obtained by 
using  \refToRule{match-u} and \refToRule{app} would derive $\NWreduce{\Conf{\Seq{\y}{\App{\f}{\x}}}{\Triple{\infty}{\infty}{1}}}{}{\Conf{\bv}{\Triple{\infty}{\infty}{1}}}$ for any $\bv$. 
However, 
 such judgments \emph{have no} finite proof tree using also the corule, which allows only to introduce divergence, since there is no rule deriving a value with divergence as a premise.  

\begin{figure}
 \begin{scriptsize}
 \begin{math}
 \begin{array}{l}
\prooftree
      \prooftree
       \justifies      
      \NWreduce{\ConfP{\y}{\Triple{\infty}{\infty}{1}}}{}{\ConfP{\un}{\Triple{\infty}{\infty}{1}}}
     \thickness=0.08em
      \shiftright 0em
      \endprooftree
      \!\!\!\!
  \BigSpace
\prooftree
      \prooftree
       \justifies      
      \NWreduce{\ConfP{\f}{\Triple{\infty}{\infty}{1}}}{\infty}{\ConfP{\diverge}{\Triple{\infty}{0}{1}}}
     \thickness=0.08em
      \shiftright 0em
      \endprooftree
      \!\!\!\!\!\!
      \prooftree
       \justifies      
      \NWreduce{\ConfP{\x}{\Triple{\infty}{0}{1}}}{}{\ConfP{\un}{\Triple{\infty}{0}{0}}}
     \thickness=0.08em
      \shiftright 0em
      \endprooftree
      \BigSpace
      \prooftree 
      \ldots
      \justifies
\NWreduce{\ConfP{\Seq{\y}{\App{\f}{\x}}}{\Triple{\infty}{\infty}{1}}}{}{\infty}
      \thickness=0.08em
      \shiftright 0em
      \endprooftree
      \BigSpace
      \justifies      
\NWreduce{\Conf{\App{\f}{\x}}{\Triple{\infty}{\infty}{1}}}{}{\infty}
      \thickness=0.08em
      \shiftright 0em
      \endprooftree
      \justifies              
\NWreduce{\Conf{\Seq{\y}{\App{\f}{\x}}}{\Triple{\infty}{\infty}{1}}}{}{\infty}
      \thickness=0.08em
      \shiftright 0em
      \endprooftree
\end{array}
\end{math}
\end{scriptsize}
\caption{Example of resource-aware evaluation: divergency with no exhaustion}
\label{fig:ex-div-infinite}
\end{figure}
\end{example}

The transformation from the inductive big-step semantics in \cref{fig:nw-sem} to that handling divergence in \cref{fig:div} is an instance of a general construction, taking as input an arbitrary big-step semantics,  fully formalized, and proved to be correct,  by \cite{Dagnino22}.  In particular, the construction is \emph{conservative}, that is, the semantics of converging computations is not affected, as stated in the following result,  which is an instance of Theorem 6.3  by  \cite{Dagnino22}. 

\begin{definition}\label{def:deriv-ind}
Let $\ind{\conf}{\rgr}{\Conf{\bv}{\env}}$ denote that the judgment can be derived by the rules in the bottom section of \cref{fig:nw-sem}, interpreted inductively.
\end{definition}
\begin{theorem}[Conservativity]\label{theo:conservativity}
$\gen{\conf}{\rgr}{\Conf{\bv}{\env'}}$  if and only if ${\ind{\conf}{\rgr}{\Conf{\bv}{\env'}}}$.
\end{theorem}

 Note that to achieve this result corules are essential since, as observed in \cref{ex:div-infinite}, 
a purely coinductive interpretation allows for infinite proof trees deriving values.  

\paragraph*{Type soundness for expressions}
The definition of the semantics  by the generalized inference system in \cref{fig:div}  allows a very simple and clean formulation of type soundness: well-typed configurations always reduce to a result (which can be possibly divergence).
Formally:

\begin{theorem}[Soundness]\label{theo:soundness}
If $\IsWFconf{\Gamma}{\conf}{\Graded{\tau}{\rgr}}$, then $\gen{\conf}{\rgr}{\res}$ for some $\res$.
\end{theorem}

We describe now  the structure of the proof, which is interesting in itself; indeed,  the semantics being  big-step, there is no consolidated proof technique as the long-time known progress plus subject reduction for the small-step case \citep{WrightF94}.

The proof is driven by coinductive reasoning on the semantic rules,  following a schema firstly adopted  by \cite{AnconaDZ17},  as detailed below. 
First of all, it is convenient to turn to the following equivalent formulation of type soundness, stating that well-typed configurations which do not converge necessarily diverge.

\begin{theorem}[Completeness-$\infty$]\label{theo:complete-div}
If $\IsWFconf{\Gamma}{\conf}{\Graded{\tau}{\rgr}}$, and there is no $\Conf{\bv}{\env}$ s.t.\ $\gen{\conf}{\rgr}{\Conf{\bv}{\env}}$, \mbox{then $\gen{\conf}{\rgr}{\infty}$.}
\end{theorem}

Indeed, with this formulation soundness of the type system can be seen as \emph{completeness} of the set of judgements $\NWreduce{\conf}{\rgr}{\infty}$ which are derivable with respect to the set of pairs $\Pair{\conf}{\rgr}$ such that $\conf$ is well-typed with grade $\rgr$, and  does  not converge. The standard technique to prove completeness of a coinductive definition with respect to a specification $S$ is the coinduction principle, that is, by showing that $S$ is \emph{consistent} with respect to the coinductive definition. This means that each element of $S$ should be the consequence of a rule whose premises are in $S$ as well.  In our case, since the definition of  $\gen{\conf}{\rgr}{\infty}$ is not purely coinductive, but refined by the corule, completeness needs to  be  proved by the \emph{bounded coinduction} principle \citep{AnconaDZ@esop17,Dagnino19}, a generalization of the coinduction principle. Namely, besides proving that $S$ is consistent, we have to prove that $S$ is \emph{bounded}, that is, each element of $S$ can be derived by the inference system consisting of the rules \emph{and the corules}, in our case, only \refToRule{co-div}, interpreted inductively.

The proof of \cref{theo:complete-div} modularly  relies on two results. The former (\cref{theo:div-cons}) is the instantiation of a general result proved  by \cite{AnconaDZ17}  (Theorem 3.3) by bounded coinduction. For  the  reader's convenience, to illustrate the proof technique in a self-contained way, we report here statement and proof for our specific case. 
Namely, \cref{theo:div-cons} states completeness of diverging configurations with respect to any family of configurations which satisfies the \emph{progress-$\infty$ property}  (this name is chosen to suggest that, for a non-converging well-typed configuration, the construction of a proof tree can never get stuck).  The latter (\cref{theo:progress}) is the progress-$\infty$ property  for our type system. 

\begin{theorem}[Progress-$\infty\Rightarrow$ Completeness-$\infty$]\label{theo:div-cons}
For each grade $\rgr$, let $\CSet_\rgr$ be a set of configurations, and set $\CSet_\rgr^\infty=\{\conf\in\CSet_\rgr\mid \not\exists\ \Conf{\bv}{\env}\ \mbox{such that}\ \ind{\conf}{\rgr}{\Conf{\bv}{\env}} \}$.
If the following condition holds:\footnote{Keep in mind that $\NWreduce{\conf}{\rgr}{\res}$ denotes just the judgment (a triple), whereas  $\ind{\conf}{\rgr}{\res}$ and  $\gen{\conf}{\rgr}{\res}$ denote derivability of the judgment (\cref{def:deriv-ind} and \cref{def:deriv-gen}, respectively). }  
\begin{quoting}
\refToRule{progress-$\infty$}\ $\conf\in\CSet_\rgr^\infty$ implies that $\NWreduce{\conf}{\rgr}{\infty}$ is the consequence of a  rule where, for all premises of shape $\NWreduce{\conf'}{\sgr}{\infty}$, $\conf'\in\CSet^\infty_\sgr$, and, for all premises of shape $\NWreduce{\conf'}{\sgr}{\res}$,  \mbox{with $\res\neq\infty$, $\ind{\conf'}{\sgr}{\res}$. }
\end{quoting}
then $\conf\in\CSet^\infty_\rgr$ implies $\gen{\conf}{\rgr}{\infty}$.
\end{theorem}
\begin{proof}
We  set  $S=\{\NWreduce{\conf}{\rgr}{\infty}\mid\conf\in\CSet^\infty_\rgr\}\cup\{\NWreduce{\conf}{\rgr}{\res}\mid \res\neq\infty, \ind{\conf}{\rgr}{\res}\}$, and prove that, for each $\NWreduce{\conf}{\rgr}{\res}\in S$, we have $\gen{\conf}{\rgr}{\res}$, by bounded coinduction. We have to prove \mbox{two conditions.}
\begin{enumerate}
\item $S$ is consistent, that is, each $\NWreduce{\conf}{\rgr}{\res}$ in $S$ is the consequence of a rule whose premises are in $S$ as well. We reason by cases:
\begin{itemize}
\item For each $\NWreduce{\conf}{\rgr}{\infty}\in S$, by the \refToRule{progress-$\infty$} hypothesis it is the consequence of a rule where, for all premises of shape $\NWreduce{\conf'}{\sgr}{\infty}$, $\conf'\in\CSet^\infty_\sgr$, hence $\NWreduce{\conf'}{\sgr}{\infty}\in S$, and, for all premises of shape $\NWreduce{\conf'}{\sgr}{\res}$,  with $\res\neq\infty$, $\ind{\conf'}{\sgr}{\res}$, hence $\NWreduce{\conf'}{\sgr}{\res}\in S$ as well.
\item For each $\NWreduce{\conf}{\rgr}{\Conf{\bv}{\env}}\in S$, we have $\ind{\conf}{\rgr}{\Conf{\bv}{\env}}$, hence this judgment is the consequence of a rule  in \cref{fig:nw-sem} where for each premise, necessarily of shape $\NWreduce{\conf'}{\rgr'}{\Conf{\bv'}{\env'}}$, we have $\ind{\conf'}{\rgr'}{\Conf{\bv'}{\env'}}$, hence $\NWreduce{\conf'}{\rgr'}{\Conf{\bv'}{\env'}}\in S$.
\end{itemize}
\item $S$ is bounded, that is, each $\NWreduce{\conf}{\rgr}{\res}$ in $S$ can be inductively derived (has a finite proof tree) using the rules and the corule in \cref{fig:div}. This is trivial, since, for $\res=\infty$, the judgment can be directly derived by \refToRule{co-div}, and, for $\res\neq\infty$,  since $\ind{\conf}{\rgr}{\res}$, this holds by definition.
\end{enumerate}
\end{proof}

Thanks to the theorem above, to prove type soundness  (formulated as in \cref{theo:complete-div}) it is enough to prove the  progress-$\infty$ property  for well-typed configurations which do not converge. 

Set $\WT_\rgr=\{\conf\mid{\IsWFconf{\Gamma}{\conf}{\Graded{\tau}{\rgr}}}\ \mbox{for some}\ \Gamma,\tau\}$, and, accordingly with the notation in \cref{theo:div-cons}, $\WTInfty_\rgr=\{\conf\mid\conf\in\WT^\rgr\ \mbox{and}\ \not\exists\ \Conf{\bv}{\env}\ \mbox{such that}\ \gen{\conf}{\rgr}{\Conf{\bv}{\env}}\}$.

\begin{theorem}[Progress-$\infty$]\label{theo:progress} If $\conf\in\WTInfty_\rgr$, then $\NWreduce{\conf}{\rgr}{\infty}$ is the consequence of a  rule where, for all premises of shape $\NWreduce{\conf'}{\sgr}{\infty}$, $\conf'\in\WTInfty_\sgr$, and, for all premises of shape $\NWreduce{\conf'}{\sgr}{\res}$,  with $\res\neq\infty$, $\ind{\conf'}{\sgr}{\res}$. 
\end{theorem}

We derive this theorem from the next one,  which needs the following notations:
\begin{itemize}
\item We use the metavariable $\envE$ for environments where grades have been erased, hence they are maps from variables into values.
\item We write  $\erase{\env}$ for  the environment obtained from $\env$ by erasing grades   and grade contexts.
\item The reduction relation $\NWreduceE{}{}$ over pairs $\Conf{\ve}{\envE}$ and $\Conf{\e}{\envE}$ is defined by the metarules in \cref{fig:nw-sem} where we remove side conditions involving grades   and grade contexts . That is, such relation models standard semantics. 
\end{itemize}

The following theorem states that, if a configuration is well-typed, and (ignoring the grades) reduces to a result, then there is a corresponding graded reduction, leading to a result which is well-typed as well.
\begin{theorem}\label{theo:adjustE}
If $\IsWFConf{\Gamma}{\e}{\env}{\Graded\tau\rgr}$ and, set $\envE=\erase{\env}$, 
$\NWreduce{\Conf\e{\envE}}{}{\Conf{\val}\envE'}$,  then there exist $\env'$, $\cctx$  and $\Gamma'$ such that  
$\NWreduce{\Conf\e{\env}}{\rgr}{\Conf{\bvPair{\cctx}{\val}}{\env'}}$ with $\erase{\env'}=\envE'$, and
$\IsWFConf{\Gamma'}{\bvPair{\cctx}{\val}}{\env'}{\Graded\tau\rgr}$ 
 consistent with $\IsWFConf\Gamma\e\env{\Graded\tau\rgr}$.  
\end{theorem}

\section{Related work}\label{sect:related}
As mentioned, the contributions closest to this work, since they present a resource-aware semantics,  are  those by \cite{ChoudhuryEEW21,BianchiniDGZ@ECOOP23,TorczonSAAVW23}, 
and the conference version of this paper \citep{BianchiniDGZ@OOPSLA23}.
 \cite{ChoudhuryEEW21} develop   \textsc{GraD}, a graded dependent type system that includes functions,
tensor products, additive sums, and a unit type. The instrumented semantics is defined on typed terms, with the only aim to show the role of the type system. 

 In the work of \cite{BianchiniDGZ@ECOOP23},  
 which adopts a Java-like language, the semantics, also given in small-step style, is defined 
\emph{independently}  of the type system, in order  to provide a simple execution model taking into account usage of resources.
 In  the conference version of this paper \citep{BianchiniDGZ@OOPSLA23},   
 the language and type system  coincide with those  of the current paper, and the reduction semantics is big-step.   As a consequence, annotating subterms is no longer necessary;  however, differently from the  semantics considered here, this  reduction allows resource waste.

 The work of \cite{TorczonSAAVW23}    also  employs a big-step semantics.  Values are closures that ``save'' the
environment in the  presence  of delayed evaluation, similarly to our annotated values.  In that setting,  resource annotations can only
count down, and provide an upper bound on  the resources required for the  computation. 

The type system  presented  in this paper follows the same design  principles  of those  introduced in  \cite{BianchiniDGZ22,BianchiniDGZ@TCS23},  which are, however,  based on  a  Java-like underlying calculus. Those works also address two interesting issues not considered here. The  first  is the definition of a canonical construction  yielding   a unique grade algebra of \emph{heterogeneous} grades from a family of grade algebras, thus enabling different notions of resource usage to coexist within the same program. 
The  second  is  the provision of  linguistic support
for specifying \emph{user-defined} grades; for instance, grade annotations could be written themselves in Java, analogously to what happens for exceptions.

 More generally, in the study of  resource-aware type systems, \emph{coeffects}, that is, grades used as annotations in contexts, were   first  introduced and later  further analyzed by \cite{PetricekOM13,PetricekOM14}.  They develop  a  generic  coeffect system which augments the simply-typed $\lambda$-calculus with context annotations indexed by \emph{coeffect shapes}. 
The framework is highly abstract, and the authors focus on two instances: 
structural (per-variable) and flat (whole context) coeffects,  which are obtained   by specific choices of  coeffect  shapes.

 Subsequently, the notion has been generalized to an arbitrary assortment of usages, formally modeled by \emph{grades}, elements of an algebraic structure (\emph{grade algebra)} defined in slightly different ways in literature\footnote{Essentially variants of an ordered semiring.}  \citep{BrunelGMZ14,GhicaS14,McBride16,Atkey18,GaboardiKOBU16,AbelB20,OrchardLE19,WoodA22,DalLagoG22}. For instance, grades can be natural numbers counting how many times a resource is used, or even express non-quantitative properties, e.g., that a resource should be used with a given privacy level.  In this way, linear/affine type systems can be generalized to type systems parametric on the grade algebra, and the notion of resource-aware soundness can be generalized as well.

Most of the  subsequent  literature  has focused  on structural coeffects (grades), for which there is a clear algebraic description in terms of semirings. 
This was first noticed by  \cite{BrunelGMZ14},  who developed a framework for structural coeffects 
for a functional language, and by \cite{GhicaS14}.  
This approach is inspired by a generalization of the exponential modality of linear logic, see, e.g., \citep{BreuvartP15}. 
That is, the distinction between linear and  unrestricted  variables of linear systems is generalized to have variables  
decorated by coeffects (grades), that determine how much they can be used. 

In this setting, many advances have been made to combine grades with other programming features, such as 
computational effects \citep{GaboardiKOBU16,OrchardLE19,DalLagoG22}, 
dependent types \citep{Atkey18,ChoudhuryEEW21,McBride16}, and 
polymorphism \citep{AbelB20}. 
In all these papers, tracking usage through grades has practical benefits like erasure
of irrelevant terms and compiler optimizations.

 \cite{McBride16,WoodA22}   observed that contexts in a structural coeffect system form a module over the semiring of grades, event though they 
restrict themselves to free modules, that is, to structural coeffect systems. Recently,  \cite{BianchiniDGZS22} show  a significant non-structural instance, namely, a coeffect system to track sharing in the imperative paradigm.

%


\section{Conclusion}\label{sect:conclu}
We illustrated, on a paradigmatic calculus, how to enrich  the  semantics and type system of a language, parametrically on an arbitrary grade algebra,  so  as  to express and prove resource-aware soundness including no-waste. 
 To this end, we defined  two instrumented semantics, the latter being a refinement of the former. In both formulations, the resources consumed by a computation are explicitly tracked, and programs cannot consume more resources than those available. In the refined formulation, in addition, programs cannot waste resources. Since the instrumented semantics are non-deterministic, \emph{soundness-may} needs to be proved, that is, for a well-typed program there is a computation which is non-stuck, as formally stated in \cref{theo:soundness}, since, besides type errors, both resource exhaustion and waste are prevented. 

More precisely,  a non-stuck computation either terminates or diverges. In either case,   resource exhaustion never happens.
However, the property  that no resources remain in the result  (except for those that can be discarded) can only be meaningfully stated  for terminating computations. We leave to future work the investigation of a possible notion of no-waste for a non terminating computation. 

The resource-aware semantics is given in big-step style, extending the judgment to  explicitly model divergence. We prove  (resource-aware) soundness, by a significant, complex application of a schema previously introduced by \cite{AnconaDZ17}.

Most works on graded type systems introduce box/unbox operators in the syntax. Hence, programs typed with the type system proposed in this paper could not even be \emph{written} in these systems. A formal comparison with a calculus/type system based on boxing/unboxing is a challenging topic for future work. However, it is not obvious how to make such a comparison, since works which present calculi with a similar expressive power to ours, e.g.,  those by \cite{BrunelGMZ14} and \cite{DalLagoG22},  do not include an instrumented semantics, so we should as a first step develop such a semantics.  

 We briefly discussed at the end of \cref{sect:nw-res} the relevance of the acyclicity assumption on enviroments. Besides being reasonable in the simple language considered in this paper, we expect such an assumption to still hold when adding more realistic language features, provided that the calculus remains pure. Indeed, cycles are typically introduced by assignment. An extension of our approach to the imperative case is interesting also because a different notion of resource usage should likely be considered, e.g., rather than any time a variable occurrence needs to be replaced, any time memory needs to be accessed, even possibily distinguishing read from write accesses.  In this way, we could characterize as grades some properties which are of paramount importance in  the imperative  context, such as \emph{mutability/immutability} or \emph{uniqueness} opposed to \emph{linearity}, revisiting what is discussed by \cite{MarshallVO22}.

Finally, some more general topics to be investigated are type inference, for which the challenging feature are recursive functions, and combination with \emph{effects}, as  investigated by \cite{DalLagoG22}. 

Concerning implementation, and mechanization of proofs, which of course would be beneficial developments as well, we just mention the Agda library, available at \url{https://github.com/LcicC/inference-systems-agda} and described  in  the paper of \cite{CicconeDZ21},  which allows one to specify (generalized) inference systems. One of the examples provided in  such a paper is, as here,  a big-step semantics including divergence.

\paragraph*{Acknowledgments}
The authors would like to thank the OOPSLA  and JFP  anonymous referees who provided useful and detailed comments on  previous versions  of the paper. 
This work 
has the financial support of the University of Eastern Piedmont.

\bibliographystyle{ACM-Reference-Format}
\bibliography{main}


\section{Proofs of \cref{sect:nw-sem}}
\begin{proofOf}{prop:paths} 
(\ref{prop:paths:0}) 
We have that $\weight\graph\p = \prod_{i \in 1..n} \graph(\x_i,\x_{i+1})$. 
Then, we derive $\weight\graph\p = \rzero$, because 
$\graph(\x_{j-1},\x_j)$  is equal to $\rzero$. 

(\ref{prop:paths:1}) 
Let $\p\in\Paths\x\y$, then, because $\x\ne\y$ and $\y\notin\Nodes\graph$, by \cref{prop:paths:0}, we have $\weight\graph\p = \rzero$. 
Hence, this proves that $\p\notin\Pathsnz\graph\x\y$, as needed. 

(\ref{prop:paths:02})
Let $\p = \x_1\cdots\x_k \in\Paths\x\y$ be a path with $k\geq n+2$. 
Then we have 
$\weight\graph\p = \prod_{i\in 1..k-1} \graph(\x_i,\x_{i+1})$. 
We distinguish two cases.
If there is $j\in 2..k$ such that $\x_j\notin\Nodes\graph$, then the thesis follows by \cref{prop:paths:0}. 
Otherwise, we have that $\x_2,\ldots,\x_k\in \Nodes\graph$, but, since $k\geq n+2$ and $\Nodes\graph$ has $n$ elements, there must be a repeated variable, i.e., 
there are indices $2\le j < h \leq k$ such that $\x_j = \x_h$. 
Because $\graph$ is acyclic, we deduce that $\weight\graph{\x_j\cdots\x_h} = \prod_{i\in j..h-1} \graph(\x_i,\x_{i+1}) = \rzero$ and this implies $\weight\graph\p = \rzero$, as needed. 

(\ref{prop:paths:2})
It is an immediate consequence of (\ref{prop:paths:02}). 
\end{proofOf} 

\begin{proofOf}{lem:nocycles}
The fact that $(\graph\rmul\Gstar\graph)(\x,\x) = (\Row\graph\x\rmul\Gstar\graph)(\x)$ is immediate, 
hence we only prove that $(\Row\graph\x\rmul\Gstar\graph)(\x) = \rzero$. 
By definition, we have
\[
(\Row{\graph}{\x}\rmul\Gstar{\graph})(\x) 
  = \sum_{\y\in\Vars} \Row{\graph}{\x}(\y)\rmul\Gstar{\graph}(\y,\x)
  = \sum_{\y\in\Vars} \graph(\x,\y)\rmul\Gstar{\graph}(\y,\x)
\] 
To conclude the proof, it suffices to show that, 
for each $\y\in\Vars$, $\graph(\x,\y)\rmul\Gstar{\graph}(\y,\x) = \rzero$. 
Indeed, we have 
\[
\graph(\x,\y)\rmul\Gstar\graph(\y,\x) 
  = \graph(\x,\y)\rmul\sum_{\p\in\Paths\y\x} \weight\graph\p 
  = \sum_{\p\in\Paths\y\x} \graph(\x,\y)\rmul\weight\graph\p 
  = \sum_{\p\in\Paths\y\x} \weight\graph{\x\p}
\] 
because $\Start\p = \y$. 
Then, since $\End\p = \x$, the path $\x\p$ is  a cycle and so $\weight\graph{\x\p} = \rzero$, because $\graph$ is acyclic. 
Therefore, we conclude $\graph(\x,\y)\rmul\Gstar\graph(\y,\x) = \rzero$, as needed. 
\end{proofOf}

\begin{proofOf}{prop:graph-iterate}
We proceed by induction on $n$. 
If $n = 0$, then $\graph^0 = \Id$ and, since the only paths of length $0$ have shape $\x$, the thesis is immediate. 
Let $n = k+1$. 
Using the induction hypothesis, we derive 
\begin{align*} 
\graph^{k+1}(\x,\y) 
  &= (\graph\rmul\graph^k)(\x,\y) 
   = \sum_{z\in\Vars} \graph(\x,z)\rmul \graph^k(z,\y)
\\
  &= \sum_{z\in\Vars}\graph(\x,z)\rmul \sum_{\p\in\Paths[k]{z}{\y}} \weight\graph\p 
   = \sum_{z\in\Vars}\sum_{\p\in\Paths[k]{z}{\y}} \weight\graph{\x\cdot \p}
\\
  &= \sum_{\p'\in\Paths[k+1]\x\y} \weight\graph{\p'} 
\end{align*}
where the last step relies on the fact that 
$\p'\in\Paths[k+1]\x\y$ iff $\p' = \x\cdot \p$ and $\p \in \Paths[k]{z}{\y}$ for some $z\in\Vars$. 
\end{proofOf} 

 For the proof of \cref{theo:balance},  we introduce the following notations: $\ZeroOut{\cctx}{\env}$ if $\cctx(\x)=\rzero$ for each $\x\not\in\dom{\env}$, that is, the support of $\cctx$ is a subset of the domain of $\env$; 
$\Disjoint{\cctx}{\env}$ if $\added(\x)=\rzero$ for all $\x\in\dom{\env}$, that is, the support of $\cctx$ and the domain of $\env$ are disjoint.

\begin{lemma}\label{lem:restrict-env}
If $\ZeroOut{\cctx}{\env}$ then $\cctx\ctxmul\Graph{\env\envSum\env'}=\cctx\ctxmul\Graph{\env}$.
\end{lemma}
\begin{proof}
Consider $\x\in\dom{\env'}$. If $\x\in\dom{\env}$, since $\OKSum{\env}{\env'}$, we have $\env(\x)=\_:\bv$ and $\env'(\x)=\_:\bv$, hence $\cctx\ctxmul\Graph{\env\envSum\env'}(\x)=\cctx\ctxmul\Graph{\env}(\x)$. If $\x\not\in\dom{\env}$ we have $\cctx\ctxmul\Graph{\env\envSum\env'}(\x)=\rzero$. Hence we get the thesis.
\end{proof}

\begin{proofOf}{theo:balance}
 
We prove the following generalized statement:
\begin{quoting}
If $\red{\E}{\env}{\rgr}{\bvPair{\cctx'}{\val}}{\env'}$ then, for each $\rgr'\rord\rgr$, there exist $\cctx,\used,\added$, with   $\cctx$ honest for $\E$,  $\ZeroOut{\cctx}{\env}$, 
and $\Disjoint{\added}{\env}$,  such that
the following conditions hold.
\begin{enumerate}
\item\label{theo:1} $\rgr'\ctxmul\cctx'\ctxmul\GraphStar{\env'}\ctxord\CCTX{\env'}$
\item\label{theo:2}  $\cctx\ctxmul\GraphStar{\env}\ctxord\CCTX{\env}$
\item\label{theo:3} $\CCTX{\env'}\ctxsum\used\ctxord\CCTX{\env}\ctxsum\added$
\item\label{theo:4} $\cctx\ctxmul\GraphStar{\env} \ctxsum\added\rord\rgr'\ctxmul\cctx'\ctxmul\GraphStar{\env'}\ctxsum\used$
\end{enumerate}
\end{quoting}

First, we show that \cref{theo:2} is implied by \cref{theo:1}, \cref{theo:3}, and \cref{theo:4}. 
Indeed
\begin{quoting}
$\cctx\ctxmul\GraphStar{\env} \ctxsum\added\rord$ by \cref{theo:4}\\
$\rgr'\ctxmul\cctx'\ctxmul\GraphStar{\env'}\ctxsum\used\rord$ by \cref{theo:1}\\
$\CCTX{\env'}\ctxsum\used\ctxsum\added\ctxord$ by \cref{theo:3}\\
$\CCTX{\env}\ctxsum\added$ 
\end{quoting}
Hence, $\cctx\ctxmul\GraphStar{\env} \ctxsum\added\rord\CCTX{\env}\ctxsum\added$, and this implies $\cctx\ctxmul\GraphStar{\env}\ctxord\CCTX{\env}$ since $\Disjoint{\added}{\env}$.

Then, we prove \cref{theo:1}, \cref{theo:3}, and \cref{theo:4} by induction on the reduction relation. We show the most relevant cases.
\begin{description}
\item[\refToRule{var}] We have 
\begin{quoting}
(1) $\red{\x}{\env}{\rgr}{\bv}{\env'}$ with $\bv=\bvPair{\cctx'}{\val}$, $\env=\hat{\env},\EnvElem{\x}{\sgr}{\bv}$, $\env'=\hat{\env},\EnvElem{\x}{\sgr'}{\bv}$\\
(2) $\rgr\rsum\sgr'\rord \sgr$\\
(3) $\rgr\rmul\cctx'\rmul{\GraphStar{\env'}}\rord\CCTX{\env'}$
\end{quoting}
Take $\rgr'\rord\rgr$. We choose $\cctx=\used=\VarGrade{\x}{\rgr'}$ and $\added=\emptyset$, hence $\cctx$ honest for $\x$, $\ZeroOut{\cctx}{\env}$, and $\Disjoint{\added}{\env}$ trivially hold. 
\begin{description}
\item[\cref{theo:1}]  We have $\rgr'\ctxmul\cctx'\ctxmul\GraphStar{\env'}\ctxord\rgr\ctxmul\cctx'\ctxmul\GraphStar{\env'}\ctxord\CCTX{\env'}$ from (3).
\item[\cref{theo:3}]  We have 
\begin{quoting}
$\CCTX{\env'}\ctxsum(\VarGrade{\x}{\rgr'})=\CCTX{\hat{\env}},\VarGrade{\x}{\sgr'}\ctxsum(\VarGrade{\x}{\rgr'})=$\\
$\CCTX{\hat{\env}},(\VarGrade{\x}{\rgr'\rsum\sgr'})\rord\CCTX{\hat{\env}},(\VarGrade{\x}{\rgr\rsum\sgr'})\ctxord$ from (2) \\
$\CCTX{\hat{\env}},(\VarGrade{\x}{\sgr})=\CCTX{\env}$
\end{quoting}
\item[\cref{theo:4}] Noting that  ${(\VarGrade{\x}{\rgr'})\ctxmul\Graph{\env}=\rgr'\ctxmul\cctx'}$, by \cref{cor:GStarId} we have 
\begin{quoting}
$(\VarGrade{\x}{\rgr'})\ctxmul\GraphStar{\env}={(\VarGrade{\x}{\rgr'})\ctxmul(\Id\ctxsum\Graph{\env}\ctxmul\GraphStar{\env})}=(\VarGrade{\x}{\rgr'})\ctxsum(\VarGrade{\x}{\rgr'})\ctxmul\Graph{\env}\ctxmul\GraphStar{\env}={\VarGrade{\x}{\rgr'}\ctxsum\rgr'\ctxmul\cctx'\ctxmul\GraphStar{\env}}$. 
\end{quoting}
\end{description}

\item[\refToRule{pair}] We have
\begin{quoting}
(1) $\red{\ve_1}{\env_1}{\rgr\rmul\rgr_1}{\bvPair{\cctx'_1}{\val_1}}{\env_1'}$\\
(2) $\red{\ve_2}{\env_2}{\rgr\rmul\rgr_2}{\bvPair{\cctx'_2}{\val_2}}{\env_2'}$\\
(3) ${\red{\PairExp{\rgr_1}{\ve_1}{\ve_2}{\rgr_2}}{\env}{\rgr}{\bvPair{\cctx'}{\PairExp{\rgr_1}{\val_1}{\val_2}{\rgr_2}}}{\env_1'\envSum\env_2'}}$ with $\cctx'=\rgr_1\ctxmul\cctx'_1\ctxsum\rgr_2\ctxmul\cctx'_2$\\
(4) $\env_1\envSum\env_2\ctxord\env$
\end{quoting}
Take $\rgr'\rord\rgr$, hence $\rgr'\rmul\rgr_1\ctxord\rgr\rmul\rgr_1$ and $\rgr'\rmul\rgr_2\ctxord\rgr\rmul\rgr_2$.\\
By inductive hypothesis we have that there exist $\CCTX{1}, \CCTX{2}, \used[1],\used[2], \added[1],\added[2]$, with $\cctx_1$ honest for $\ve_1$, $\cctx_2$ honest for $\ve_2$, $\ZeroOut{\cctx_1}{\env_1}$, $\ZeroOut{\cctx_2}{\env_2}$, $\Disjoint{\added[1]}{\env_1}$, and $\Disjoint{\added[2]}{\env_2}$, such that
\begin{enumerate}
\item[1a] $\rgr'\rmul\rgr_1\ctxmul\cctx'_1\ctxmul\GraphStar{\env'_1}\ctxord\CCTX{\env'_1}$
\item[1b] $\rgr'\rmul\rgr_2\ctxmul\cctx'_2\ctxmul\GraphStar{\env'_2}\ctxord\CCTX{\env'_2}$
\item[2a] $\CCTX{1}\ctxmul\GraphStar{\env_1}\ctxord\CCTX{\env_1}$
\item[2b] $\CCTX{2}\ctxmul\GraphStar{\env_2}\ctxord\CCTX{\env_2}$
\item[3a] $\CCTX{\env_1'}\ctxsum\used[1]\ctxord\CCTX{\env_1}\ctxsum\added[1]$
\item[3b] $\CCTX{\env'}\ctxsum\used[2]\ctxord\CCTX{\env_2}\ctxsum\added[2]$
\item[4a] $\CCTX{1}\ctxmul\GraphStar{\env_1} \ctxsum\added[1]\rord\rgr'\rmul\rgr_1\ctxmul\CCTX{1}'\ctxmul\GraphStar{\env_1'}\ctxsum\used[1]$
\item[4b] $\CCTX{2}\ctxmul\GraphStar{\env_2}\ctxsum\added[2]\rord\rgr'\rmul\rgr_2\ctxmul\cctx'_2\ctxmul\GraphStar{\env'_2}\ctxsum\used[2]$
\end{enumerate}
We choose $\cctx=\cctx_1\ctxsum\cctx_2$, $\used=\used[1]\ctxsum\used[2]$, and $\added=\added[1]\ctxsum\added[2]$. Clearly
$\cctx$ is honest for $\PairExp{\rgr_1}{\ve_1}{\ve_2}{\rgr_2}$, $\ZeroOut{\cctx}{\env}$, and $\Disjoint{\added}{\env}$. 
\begin{description}
\item[\cref{theo:1}] We have 
\begin{quoting}
$\rgr'\ctxmul(\rgr_1\ctxmul\cctx'_1\ctxsum\rgr_2\ctxmul\cctx'_2)\ctxmul\GraphStar{\env_1'\envSum\env'_2}=$\\
$\rgr'\ctxmul\rgr_1\ctxmul\cctx'_1\ctxmul\GraphStar{\env_1'\envSum\env'_2}\ctxsum\rgr'\ctxmul\rgr_2\ctxmul\cctx'_2\ctxmul\GraphStar{\env_1'\envSum\env'_2}=$ by \cref{lem:restrict-env} since $\ZeroOut{\cctx'_1}{\env'_1}$ and $\ZeroOut{\cctx'_2}{\env'_2}$\\
$\rgr'\ctxmul\rgr_1\ctxmul\cctx'_1\ctxmul\GraphStar{\env_1'}\ctxsum\rgr'\ctxmul\rgr_2\ctxmul\cctx'_2\ctxmul\GraphStar{\env'_2}\rord$ by [1a] and [1b]\\
$\CCTX{\env'_1}\ctxsum\CCTX{\env'_2}=\CCTX{\env_1'\envSum\env'_2}$
\end{quoting}

\item[\cref{theo:3}] 
We have to prove $\CCTX{\env'_1\ctxsum\env'_2}\ctxsum\used\ctxord\CCTX{\env}\ctxsum\added$. We have:
\begin{quoting}
$\CCTX{\env'_1\ctxsum\env'_2}\ctxsum\used[1]\ctxsum\used[2]=\CCTX{\env'_1}\ctxsum\CCTX{\env'_2}\ctxsum\used[1]\ctxsum\used[2]\ctxord$ by [3a] and [3b]\\
$\CCTX{\env_1}\ctxsum\added[1]\ctxsum\CCTX{\env_2}\ctxsum\added[2]=\CCTX{\env_1\ctxsum\env_2}\ctxsum\added$
\end{quoting}

\item[\cref{theo:4}] 
We have to prove $(\cctx_1\ctxsum\cctx_2)\ctxmul\GraphStar{\env} \ctxsum\added\ctxord\rgr'\ctxmul(\rgr_1\ctxmul\cctx'_1\ctxsum\rgr_2\ctxmul\cctx'_2)\ctxmul\GraphStar{\env'}\ctxsum\used$.
We have
\begin{quoting}
$(\cctx_1\ctxsum\cctx_2)\ctxmul\GraphStar{\env} \ctxsum\added[1]\ctxsum\added[2]=$\\
$\cctx_1\ctxmul\GraphStar{\env}\ctxsum\cctx_2\ctxmul\GraphStar{\env}\ctxsum\added[1]\ctxsum\added[2]=$ by \cref{lem:restrict-env} since $\ZeroOut{\cctx_1}{\env_1}$ and $\ZeroOut{\cctx_2}{\env_2}$ \\
$\cctx_1\ctxmul\GraphStar{\env_1}\ctxsum\cctx_2\ctxmul\GraphStar{\env_2}\ctxsum\added[1]\ctxsum\added[2]\rord$ by [4a] and [4b]\\
$\rgr'\rmul\rgr_1\ctxmul\CCTX{1}'\ctxmul\GraphStar{\env_1'}\ctxsum\used[1]\ctxsum\rgr'\rmul\rgr_2\ctxmul\cctx'\ctxmul\GraphStar{\env_2'}\ctxsum\used[2]$ \\
$\rgr'\rmul(\rgr_1\ctxmul\CCTX{1}'\ctxmul\GraphStar{\env_1'}\ctxsum\rgr_2\ctxmul\cctx_2'\ctxmul\GraphStar{\env_2'})\ctxsum\used=$ by \cref{lem:restrict-env} since $\ZeroOut{\cctx'_1}{\env'_1}$ and $\ZeroOut{\cctx'_2}{\env'_2}$ \\
$\rgr'\rmul(\rgr_1\ctxmul\CCTX{1}'\ctxmul\GraphStar{\env'}\ctxsum\rgr_2\ctxmul\cctx'\ctxmul\GraphStar{\env'})\ctxsum\used=$ \\
$\rgr'\ctxmul(\rgr_1\ctxmul\cctx'_1\ctxsum\rgr_2\ctxmul\cctx'_2)\ctxmul\GraphStar{\env'}\ctxsum\used$
\end{quoting}
\end{description}

\item[\refToRule{let}]
We have 
\begin{quoting}
(1) $\red{\e_1}{\env_1}{\sgr}{\bvPair{\CCTX{1}'}{\val_1}}{\env_1'}$\\
(2) $ \NWreduce{\Conf{\Subst{\e_2}{\x'}{\x}}{\hat{\env}}}{\rgr}{\Conf{\bvPair{\cctx'}{\val}}{\env'}}$ with $\hat{\env}=\AddToEnv{(\env_1'\envSum\env_2)}{\x'}{\sgr}{\bvPair{\CCTX{1}'}{\val_1}}$, $\x'$ fresh\\
(3) $\NWreduce{\Conf{\Let{\x}{\e_1}{\e_2}}{\env}}{\rgr}{\Conf{\bvPair{\cctx'}{\val}}{\env'}}$\\
(4) $\env_1\envSum\env_2\ctxord\env$
\end{quoting}

Take $\rgr'\rord\rgr$. By inductive hypothesis we have that, for each $\sgr'\rord\sgr$, there exist $\CCTX{1}, \CCTX{2}, \used[1],\used[2], \added[1],\added[2]$, with $\cctx_1$ honest for $\e_1$, $\cctx_2$ honest for $\Subst{\e_2}{\x'}{\x}$, $\ZeroOut{\cctx_1}{\env_1}$, $\ZeroOut{\cctx_2}{\hat{\env}}$, $\Disjoint{\added[1]}{\env_1}$, and $\Disjoint{\added[2]}{\hat{\env}}$, such that
\begin{enumerate}
\item[1a] $\sgr'\ctxmul\cctx'_1\ctxmul\GraphStar{\env'_1}\ctxord\CCTX{\env'_1}$
\item[1b] $\rgr'\ctxmul\cctx'\ctxmul\GraphStar{\env'}\ctxord\CCTX{\env'}$
\item[3a] $\CCTX{\env_1'}\ctxsum\used[1]\ctxord\CCTX{\env_1}\ctxsum\added[1]$
\item[3b] $\CCTX{\env'}\ctxsum\used[2]\ctxord\CCTX{\hat{\env}}\ctxsum\added[2]$
\item[4a] $\CCTX{1}\ctxmul\GraphStar{\env_1} \ctxsum\added[1]\rord\sgr\ctxmul\CCTX{1}'\ctxmul\GraphStar{\env_1'}\ctxsum\used[1]$ 
\item[4b] $\CCTX{2}\ctxmul\GraphStar{\hat{\env}}\ctxsum\added[2]\rord\rgr\ctxmul\cctx'\ctxmul\GraphStar{\env'}\ctxsum\used[2]$
\end{enumerate}
\begin{description}
\item[\cref{theo:1}] We have to prove $\rgr'\ctxmul\cctx'\ctxmul\GraphStar{\env'}\ctxord\CCTX{\env'}$, and this is [1b].

\item[\cref{theo:3}] 
We choose $\used=\used[1]\ctxsum\used[2]$, $\added=\added[1]\ctxsum\added[2]\ctxsum\VarGrade{\x'}{\sgr}$. Clearly $\Disjoint{\added}{\env}$. We have to prove 
\begin{quoting}
$\CCTX{\env'}\ctxsum\used[1]\ctxsum\used[2]\ctxord\CCTX{\env}\ctxsum\added[1]\ctxsum\added[2]\ctxsum\VarGrade{\x'}{\sgr}$.
\end{quoting}
We have
\begin{quoting}
$\CCTX{\env'}\ctxsum\used[1]\ctxsum\used[2]\ctxord$ by [3b]\\
$\CCTX{\hat{\env}}\ctxsum\added[2]\ctxsum\used[1]=\CCTX{\env'_1}\ctxsum\CCTX{\env_2}\ctxsum(\VarGrade{\x'}{\sgr})\ctxsum\added[2]\ctxsum\used[1]\ctxord$ by [3a]\\
$\CCTX{\env_1}\ctxsum\added[1]\ctxsum\CCTX{\env_2}\ctxsum(\VarGrade{\x'}{\sgr})\ctxsum\added[2]\rord\CCTX{\env}\ctxsum\added[1]\ctxsum\added[2]\ctxsum(\VarGrade{\x'}{\sgr})$
\end{quoting}

\item[\cref{theo:4}] 
We have to prove that there exist $\cctx$ such that
\begin{quoting}
$\cctx\ctxmul\GraphStar{\env} \ctxsum\added[1]\ctxsum\added[2]\ctxsum(\VarGrade{\x'}{\sgr'})\rord\rgr'\ctxmul\cctx'\ctxmul\GraphStar{\env'}\ctxsum\used[1]\ctxsum\used[2]$
\end{quoting}
Since $\CCTX{2}\ctxmul\GraphStar{\env_2}\ctxord\CCTX{\env_2}$ follows from [1b], [3b], and [4b], and $\x'$ is fresh, $\CCTX{2}(\x')=\sgr'\rord\sgr$. Hence, there exists $\CCTX{1}$ such that [4a] holds for such $\sgr'$, that is, 
$\CCTX{1}\ctxmul\GraphStar{\env_1} \ctxsum\added[1]\rord\sgr'\ctxmul\CCTX{1}'\ctxmul\GraphStar{\env_1'}\ctxsum\used[1]$.\\
Let us decompose $\cctx_2$ as $\cctx_2=\cctx_2^\env\ctxsum\cctx_2^{\bar{\env}}$ such that $\ZeroOut{\cctx_2^\env}{\env}$ and $\Disjoint{\cctx_2^{\bar{\env}}}{\env}$.
Moreover:
\begin{quoting}
$(\star)$\Space$\cctx_2^{\bar{\env}}\ctxmul\GraphStar{\hat{\env}}=(\VarGrade{\x'}{\sgr'})\ctxsum\sgr'\ctxmul\CCTX{1}'\ctxmul\GraphStar{\env_1'}$
\end{quoting}
Then, 
we choose $\cctx=\cctx_1\ctxsum\cctx^\env_2$, and we have: 
\begin{quoting}
$(\cctx_1\ctxsum\cctx^\env_2)\ctxmul\GraphStar{\env} \ctxsum\added[1]\ctxsum\added[2]\ctxsum(\VarGrade{\x'}{\sgr'})=$ since $\ZeroOut{\cctx_1}{\env_1}$ and $\ZeroOut{\cctx^\env_2}{\hat{\env}}$ \\
$\cctx_1\ctxmul\GraphStar{\env_1}\ctxsum\cctx^\env_2\ctxmul\GraphStar{\hat{\env}}\ctxsum\added[1]\ctxsum\added[2]\ctxsum(\VarGrade{\x'}{\sgr'})\rord$ \mbox{from [4a]}\\
$\sgr'\ctxmul\CCTX{1}'\ctxmul\GraphStar{\env_1'}\ctxsum\used[1]\ctxsum\cctx^\env_2\ctxmul\GraphStar{\hat{\env}}\ctxsum\added[2]\ctxsum(\VarGrade{\x'}{\sgr'})=$ from $(\star)$\\
$\used[1]\ctxsum\CCTX{2}\ctxmul\GraphStar{\hat{\env}}\ctxsum\added[2]\rord$ from [4b]\\
$\used[1]\ctxsum\rgr\ctxmul\cctx'\ctxmul\GraphStar{\env'}\ctxsum\used[2]$
\end{quoting}
\end{description}
\end{description}
\end{proofOf}

\begin{proofOf}{prop:env-typing} 
{\em If $\IsWFEnv{\cctx}{\env}{\Delta}$ then $\CCTX{\Delta}\ctxsum\CCTX{\env}\ctxmul\Graph{\env}\ctxord\CCTX{\env}$.}\\
Set $\env=\EnvElem{\x_1}{\rgr_1}{\bvPair{\cctx_1}{\val_1}},\ldots,\EnvElem{\x_n}{\rgr_n}{\bvPair{\cctx_n}{\val_n}}$.
To derive $\IsWFEnv{\cctx}{\env}{\Delta}$ we have 
necessarily applied rule \refToRule{t-env}, and then \refToRule{t-val} to the premises, hence 
$\CCTX{\cctx}=\VarGrade{\x_1}{\rgr_1},\ldots,\VarGrade{\x_n}{\rgr_n}$,  and, for some $\cctx_1, \ldots, \cctx_n$, we have $\cctx_i\rord\CCTX{\cctx_i}$ for ${i \in 1..n}$, and $\Delta\ctxsum\sum_{i\in 1..n}\rgr_i\rmul\cctx_i \ctxord \cctx$.
From the last condition, $\CCTX{\Delta}\ctxsum\sum_{i\in 1..n}\rgr_i\rmul\CCTX{\cctx_i}\ctxord \CCTX{\cctx}$, hence, since $\CCTX{\cctx}=\CCTX{\env}$, we get $\CCTX{\Delta}\ctxsum\sum_{i\in 1..n}\rgr_i\rmul\cctx_i\ctxord \CCTX{\env}$.
Finally, we show that $\sum_{i\in 1..n}\rgr_i\rmul\cctx_i=\CCTX{\env}\rmul\Graph{\env}$. Indeed, for $i\in 1..n$, $\rgr_i=\CCTX{\env}(\x_i)$, and $\cctx_i(\x)=\Graph{\env}(\x_i,\x)$, hence 
we have $(\sum_{i\in 1..n}\rgr_i\rmul\cctx_i)(\x)=\sum_{i\in 1..n}\CCTX{\env}(\x_i)\rmul\Graph{\env}(\x_i,\x)$, which is equal to $(\CCTX{\env}\rmul\Graph{\env})(\x) = \sum_{\y\in\Vars} \CCTX{\env}(\y)\cdot\Graph{\env}(\y,\x)$ since $\CCTX{\env}(\y)\neq\rzero$ implies $\y=\x_i$ for some $i\in 1..n$.
\end{proofOf}


\section{Proofs of \cref{sect:soundness}}

\begin{lemma}[Inversion for value expressions]\label[lemma]{lem:inversion}\
\begin{enumerate}

\item \label{lem:inversion:x} If $\IsWFExp{\Gamma}{\x}{\Graded{\tau}{\rgr}}$ then $\Gamma=\VarGradeType{\x}{\rgr'}{\tau},\Gamma'$ and
$\VarGradeType{\x}{\rgr'}{\tau}\ctxord\Gamma$ and $\rgr\rord\rgr'$.

\item \label{lem:inversion:fun} If $\IsWFExp{\Gamma}{\RecFun{\f}{\x}{\e}}{\Graded{\tau}{\rgr}}$ then $\rgr'\ctxmul\Gamma'\ctxord\Gamma$ and $\tau = \funType{\GradedInd{\tau}{1}{\rgr}}{\sgr}{\GradedInd{\tau}{2}{\rgr}}$ such that \\
${\IsWFExp{\Gamma',\VarGradeType{\f}{\sgr}{\funType{\GradedInd{\tau}{1}{\rgr}}{\sgr}{\GradedInd{\tau}{2}{\rgr}}},\VarGradeType{\x}{\rgr_1}{\tau_1}}{\e}{\GradedInd{\tau}{2}{\rgr}}}$ and $\rgr\rord\rgr'$.

\item \label{lem:inversion:unit} If $\IsWFExp{\Gamma}{\unit}{\Graded{\tau}{\rgr}}$ then $\tau = \Unit$ and $\emptyset\ctxord\Gamma$.

\item \label{lem:inversion:pair} If $\IsWFExp{\Gamma}{\PairExp{\rgr_1}{\ve_1}{\ve_2}{\rgr_2}}{\Graded{\tau}{\rgr}}$ then $\rgr'\ctxmul(\Gamma_1 \ctxsum \Gamma_2)\ctxord\Gamma$, $\tau = \PairT{\GradedInd{\tau}{1}{\rgr}}{\GradedInd{\tau}{2}{\rgr}}$ and $\rgr \rord \rgr'$ such that \mbox{$\IsWFExp{\Gamma_1}{\ve_1}{\GradedInd{\tau}{1}{\rgr}}$, $\IsWFExp{\Gamma_2}{\ve_2}{\Graded{\tau_2}{\rgr_2}}$.}

\item \label{lem:inversion:inl} If $\IsWFExp{\Gamma}{\Inl{\rgr_1}{\ve}}{\Graded{\tau}{\rgr}}$ then  $\rgr'\ctxmul\Gamma'\ctxord\Gamma$, $\tau = \SumT{\GradedInd{\tau}{1}{\rgr}}{\GradedInd{\tau}{2}{\rgr}}$ and $\rgr \rord \rgr'$ such that $\IsWFExp{\Gamma'}{\ve}{\GradedInd{\tau}{1}{\rgr}}$.

\item \label{lem:inversion:inr} If $\IsWFExp{\Gamma}{\Inr{\rgr_2}{\ve}}{\Graded{\tau}{\rgr}}$ then  $\rgr'\ctxmul\Gamma'\ctxord\Gamma$, $\tau = \SumT{\GradedInd{\tau}{1}{\rgr}}{\GradedInd{\tau}{2}{\rgr}}$ and $\rgr \rord \rgr'$ such that $\IsWFExp{\Gamma'}{\ve}{\GradedInd{\tau}{2}{\rgr}}$.
\end{enumerate}
\end{lemma}

\begin{lemma}[Canonical Forms]\label{lem:cf}\
\begin{enumerate}
\item \label{lem:cf:fun} If $\IsWFExp{\Gamma}{\ve}{\Graded{(\funType{\GradedInd{\tau}{1}{\rgr}}{\sgr}{\GradedInd{\tau}{2}{\rgr}})}{\rgr_3}}$ then 
$\ve  = \RecFun{\f}{\x}{\e}$.
\item \label{lem:cf:unit} If $\IsWFExp{\Gamma}{\ve}{\Unit}$ then 
$\ve  = \unit$.
\item \label{lem:cf:pair}If $\IsWFExp{\Gamma}{\ve}{\Graded{(\PairT{\GradedInd{\tau}{1}{\rgr}}{\GradedInd{\tau}{2}{\rgr}})}{\rgr_3}}$ then 
$\ve  = \PairExp{\rgr_1}{\ve_1}{\ve_2}{\rgr_2}$.
\item \label{lem:cf:sum} If $\IsWFExp{\Gamma}{\ve}{\Graded{\SumT{(\GradedInd{\tau}{1}{\rgr}}{\GradedInd{\tau}{2}{\rgr}})}{\rgr}}$ then 
$\ve  = \Inl{\rgr_1}{\ve_1}$ or $\ve  = \Inr{\rgr_2}{\ve_2}$.
\end{enumerate}
\end{lemma} 

\begin{lemma}[Renaming]
\label{lem:subst}
If $\IsWFExp{\Gamma,\VarGradeType{\x}{\rgr_1}{\tau_1}}{\e}{\GradedInd{\tau}{2}{\rgr}}$ then 
$\IsWFExp{\Gamma,\VarGradeType{\x'}{\rgr_1}{\tau_1}}{\Subst{\e}{\x'}{\x}}{\GradedInd{\tau}{2}{\rgr}}$  \mbox{with $\x'$ fresh.}
\end{lemma}

\begin{lemma}\label{lem:sumEnv}
If $\cctx=\cctx_1\ctxsum\cctx_2$ then $\cctx\ctxmul\GraphStar{\env}=(\cctx'_1\ctxmul\GraphStar{\env})\rsum(\cctx'_2\ctxmul\GraphStar{\env})$.
\end{lemma}

\begin{lemma}\label{lem:splitEnv}
Given $\env=\EnvElem{\x_1}{\rgr_1}{\bvPair{\CCTX{1}}{\val_1}},\ldots,\EnvElem{\x_n}{\rgr_n}{\bvPair{\CCTX{n}}{\val_n}}$, $\Gamma=\VarGradeType{\x_1}{\rgr_1}{\Graded{\tau_1}{\rgr_1}},\ldots,\VarGradeType{\x_n}{\rgr_n}{\Graded{\tau_n}{\rgr_n}}$, $\Delta\ctxord\Gamma$ and for all $i\in 1..n$ it exists $\Gamma_i$ such that $\IsWFExp{\Gamma_i}{\bvPair{\CCTX{i}}{\val_i}}{\Graded{\tau_i}{\rgr_i}}$ and $\CCTX{\Delta}\ctxmul\GraphStar{\env}\rord\CCTX{\env}$ then $\IsWFEnv{\Gamma}{\env}{\Delta}$.
\end{lemma}

\begin{lemma}[Environment splitting]\label{lem:ev-split}
If $\IsWFEnv\Gamma\env\Delta$ and $\Delta_1\rsum\Delta_2 \rord \Delta$, then 
there are $\env_1$ and $\env_2$ such that 
$\env_1\envSum\env_2\ctxord\env$ and 
$\IsWFEnv{\Gamma_1}{\env_1}{\Delta_1}$ and 
$\IsWFEnv{\Gamma_2}{\env_2}{\Delta_2}$ and these are consistent. 
\end{lemma}
\begin{proof}

Let $\env=\EnvElem{\x_1}{\rgr_1}{\bvPair{\cctx_1}{\val_1}},\ldots,\EnvElem{\x_n}{\rgr_n}{\bvPair{\cctx_n}{\val_n}}$ and $\Gamma=\VarGradeType{\x_1}{\rgr_1}{\tau_1},\dots,\VarGradeType{\x_n}{\rgr_n}{\tau_n}$.
From rule \refToRule{t-env}, 
$\sum_{i\in 1..n}\rgr_i\rmul\Gamma'_i\rsum\Delta\rord\Gamma$ where $\IsWFExp{\Gamma'_i}{\bvPair{\cctx_i}{\val_i}}{\Graded{\tau_i}{\rone}}$
 and $\cctx_i\rord\cctx_{\Gamma'_i}$, by rule \refToRule{t-val}. \\
Let  $\cctx'_i=\cctx_{\Delta_i}\rmul\GraphStar{\env}$ and 
$\env_i=\EnvElem{\x_1}{\cctx'_i(\x_1)}{\bvPair{\cctx_1}{\val_1}},\ldots,\EnvElem{\x_n}{\cctx'_i(\x_n)}{\bvPair{\cctx_n}{\val_n}}$ for  $i=1,2$.
From
$\Delta_1 \ctxsum \Delta_2\ctxord\Delta$ and  $\CCTX{\Delta}\ctxmul\GraphStar{\env}\ctxord\CCTX{\env}$ we get
$\cctx_{\Delta_1\rsum\Delta_2}\rmul\GraphStar{\env} \ctxord\CCTX{\env}$. Therefore
$\env_1\rsum\env_2\rord\env$. Let  $\Gamma_i=\VarGradeType{\x_1}{\cctx'_i(\x_1)}{\tau_1},\dots,\VarGradeType{\x_n}{\cctx'_i(\x_n)}{\tau_n}$, since
$\GraphStar{\env}=\GraphStar{\env_i}$  and by definition
$\CCTX{\env_i}\rmul\GraphStar{\env}=\CCTX{\env_i}\rmul\GraphStar{\env_i}=\CCTX{\env_i}$, applying \cref{lem:splitEnv} we get $\IsWFEnv{\Gamma_i}{\env_i}{\Delta_i}$. 

\end{proof}

\begin{proofOf}{theo:adjustV} 
Let $\env=\EnvElem{\x_1}{\rgr_1}{\bvPair{\cctx_1}{\val_1}},\ldots,\EnvElem{\x_n}{\rgr_n}{\bvPair{\cctx_n}{\val_n}}$ and $\Gamma=\VarGradeType{\x_1}{\rgr_1}{\tau_1},\dots,\VarGradeType{\x_n}{\rgr_n}{\tau_n}$.
From rule \refToRule{t-res}, for some $\Delta$, we get $\IsWFExp{\Delta}{\ve}{\Graded{\tau}{\rgr}}$ and $\IsWFEnv{\Gamma}{\env}{\Delta}$. From rule \refToRule{t-env}, 
$\sum_{i\in 1..n}\rgr_i\rmul\Gamma_i\rsum\Delta\rord\Gamma$ where $\IsWFExp{\Gamma_i}{\bvPair{\cctx_i}{\val_i}}{\Graded{\tau_i}{\rone}}$
 and $\cctx_i\rord\cctx_{\Gamma_i}$, by rule \refToRule{t-val}. \\
By induction on the syntax of $\ve$. We show only relevant cases.
\begin{description}
\item[$\ve=\x$]  From $\IsWFExp{\Delta}{\x}{\Graded{\tau}{\rgr}}$ and $\Delta\rord\Gamma$ then $\x=\x_k$ for some $k\in 1..n$ and $\tau_k=\tau$ and $\rgr_k=\sgr$ and
 $\IsWFExp{\Delta}{\x_k}{\Graded{\tau}{\rgr}}$. By  \refItem{lem:inversion}{x} $\Delta=\VarGradeType{\x_k}{\rgr'}{\tau_k},\Delta'$ with $\rgr\rord\rgr'$. 
 Let $\sgr'=\sum_{i\in 1..n}\sgr_i$ where $\sgr_i=\rgr_i\rmul\rgr_{i,k}$ and\\
 \centerline{
 $\env'=\EnvElem{\x_1}{\rgr_1}{\bvPair{\cctx_1}{\val_1}},\ldots,\EnvElem{\x_k}{\sgr'}{\bvPair{\cctx_k}{\val_k}},\ldots,\EnvElem{\x_n}{\rgr_n}{\bvPair{\cctx_n}{\val_n}}$.}
 We want to prove that  $\IsWFConf{\Gamma'}{\bvPair{\cctx_k}{\val}}{\env'}{\Graded{\tau_k}{\rgr}}$
where $\Gamma'$ is $\Gamma$ in which $\rgr_k$ is substituted by $\sgr'$. 
Let  $\Delta''=\rgr\rmul\Gamma_k$, from
  $\IsWFExp{\Gamma_k}{\bvPair{\cctx_k}{\val_k}}{\Graded{\tau}{\rone}}$ we get $\IsWFExp{\Delta''}{\bvPair{\cctx_k}{\val_k}}{\Graded{\tau}{\rgr}}$.
 Let $\Gamma_0=\sum_{i\in 1..n\ \wedge\  i\neq k}\rgr_i\rmul\Gamma_i$.
From $\sum_{i\in 1..n}\rgr_i\rmul\Gamma_i\rsum\Delta=\Gamma''\rsum\sgr\rmul\Gamma_k\rsum\Delta\rord\Gamma$ we have that $\rgr\rsum\sgr'\rord\sgr$.
Therefore $\Delta''\rsum\sgr'\rmul\Gamma_k=(\rgr\rsum\sgr')\rmul\Gamma_k\rord\sgr\rmul\Gamma_k$ and 
$\Delta''\rsum\Gamma''\rsum\sgr'\rmul\Gamma_k\rord\sum_{i\in 1..n}\rgr_i\rmul\Gamma_i\rsum\Delta$.  Since $\Gamma(\x_i)=\Gamma'(\x_i)$ for all $i\in 1..n$ and $i\neq k$
and $(\Delta''\rsum\Gamma''\rsum\sgr'\rmul\Gamma_k)(\x)=\sgr'$ we have that $\Delta''\rsum\Gamma''\rsum\sgr'\rmul\Gamma_k\rord\Gamma'$. Applying rule
\refToRule{t-res} we get   $\IsWFConf{\Gamma'}{\bvPair{\cctx_k}{\val}}{\env'}{\Graded{\tau_k}{\rgr}}$. Finally from \refItem{theo:no-waste-typing}{val} we get 
$\NoWaste{\rgr\ctxmul\cctx_k}{\env'}$, so rule \refToRule{var} can be applied and $\NWreduce{\Conf{\x}{\env}}{\rgr}{\Conf{\bvPair{\cctx_k}{\val}}{\env'}}$.

\smallskip\noindent
For  function and unit values the result derives from \refItem{theo:no-waste-typing}{val}.

\smallskip
\item[$\ve=\PairExp{\sgr_1}{\ve_1}{\ve_2}{\sgr_2}$] From $\IsWFExp{\Delta}{\PairExp{\sgr_1}{\ve_1}{\ve_2}{\sgr_2}}{\Graded{\tau}{\rgr}}$ for some $\Delta$. 
By \refItem{lem:inversion}{pair} 
for some $\Delta_1$ and $\Delta_2$ and $\rgr'$  we get  $\rgr'\ctxmul(\Delta_1 \ctxsum \Delta_2)\ctxord\Delta$, $\tau = \PairT{\GradedInd{\tau}{1}{\rgr}}{\GradedInd{\tau}{2}{\sgr}}$
 and $\rgr \rord \rgr'$ and $\IsWFExp{\Delta_i}{\ve_i}{\GradedInd{\tau}{i}{\sgr}}$ for $i=1,2$.
From  $\rgr'\ctxmul(\Delta_1 \ctxsum \Delta_2)\ctxord\Delta$ and $\rgr \rord \rgr'$ we get $(\rgr\ctxmul\Delta_1)\ctxsum (\rgr\ctxmul\Delta_2)\ctxord\Delta$
and from $\IsWFEnv{\Gamma}{\env}{\Delta}$ and \cref{lem:ev-split} we have that
there are $\env_1$ and $\env_2$ such that 
$\env_1\envSum\env_2\ctxord\env$ and 
$\IsWFEnv{\Gamma'_1}{\env_1}{\rgr\rmul\Delta_1}$ and 
$\IsWFEnv{\Gamma'_2}{\env_2}{\rgr\rmul\Delta_2}$ and these are consistent.  
Therefore from  $\IsWFExp{\rgr\rmul\Delta_i}{\ve_i}{\GradedInd{\tau}{i}{\rgr\rmul\sgr}}$ and 
$\IsWFEnv{\Gamma'_i}{\env_i}{\rgr\rmul\Delta_i}$, applying rule \refToRule{T-conf} we get $\IsWFConf{\Gamma'_i}{\ve_i}{\env_i}{\Graded{\tau}{\GradedInd{\tau}{i}{\rgr\rmul\sgr}}}$.
By inductive hypotheses 
$ \red{\ve_i}{\env_i}{\rgr\rmul\sgr_i}{\bvPair{\cctx''_i}{\val_i}}{\env_i'}$ and therefore 
$\red{\PairExp{\sgr_1}{\ve_1}{\ve_2}{\sgr_2}}{\env}{\rgr}{\bvPair{\cctx''_1\ctxsum\cctx''_2}{\PairExp{\sgr_1}{\val_1}{\val_2}{\sgr_2}}}{(\env_1'\envSum\env_2')}$. 

 \end{description}
\end{proofOf}

\bigskip
\begin{lemma}[Inversion for possibly diverging expressions]\label[lemma]{lem:inversionE}\
\begin{enumerate}
\item \label{lem:inversion:return} If $\IsWFExp{\Gamma}{\Return{\ve}}{\Graded{\tau}{\rgr}}$ then $\IsWFExp{\Gamma'}{\ve}{\Graded{\tau}{\rgr'}}$ with $\rgr\rord\rgr'$, $\Gamma'\ctxord\Gamma$ and $\rgr'\neq\rzero$.
\item \label{lem:inversion:let} If $\IsWFExp{\Gamma}{\Let{\x}{\e_1}{\e_2}}{\Graded{\tau}{\rgr}}$ then $\IsWFExp{\Gamma_1}{\e_1}{\GradedInd{\tau}{1}{\rgr}}$ and $\IsWFExp{\Gamma_2,\VarGradeType{\x}{\rgr_1}{\tau_1}}{\e_2}{\Graded{\tau}{\rgr'}}$ and $\Gamma_1\ctxsum\Gamma_2 \ctxord\Gamma$ and $\rgr\rord\rgr'$.
\item \label{lem:inversion:app} If $\IsWFExp{\Gamma}{\App{\ve_1}{\ve_2}}{\Graded{\tau_2}{\rgr}}$ then we have $\Gamma_1\rsum\Gamma_2\ctxord\Gamma$ and $\rgr \rord \rgr'\rmul\rgr_2$ such that $\IsWFExp{\Gamma_1}{\ve_1}{\Graded{(\funType{\GradedInd{\tau}{1}{\rgr}}{\sgr}{\GradedInd{\tau}{2}{\rgr}})}{\rgr'\rsum\rgr'\rmul\sgr}}$ and $\IsWFExp{\Gamma_2}{\ve_2}{\Graded{\tau_1}{\rgr'\rmul\rgr_1}}$ and $\rgr'\neq\rzero$.

\item \label{lem:inversion:match-unit} If $\IsWFExp{\Gamma}{\MatchUnit{\ve}{\e}}{\Graded{\tau}{\rgr}}$ then we have $\Gamma_1\ctxsum\Gamma_2\ctxord\Gamma$ and $\rgr \rord \tgr$ such that $\IsWFExp{\Gamma_1}{\ve}{\Graded{\Unit}{\rgr'}}$ and $\IsWFExp{\Gamma_2}{\e}{\Graded{\tau}{\tgr}}$.
\item \label{lem:inversion:match-p} If $\IsWFExp{\Gamma}{\Match{\x}{\y}{\ve}{\e}}{\Graded{\tau}{\rgr}}$ then we have $\Gamma_1\ctxsum\Gamma_2\ctxord\Gamma$ and $\rgr\rord\tgr$ such that $\IsWFExp{\Gamma_1}{\ve}{\Graded{(\PairT{\Graded{\tau_1}{\rgr_1}}{\Graded{\tau_2}{\rgr_2}})}{\sgr}}$ and $\IsWFExp{\Gamma_2,\VarGradeType{\x}{\sgr\rmul\rgr_1}{\tau},\VarGradeType{\y}{\sgr\rmul\rgr_2}{\tau_2}}{\e}{\Graded{\tau}{\tgr}}$ and $\sgr\neq\rzero$.
\item \label{lem:inversion:match-in} If $\IsWFExp{\Gamma}{\Case{\ve}{\x}{\e_1}{\x}{\e_2}}{\Graded{\tau}{\rgr}}$ then we have $\Gamma_1\ctxsum\Gamma_2\ctxord\Gamma$ and $\rgr\rord\sgr$ such that $\IsWFExp{\Gamma_1}{\ve}{\Graded{(\SumT{\Graded{\tau_1}{\rgr_1}}{\Graded{\tau_2}{\rgr_2}})}{\rgr'}}$, $\IsWFExp{\Gamma_2,\VarGradeType{\x}{\rgr'\rmul\rgr_1}{\tau_1}}{\e_1}{\Graded{\tau}{\sgr}}$, $\IsWFExp{\Gamma_2,\VarGradeType{\x}{\rgr'\rmul\rgr_2}{\tau_2}}{\e_2}{\Graded{\tau}{\sgr}}$ and $\rgr'\neq\rzero$.
\end{enumerate}
\end{lemma}

\begin{lemma}[Environment extension]\label{lem:env-ext}
If $\IsWFEnv{\Gamma_1}{\env_1}{\Delta_1}$ and $\IsWFConf{\Gamma_2}{\bv}{\env_2}{\Graded\tau\rgr}$ and these are consistent,  then 
$\env = \env_1\envSum\env_2$ and $\Gamma = \Gamma_1\ctxsum\Gamma_2$, then  
$\IsWFEnv{\Gamma,\VarGradeType\x\rgr\tau}{\env,\EnvElem\x\rgr\bv}{\Delta_1,\VarGradeType\x\rgr\tau}$. 
\end{lemma}
\begin{proof}
By consistency, we know that 
$\env = \env_1\envSum\env_2$ and $\Gamma = \Gamma_1\ctxsum\Gamma_2$ are well defined. 
Let 
$\env = \EnvElem{\x_1}{\rgr_1}{\bvPair{\cctx_1}{\val_1}},\ldots,\EnvElem{\x_n}{\rgr_n}{\bvPair{\cctx_n}{\val_n}}$ and 
$\Gamma = \VarGradeType{\x_1}{\rgr_1}{\tau_1},\ldots,\VarGradeType{\x_n}{\rgr_n}{\tau_n}$ and suppose, without loss of generality, that 
there are $\leq h\leq k \leq n+1$ such that 
$\dom{\env_1}\cap\dom{\env_2} = \{\x_1,\ldots,\x_{h-1}\}$, $\dom{\env_1}\setminus\dom{\env_2} = \{\x_h,\ldots,\x_{k-1}\}$ and $\dom{\env_2}\setminus\dom{\env_1} = \{\x_k,\ldots,\x_n\}$. 
Hence, for all $i \in 1..h-1$, we have 
$\rgr_i = \rgr'_i\rsum\rgr''_i$ where 
$\Gamma_1(\x_i) = \Pair{\rgr'_i}{\tau_i}$ and $\Gamma_2(\x_i) = \Pair{\rgr''_i}{\tau_i}$ and similarly for the environments. 
By Rule \refToRule{t-conf}, we  know that 
$\IsWFEnv{\Gamma_2}{\env_2}{\Delta_2}$ and $\IsWFExp{\Delta_2}{\bv}{\Graded\tau\rgr}$.
By consistency and Rules \refToRule{t-env} and \refToRule{t-val},  we deduce that, 
for all $i \in 1..n$, $\IsWFExp{\Theta_i}{\val_i}{\Graded{\rgr_i}{\tau_i}}$ and 
the following inequalities hold: 
\begin{align*} 
\Delta_1 \ctxsum \sum_{i = 1}^{h-1} \rgr'_i\rmul\Theta_i  \ctxsum \sum_{i = h}^{k-1} \rgr_i\rmul\Theta_i &\ctxord \Gamma_1  \\ 
\Delta_2 \ctxsum \sum_{i = 1}^{h-1} \rgr''_i\rmul\Theta_i \ctxsum \sum_{i = k}^n \rgr_i\rmul\Theta_i     &\ctxord\Gamma_2 
\end{align*} 
Combining these inequalities, we derive 
\[ \Delta_1\ctxsum\Delta_2 \ctxsum \sum_{i = 1}^n \rgr_i\rmul\Theta_i \ctxord \Gamma_1\ctxsum\Gamma_2 = \Gamma \] 
Therefore, since $\x\notin\dom{\env}$, by Rule \refToRule{t-env} we immediately get the thesis. 
\end{proof}

\begin{proofOf}{theo:progress} We have $\IsWFConf{\Gamma}{\e}{\env}{\Graded\tau\rgr}$. By rule \refToRule{t-conf} we have $\IsWFEnv{\Gamma}{\env}{\Delta}$ and $\IsWFExp{\Delta}{\e}{\Graded\tau\rgr}$.
Proof is by case analysis on $\e$. We show only one case for the sake of brevity, since other ones are similar.

$\e=\Let{\x}{\e_1}{\e_2}$
By \cref{lem:inversionE}(\ref{lem:inversion:let}) we have $\IsWFExp{\Delta_1}{\e_1}{\GradedInd{\tau}{1}{\rgr}}$ and $\IsWFExp{\Delta_2,\VarGradeType{\x}{\rgr_1}{\tau_1}}{\e_2}{\Graded{\tau}{\rgr'}}$ and $\Delta_1\ctxsum\Delta_2 \ctxord\Delta$ and $\rgr\rord\rgr'$. By \cref{lem:ev-split} we have $\env_1$ and $\env_2$ such that 
$\env_1\envSum\env_2\ctxord\env$ and 
$\IsWFEnv{\Gamma_1}{\env_1}{\Delta_1}$ and 
$\IsWFEnv{\Gamma_2}{\env_2}{\Delta_2}$. By rule \refToRule{t-conf} we have $\IsWFConf{\Gamma_1}{\e_1}{\env_1}{\GradedInd\tau{1}\rgr}$.
We have two cases:
\begin{itemize}
\item $\not\exists$ $\Conf{\bv}{\env_1'}$ such that $\NWreduce{\Conf{\e_1}{\env_1}}{\rgr_1}{\Conf{\bv}{\env_1'}}$
By applying rule \refToRule{let-div1} we have the thesis.

\item $\exists$ $\Conf{\bv}{\env_1'}$ such that $\NWreduce{\Conf{\e_1}{\env_1}}{\rgr_1}{\Conf{\bv}{\env_1'}}$
 
Since we have $\NWreduce{\Conf{\e_1}{\env_1}}{\rgr_1}{\Conf{\bv}{\env_1'}}$ we also have $\NWreduce{\Conf{\e_1}{\envE_1}}{}{\Conf{\bv}{\envE'_1}}$ with $\erase{\env_1}=\envE_1$ and $\erase{\env_1'}=\envE_1'$. By \cref{theo:adjustE} we have $\IsWFConf{\Gamma_1'}{\bv}{\env_1'}{\GradedInd\tau{1}\rgr}$.
Since $\IsWFEnv{\Gamma_2}{\env_2}{\Delta_2}$ and $\IsWFConf{\Gamma_1'}{\bv}{\env_1'}{\GradedInd\tau{1}\rgr}$ and $\hat\env = \env_1\envSum\env_2$ and $\Gamma' = \Gamma_1'\ctxsum\Gamma_2$, then  
$\IsWFEnv{\Gamma',\VarGradeType\x{\rgr_1}{\tau_1}}{\hat\env,\EnvElem\x{\rgr_1}\bv}{\Delta_2,\VarGradeType\x{\rgr_1}{\tau_1}}$.  By renaming we have $\IsWFEnv{\Gamma',\VarGradeType{\x'}{\rgr_1}{\tau_1}}{\hat\env,\EnvElem{\x'}{\rgr_1}\bv}{\Delta_2,\VarGradeType{\x'}{\rgr_1}{\tau_1}}$. By \cref{lem:subst} we have $\IsWFExp{\Delta_2,\VarGradeType{\x'}{\rgr_1}{\tau_1}}{\Subst{\e_2}{\x'}{\x}}{\Graded{\tau}{\rgr'}}$. By \refToRule{t-conf} and \refToRule{t-sub} we have $\IsWFConf{\Gamma',\VarGradeType{\x'}{\rgr_1}{\tau_1}}{\Subst{\e_2}{\x'}{\x}}{\hat\env,\EnvElem{\x'}{\rgr_1}\bv}{\Graded\tau{\rgr}}$.

If $\not\exists$ $\Conf{\bv_2}{\env_3}$ such that $\NWreduce{\Conf{\Subst{\e_2}{\x'}{\x}}{\AddToEnv{\hat\env}{\x'}{\rgr_1}{\bv}}}{\rgr}{\Conf{\bv_2}{\env_3}}$ then by rule \refToRule{let/let-div2} we have the thesis, otherwise we have an absurd, since if $\Conf{\bv_2}{\env_3}$ would exists then, by rule \refToRule{let/let-div2} also $\NWreduce{\conf}{\rgr}{\Conf{\bv_2}{\env_3}}$.
\end{itemize}
\end{proofOf}


\begin{proofOf}{theo:adjustE} 
\newcommand{\mioBv}{{\val_1}^{\cctx_1}}
\newcommand{\mioBvv}{{\val}^{\cctx}}

By induction on the syntax of $\e$. We show the case $\e=\Let{\x}{\e_1}{\e_2}$. \\
From rule \refToRule{t-conf}, for some $\Delta$, we get $\IsWFExp{\Delta}{\e}{\Graded{\tau}{\rgr}}$ and $\IsWFEnv{\Gamma}{\env}{\Delta}$. 
From $\IsWFExp{\Delta}{\Let{\x}{\e_1}{\e_2}}{\Graded{\tau}{\rgr}}$ and  \cref{lem:inversionE}(\ref{lem:inversion:let}) we get 
$\IsWFExp{\Delta_1}{\e_1}{\GradedInd{\tau}{1}{\sgr}}$ and $\IsWFExp{\Delta_2,\VarGradeType{\x}{\sgr_1}{\tau_1}}{\e_2}{\Graded{\tau}{\rgr'}}$ and $\Delta_1\ctxsum\Delta_2 \ctxord\Delta$ and $\rgr\rord\rgr'$. 
From $\NWreduce{\Conf{\Let{\x}{\e_1}{\e_2}}{\envE}}{}{\Conf{\val}\envE'}$ we get
$\NWreduce{\Conf{\e_1}{\env_1}}{}{\Conf{\val'_1}{\envE_1'}}$ and $\NWreduce{\Conf{\Subst{\e_2}{\x'}{\x}}{(\envE_1'\envSum\envE_2),\x:\val'_1}}{}{\Conf{\val}{\envE'}}$ with $\envE'_1\envSum\envE'=\envE$ 
with $\envE=\erase{\env}$. \\
From \cref{lem:ev-split}, there are $\env_1$ and $\env_2$ such that 
$\env_1\envSum\env_2\ctxord\env$ and 
$\IsWFEnv{\Gamma_1}{\env_1}{\Delta_1}$ and 
$\IsWFEnv{\Gamma_2}{\env_2}{\Delta_2}$ and $\env_1$ and $\env_2$ are consistent. By rule \refToRule{t-conf} we have $\IsWFConf{\Gamma_1}{\e_1}{\env_1}{\GradedInd\tau{1}\rgr}$.
Applying  the inductive hypothesis to
$\e_1$, with  $\NWreduce{\Conf{\e_1}{\env_1}}{}{\Conf{\val'_1}{\envE_1'}}$, we get that there are  $\env'_1$, $\cctx_1$  and $\Gamma'_1$ such that
\begin{enumerate}
\item  \label{let1} $\NWreduce{\Conf{\e_1}{\env'_1}}{\sgr_1}{\Conf{\mioBv}{\env'_1}}$ with  $\erase{\env'_1}=\envE'_1$ and 
\item  \label{let2} $\IsWFConf{\Gamma'_1}{\mioBv}{\env'_1}{\GradedInd\tau{1}\sgr}$ and
\item \label{let3}  and $\env_1$ and $\env'_1$ are consistent.
\end{enumerate}
From $\IsWFEnv{\Gamma_2}{\env_2}{\Delta_2}$, (\ref{let2}), (\ref{let3})  and \cref{lem:env-ext} we get that, 
\begin{enumerate}\setcounter{enumi}{3}
\item \label{let4} $\IsWFEnv{(\Gamma'_1\ctxsum\Gamma_2),\VarGradeType{\x'}{\sgr_1}{\tau_1}}{(\env'_1\ctxsum\env_2),\EnvElem{\x'}{\sgr_1}{\bvPair{\cctx_1}{\val_1}   } } {\Delta_2,\VarGradeType{\x'}{\sgr_1}{\tau_1}}$
\end{enumerate}
Let $\e'_2=\Subst{\e_2}{\x'}{\x}$,
from $\IsWFExp{\Delta_2,\VarGradeType{\x}{\sgr_1}{\tau_1}}{\e_2}{\Graded{\tau}{\rgr'}}$ and \cref{lem:subst} we get 
$\IsWFExp{\Delta_2,\VarGradeType{\x'}{\sgr_1}{\tau_1}}{\Subst{\e_2}{\x'}{\x}}{\Graded{\tau}{\rgr'}}$ and
$\IsWFExp{\Delta_2,\VarGradeType{\x'}{\sgr_1}{\tau_1}}{\Subst{\e_2}{\x'}{\x}}{\Graded{\tau}{\rgr}}$ by rule \refToRule{T-sub}.
Therefore from (\ref{let4}) we get 
\begin{enumerate}\setcounter{enumi}{4}
\item \label{let5}
$\IsWFConf{(\Gamma'_1\ctxsum\Gamma_2),\VarGradeType{\x'}{\sgr_1}{\tau_1}}{\e'_2}{(\env'_1\ctxsum\env_2),\EnvElem{\x'}{\sgr_1}{\bvPair{\cctx_1}{\val_1}   } }{\Graded{\tau}{\rgr}}$
\end{enumerate}
Let $\env'_2=(\env'_1\ctxsum\env_2),\EnvElem{\x'}{\sgr_1}{\bvPair{\cctx_1}{\val_1}}$. Applying  the inductive hypothesis to $\e'_2$,  (\ref{let5}) with
$\NWreduce{\Conf{\Subst{\e'_2}{\x'}{\x}}{(\envE_1'\envSum\envE_2),\x:\val'_1}}{}{\Conf{\val}{\envE'}}$,
we have that there are  $\env'$,  $\cctx$  and $\Gamma'$ such that
\begin{enumerate}\setcounter{enumi}{5}
\item \label{let6}   $\NWreduce{\Conf{\e'_2}{\env'_2}}{\rgr}{\Conf{{\val}^{\cctx}}{\env'}}$ with  $\erase{\env'}=\envE'$ and
\item\label{let7}   $\IsWFConf{\Gamma'}{\val^{\cctx}}{\env'}{\Graded\tau\rgr}$ and 
\item\label{let8}   $\env'_2$ and $\env'$ are consistent.
\end{enumerate}
From (\ref{let1}) and (\ref{let6}) and  $\env_1\rsum\env_2\rord\env$ applying rule \refToRule{let} we get
$\NWreduce{\Conf{\Let{\x}{\e_1}{\e_2}}{\env}}{\rgr}{\Conf{{\val}^{\cctx}}{\env'}}$.
Moreover, from  $\env_1$ and  $\env_2$ consistent with $\env$, (\ref{let3}) and definition of $\env'_2$ we have that $\env$ is consistent with $\env'_2$.
Finally, (\ref{let7}) and (\ref{let8}) conclude the proof.

\end{proofOf}

\end{document}